\documentclass[showpacs,twocolumn,pra,superscriptaddress,notitlepage]{revtex4-2}
\usepackage{qcircuit}
\usepackage[dvips]{graphicx}
\usepackage{amsmath,amssymb,amsthm,mathrsfs,amsfonts,dsfont}
\usepackage{subfigure, epsfig}
\usepackage{braket}
\usepackage{bm}
\usepackage{enumerate}
\usepackage{algorithm}
\usepackage{algpseudocode}
\usepackage{diagbox}
\usepackage{physics}
\usepackage{color}
\usepackage{multirow}
\usepackage{comment}
\usepackage{tikz}
\usepackage[]{qcircuit}
\usetikzlibrary{arrows}
\usetikzlibrary{shapes,fadings,snakes}
\usetikzlibrary{decorations.pathmorphing,patterns}
\usetikzlibrary{calc}
\usetikzlibrary{positioning}
\usepackage[colorlinks = true]{hyperref}

\graphicspath{{./figure/}}

\newtheorem{theorem}{Theorem}

\newtheorem{fact}{Fact}

\newtheorem{proposition}{Proposition}

\newcommand{\mbb}{\mathbb}

\newcommand{\comments}[1]{}

\begin{document}
	\title{Detecting entanglement in quantum many-body systems via permutation moments}
	
	\begin{abstract}
		Multipartite entanglement plays an essential role in both quantum information science and many-body physics. Due to the exponentially large dimension and complex geometric structure of the state space, the detection of entanglement in many-body systems is extremely challenging in reality. Conventional means, like entanglement witness and entropy criterion, either highly depend on the prior knowledge of the studied systems or the detection capability is relatively weak. In this work, we propose a framework for designing multipartite entanglement criteria based on permutation moments, which have an effective implementation with either the generalized control-SWAP quantum circuits or the random unitary techniques. As an example, in the bipartite scenario, we develop an entanglement criterion that can detect bound entanglement and show strong detection capability in the multi-qubit Ising model with a long-range $XY$ Hamiltonian. In the multipartite case, the permutation-moment-based criteria can detect entangled states that are not detectable by any criteria extended from the bipartite case. Our framework also shows potential in entanglement quantification and entanglement structure detection.
	\end{abstract}
 \date{\today}

\author{Zhenhuan Liu}
\affiliation{Center for Quantum Information, Institute for Interdisciplinary Information Sciences, Tsinghua University, Beijing 100084, China}

\author{Yifan Tang}
\affiliation{Center for Quantum Information, Institute for Interdisciplinary Information Sciences, Tsinghua University, Beijing 100084, China}
\affiliation{Department of Mathematics and Computer Science, Freie Universität Berlin, 14195 Berlin, Germany}
\affiliation{Department of Physics, Freie Universität Berlin, 14195 Berlin, Germany}

\author{Hao Dai}
\affiliation{Center for Quantum Information, Institute for Interdisciplinary Information Sciences, Tsinghua University, Beijing 100084, China}

\author{Pengyu Liu}
\affiliation{Center for Quantum Information, Institute for Interdisciplinary Information Sciences, Tsinghua University, Beijing 100084, China}

\author{Shu Chen}
\affiliation{Center for Quantum Information, Institute for Interdisciplinary Information Sciences, Tsinghua University, Beijing 100084, China}

\author{Xiongfeng Ma}
\email{xma@tsinghua.edu.cn}
\affiliation{Center for Quantum Information, Institute for Interdisciplinary Information Sciences, Tsinghua University, Beijing 100084, China}

\maketitle

The past decades have witnessed great progress in understanding quantum entanglement \cite{Horodecki2009entanglement}. To date, entanglement acts not only as the cornerstone of quantum information science, but also as a new perspective in many other fields, like quantum thermodynamics \cite{ueda2020quantum,adam2016thermal}, condensed matter physics \cite{LAFLORENCIE2016quantum}, and quantum gravity \cite{nishioka2018entanglement}. Especially in many-body physics, the dynamical behavior, scaling property, and spectral form of entanglement are key indicators to characterize different phases of the system \cite{abanin2019manybody,amico2008entanglement}.

As a resource that cannot be produced by local operations and classical communication (LOCC), entangled $k$-partite states are those that cannot be written in a separable form, $\rho=\sum_ip_i\rho_{1}^i\otimes \cdots \otimes\rho_{k}^i$, where $p_i\ge 0$ satisfying the normalization condition $\sum_ip_i=1$ and $\rho_{r}^i$ is the density matrix of the $r$-th subsystem. As the dimension of a quantum system grows exponentially with the number of qubits, the geometric structure of state space becomes highly complicated, making entanglement detection a resource-consuming task. In fact, determining whether a state is entangled or not is generally a NP-hard problem \cite{gurvits2003complexity}. For pure states or states with enough prior knowledge, entanglement can be effectively detected by purity measurements \cite{adam2016thermal,brydges2019probing}, variational algorithms \cite{kokail2021variational}, or entanglement witness \cite{gunhe2009entanglement,zhang2021efficient}. While in the noisy intermediate-scale quantum era \cite{Preskill2018quantumcomputingin}, the processed states are usually disturbed by unpredictable noise, rendering the detection capability of conventional means ineffective. Consequently, it is important to find implementable and efficient methods to detect multipartite entanglement with state-of-the-art devices.

For a generic mixed multipartite state without prior knowledge, there are two commonly-used techniques to detect entanglement, density matrix moments and index permutation. Moments of density matrix, $\tr(\rho^n)$, carry much information about the states and are relatively easy to measure \cite{ekert2002direct,van2012measuring,elben2019toolbox}.  Hence, they become practical tools in estimating properties of quantum systems \cite{smith2017quantifying}, including quantum entanglement \cite{horodecki2003measuring,gunhe2009entanglement,imai2021bound,beckey2021computable,ketter2022stat}. However, most moment-based entanglement criteria are specially designed for states with few parties or low dimensions. 
A general moment-based entanglement-detection framework for multipartite systems is still missing. 

Based on the rearrangement of density matrix elements, the index permutation criterion \cite{horodecki2006separability} can be applied in systems with an arbitrary number of parties and dimensions and has many generalizations \cite{zhang2008entanglement}. In general, a $k$-partite quantum state can be represented using a matrix with $2k$ indices,
\begin{equation}
\begin{split}
\rho=\sum_{s_1,\cdots,s_{2k}}\rho_{s_1s_2,\cdots , s_{2k-1}s_{2k}}\ketbra{s_1\cdots s_{2k-1}}{s_{2}\cdots s_{2k}},
\end{split}
\end{equation}
where $s_1,s_3,\dots,s_{2k-1}$ represent the row indices, and $s_2,s_4,\dots,s_{2k}$ represent the column ones. The two indices, $s_{2r-1}$ and $s_{2r}$, denote for the $r$-th subsystem. By changing the order of these $2k$ indices, one gets a new matrix, $\mathcal{R}_{\pi}(\rho)$, with
\begin{equation}\label{eq:defRpi}
\begin{split}
\left[\mathcal{R}_{\pi}(\rho)\right]_{s_1s_2,\cdots, s_{2k-1}s_{2k}} = \rho_{s_{\pi(1)}s_{\pi(2)},\cdots , s_{\pi(2k-1)}s_{\pi(2k)}},
\end{split}
\end{equation}
where $\pi$ is an element of $2k$-th permutation group $\mathcal{S}_{2k}$. For simplicity, hereafter, we use $\mathcal{R}_\pi$ to denote $\mathcal{R}_\pi(\rho)$. Using the property of index permutation, one could prove that \cite{horodecki2006separability}
\begin{equation}\label{eq:permutationcriterion}
\begin{split}
\norm{\mathcal{R}_\pi}=\tr\left(\sqrt{\mathcal{R}_\pi\mathcal{R}_\pi^\dagger}\right)=\sum_i\lambda_i \le 1,
\end{split}
\end{equation}
for all $k$-partite separable states, where $\{\lambda_i\}$ are the singular values of $\mathcal{R}_\pi$. The violation of this inequality indicates entanglement.  In the bipartite scenario,  when setting $\pi$ to be $(1,2)$ and $(2,3)$, where $(\cdot,\cdot)$ denotes exchanging two indices, one gets the widely-used positive partial transposition (PPT) criterion \cite{peres1996separability} and the computable cross norm (CCNR) criterion \cite{chen2002matrix}, respectively. However, because index permutation is an unphysical operation, and the permutation criteria are based on singular value decomposition, a highly nonlinear operation, measurement of $\norm{\mathcal{R}_\pi}$ usually requires full state tomography, which is extremely resource-consuming \cite{haah2017sample}.

To harness the power of permutation criteria in multipartite entanglement detection, we borrow the idea from moment criteria. Although it is generally hard to measure $\norm{\mathcal{R}_\pi}$ directly, one can alternatively estimate the higher-order moments, $M_{2n}^\pi=\tr[(\mathcal{R}_\pi\mathcal{R}_\pi^\dagger)^n]=\sum_i\lambda_i^{2n}$, which are much easier to access. These permutation moments can help to lower bound $\norm{\mathcal{R}_\pi}=\sum_i\lambda_i$ and infer whether the state is multipartite entangled or not. A similar idea has also been used in the estimation of quantum negativity \cite{gray2018machine,yu2021optimal} and entropy \cite{smith2017quantifying}. By changing the index permutation operation $\mathcal{R}_\pi(\cdot)$ and measuring different orders of moments, we generate a series of implementable multipartite entanglement criteria, which we call \emph{moment-based permutation criterion}. The entanglement detection flowchart is shown in Fig.~\ref{fig:overview}.

\begin{figure}[htbp]
    \centering
    \includegraphics[scale=0.24]{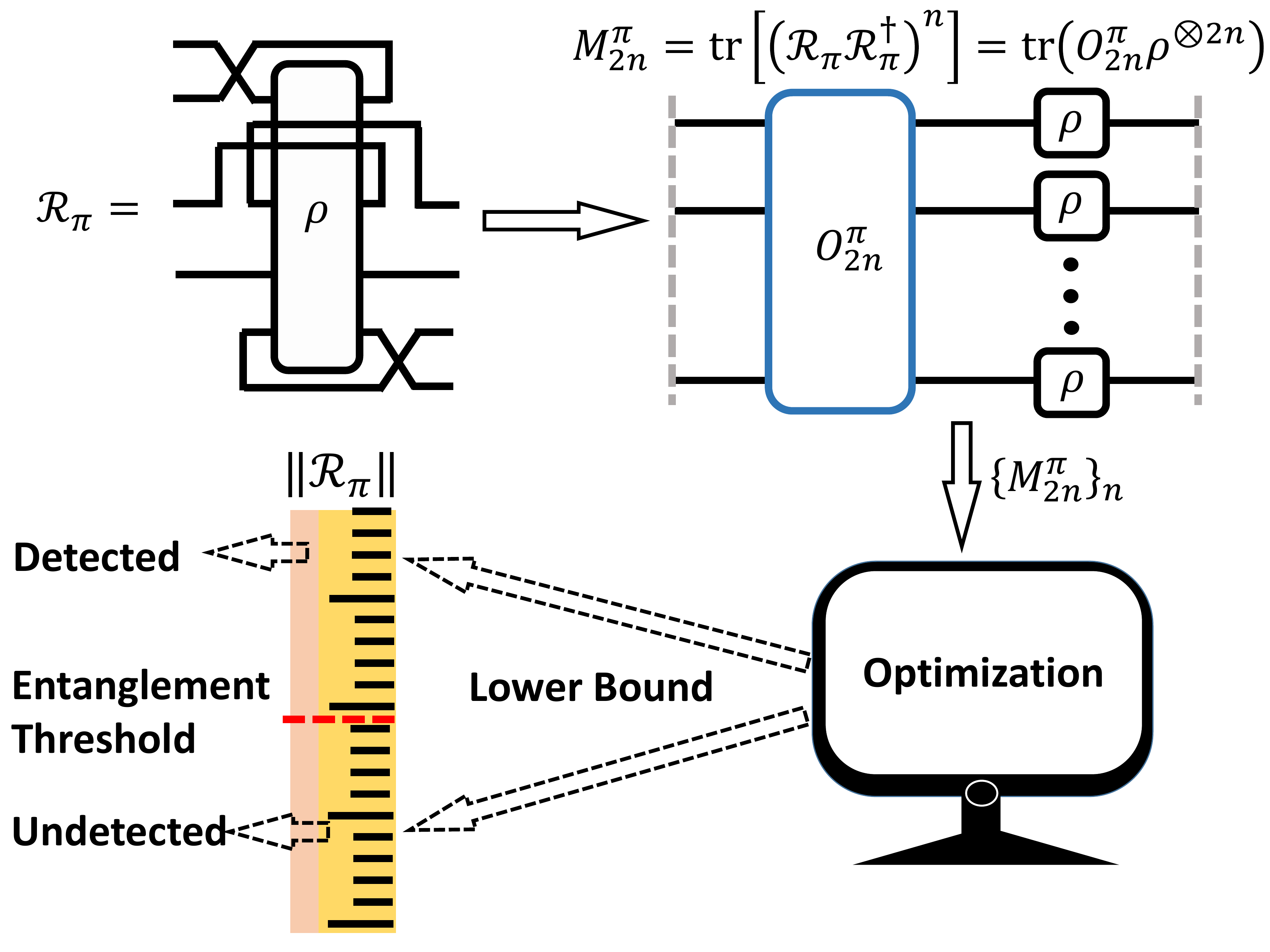}
    \caption{Flowchart of entanglement detection. To detect multipartite entanglement of $\rho$, one needs to first choose an index permutation operation $\mathcal{R}_\pi(\cdot)$ and set $\norm{\mathcal{R}_\pi(\rho)}=\sum_i\lambda_i$ as the entanglement indicator. Then, one measures the permutation moments $\{M_{2n}^\pi=\sum_i\lambda_i^{2n}\}_n$ and use these moments to lower bound $\norm{\mathcal{R}_\pi(\rho)}$. If the lower bound is larger than the entanglement threshold set for $\norm{\mathcal{R}_\pi(\rho)}$, the multipartite entanglement is successfully detected. Otherwise, one can measure higher-order moments or pick another index permutation and repeat the procedure. }
    \label{fig:overview}
\end{figure}

\emph{Moment-based permutation criteria.} ---  For multipartite quantum state $\rho$, each party has two indices, one for row and one for column. Given $\mathcal{R}_\pi$, all the parties can be divided into four types, based on the position transition of their two indices, as shown in the first line of Table \ref{tab:mainres}.

We find two properties of index permutation. First, if a party is T1-type or T2-type for $\mathcal{R}_\pi$, then it will keep the type for $\mathcal{R}_\pi^\dagger$; while if it is R1-type or R2-type for $\mathcal{R}_\pi$, then it becomes R2-type or R1-type for $\mathcal{R}_\pi^\dagger$, respectively. Second, the indices contraction in $M_{2n}^\pi=\tr[(\mathcal{R}_\pi\mathcal{R}_\pi^\dagger)^n]$ only acts on the indices from the same party of the $2n$ copies of $\rho$. Hence, if we list these $2n$ copies of states in order, the indices from an R-type party will contract with one of its two neighboring states and the indices from a T-type party will contract with both of its neighboring states. Note that $\mathcal{R}_\pi$ might not be a square matrix, so its odd moments are generally inaccessible. 

\begin{table}
\begin{tabular}{ccccc}
\hline
Type & T1 & T2 & R1 & R2 \\
\hline
$\mathcal{R}_{\pi}$ &
\begin{tikzpicture}[
	scale=1.5,
	baseline, 
	Lline/.style={-,blue,very thick,rounded corners},
	Rline/.style={-,red,very thick,rounded corners},
	]
	\node[draw=black,rectangle,rounded corners=0.1cm,minimum size=.42cm,,text=black,very thick] (rho) at (0,0) {$\rho$};
	\draw[Lline] (rho.west) --++ (-.2,0);
	\draw[Rline] (rho.east) --++ (.2,0);
\end{tikzpicture} 
& 
\begin{tikzpicture}[
	scale=1.5,
	baseline, 
	Lline/.style={-,blue,very thick,rounded corners},
	Rline/.style={-,red,very thick,rounded corners},
	]
	\node[draw=black,rectangle,rounded corners=0.1cm,minimum size=.42cm,,text=black,very thick] (rho) at (0,0) {$\rho$};
	\draw[Lline] (rho.west) --++ (-.15,0) to [bend left=80] ($(rho.east)+(.225,0)$) --++ (.15,0);
	\draw[Rline] (rho.east) --++ (.15,0) to [bend left=80] ($(rho.west)+(-.225,0)$) --++ (-.15,0);
\end{tikzpicture} 
& 
\begin{tikzpicture}[
	scale=1.5,
	baseline, 
	Lline/.style={-,blue,very thick,rounded corners},
	Rline/.style={-,red,very thick,rounded corners},
	]
	\node[draw=black,rectangle,rounded corners=0.1cm,minimum size=.42cm,,text=black,very thick] (rho) at (0,0) {$\rho$};
	\draw[Lline] (rho.west) --++ (-.3,0);
	\draw[Rline] (rho.east) --++ (.15,0) to [bend left=50] ($(rho.west)+(-.15,-0.15)$) --++ (-.15,0);
\end{tikzpicture}
&
\begin{tikzpicture}[
	scale=1.5,
	baseline, 
	Lline/.style={-,blue,very thick,rounded corners},
	Rline/.style={-,red,very thick,rounded corners},
	]
	\node[draw=black,rectangle,rounded corners=0.1cm,minimum size=.42cm,,text=black,very thick] (rho) at (0,0) {$\rho$};
	\draw[Lline] (rho.west) --++ (-.15,0) to [bend left=50] ($(rho.east)+(.15,0.15)$) --++ (.15,0);;
	\draw[Rline] (rho.east) --++ (.3,0);
\end{tikzpicture}
\\
$\tr(\mathcal{R}_{\pi}\mathcal{R}_{\pi}^{\dagger})$ &
\begin{tikzpicture}[
	scale=1.5,
	baseline, 
	Lline/.style={-,blue,very thick,rounded corners},
	Rline/.style={-,red,very thick,rounded corners},
	Bline/.style={dashed,gray,thick},
	]
	\node[draw=black,rectangle,rounded corners=0.1cm,minimum size=.42cm,,text=black,very thick] (rho1) at (0,.175) {$\rho$};
	\draw[Lline] (rho1.west) --++ (-.1,0) --++ (-.2,-.35) --++ (-0.1,0) node (L1) {};
	\draw[Rline] (rho1.east) --++ (.1,0) node (R1) {};
	\node[draw=black,rectangle,rounded corners=0.1cm,minimum size=.42cm,,text=black,very thick] (rho2) at ($(rho1)+(0,-.35)$) {$\rho$};
	\draw[Lline] (rho2.west) --++ (-.1,0) --++ (-.2,.35) --++ (-0.1,0) node (L2) {};
	\draw[Rline] (rho2.east) --++ (.1,0) node (R2) {};
	\draw[Bline] (L1.south) -- (L2.north);
	\draw[Bline] (R1.north) -- (R2.south);
\end{tikzpicture} 
& 
\begin{tikzpicture}[
	scale=1.5,
	baseline, 
	Lline/.style={-,blue,very thick,rounded corners},
	Rline/.style={-,red,very thick,rounded corners},
	Bline/.style={dashed,gray,thick},
	]
	\node[draw=black,rectangle,rounded corners=0.1cm,minimum size=.42cm,,text=black,very thick] (rho1) at (0,.175) {$\rho$};
	\draw[Lline] (rho1.west) --++ (-.1,0) --++ (-.2,-.35) --++ (-0.1,0) node (L1) {};
	\draw[Rline] (rho1.east) --++ (.1,0) node (R1) {};
	\node[draw=black,rectangle,rounded corners=0.1cm,minimum size=.42cm,,text=black,very thick] (rho2) at ($(rho1)+(0,-.35)$) {$\rho$};
	\draw[Lline] (rho2.west) --++ (-.1,0) --++ (-.2,.35) --++ (-0.1,0) node (L2) {};
	\draw[Rline] (rho2.east) --++ (.1,0) node (R2) {};
	\draw[Bline] (L1.south) -- (L2.north);
	\draw[Bline] (R1.north) -- (R2.south);
\end{tikzpicture}
&
\begin{tikzpicture}[
	scale=1.5,
	baseline, 
	Lline/.style={-,blue,very thick,rounded corners},
	Rline/.style={-,red,very thick,rounded corners},
	Bline/.style={dashed,gray,thick},
	]
	\node[draw=black,rectangle,rounded corners=0.1cm,minimum size=.42cm,,text=black,very thick] (rho1) at (0,.175) {$\rho$};
	\draw[Lline] (rho1.west) --++ (-.1,0) --++ (-.2,-.35) --++ (-0.1,0) node (L1) {};
	\draw[Rline] (rho1.east) --++ (.1,0) node (R1) {};
	\node[draw=black,rectangle,rounded corners=0.1cm,minimum size=.42cm,,text=black,very thick] (rho2) at ($(rho1)+(0,-.35)$) {$\rho$};
	\draw[Lline] (rho2.west) --++ (-.1,0) --++ (-.2,.35) --++ (-0.1,0) node (L2) {};
	\draw[Rline] (rho2.east) --++ (.1,0) node (R2) {};
	\draw[Bline] (L1.south) -- (L2.north);
	\draw[Bline] (R1.north) -- (R2.south);
\end{tikzpicture}
&
\begin{tikzpicture}[
	scale=1.5,
	baseline, 
	Lline/.style={-,blue,very thick,rounded corners},
	Rline/.style={-,red,very thick,rounded corners},
	Bline/.style={dashed,gray,thick},
	]
	\node[draw=black,rectangle,rounded corners=0.1cm,minimum size=.42cm,,text=black,very thick] (rho1) at (0,.175) {$\rho$};
	\draw[Lline] (rho1.west) --++ (-.1,0) --++ (-.2,-.35) --++ (-0.1,0) node (L1) {};
	\draw[Rline] (rho1.east) --++ (.1,0) node (R1) {};
	\node[draw=black,rectangle,rounded corners=0.1cm,minimum size=.42cm,,text=black,very thick] (rho2) at ($(rho1)+(0,-.35)$) {$\rho$};
	\draw[Lline] (rho2.west) --++ (-.1,0) --++ (-.2,.35) --++ (-0.1,0) node (L2) {};
	\draw[Rline] (rho2.east) --++ (.1,0) node (R2) {};
	\draw[Bline] (L1.south) -- (L2.north);
	\draw[Bline] (R1.north) -- (R2.south);
\end{tikzpicture}
\vspace{0.2cm} \\
$\tr[(\mathcal{R}_{\pi}\mathcal{R}_{\pi}^{\dagger})^2]$ &
\begin{tikzpicture}[
	scale=1.5,
	baseline, 
	Lline/.style={-,blue,very thick,rounded corners},
	Rline/.style={-,red,very thick,rounded corners},
	Bline/.style={dashed,gray,thick},
	]
	\node[draw=black,rectangle,rounded corners=0.1cm,minimum size=.42cm,,text=black,very thick] (rho1) at (0,.525) {$\rho$};
	\draw[Lline] (rho1.west) --++ (-.1,0) --++ (-.2,-.35) --++ (-0.1,0);
	\draw[Rline] (rho1.east) --++ (.1,0) node (R1) {};
	\node[draw=black,rectangle,rounded corners=0.1cm,minimum size=.42cm,,text=black,very thick] (rho2) at ($(rho1)+(0,-.35)$) {$\rho$};
	\draw[Lline] (rho2.west) --++ (-.1,0) --++ (-.2,-.35) --++ (-0.1,0);
	\draw[Rline] (rho2.east) --++ (.1,0);
	\node[draw=black,rectangle,rounded corners=0.1cm,minimum size=.42cm,,text=black,very thick] (rho3) at ($(rho2)+(0,-.35)$) {$\rho$};
	\draw[Lline] (rho3.west) --++ (-.1,0) --++ (-.2,-.35) --++ (-0.1,0) node (L1) {};
	\draw[Rline] (rho3.east) --++ (.1,0);
	\node[draw=black,rectangle,rounded corners=0.1cm,minimum size=.42cm,,text=black,very thick] (rho4) at ($(rho3)+(0,-.35)$) {$\rho$};
	\draw[Lline] (rho4.west) --++ (-.1,0) --++ (-.2,1.05) --++ (-0.1,0) node (L2) {};
	\draw[Rline] (rho4.east) --++ (.1,0) node (R2) {};
	\draw[Bline] (L1.south) -- (L2.north);
	\draw[Bline] (R1.north) -- (R2.south);
\end{tikzpicture}
&
\begin{tikzpicture}[
	scale=1.5,
	baseline, 
	Lline/.style={-,blue,very thick,rounded corners},
	Rline/.style={-,red,very thick,rounded corners},
	Bline/.style={dashed,gray,thick},
	]
	\node[draw=black,rectangle,rounded corners=0.1cm,minimum size=.42cm,,text=black,very thick] (rho1) at (0,.525) {$\rho$};
	\draw[Lline] (rho1.west) --++ (-.1,0) --++ (-.2,-1.05) --++ (-0.1,0) node (L1) {};
	\draw[Rline] (rho1.east) --++ (.1,0) node (R1) {};
	\node[draw=black,rectangle,rounded corners=0.1cm,minimum size=.42cm,,text=black,very thick] (rho2) at ($(rho1)+(0,-.35)$) {$\rho$};
	\draw[Lline] (rho2.west) --++ (-.1,0) --++ (-.2,.35) --++ (-0.1,0) node (L2) {};
	\draw[Rline] (rho2.east) --++ (.1,0);
	\node[draw=black,rectangle,rounded corners=0.1cm,minimum size=.42cm,,text=black,very thick] (rho3) at ($(rho2)+(0,-.35)$) {$\rho$};
	\draw[Lline] (rho3.west) --++ (-.1,0) --++ (-.2,.35) --++ (-0.1,0);
	\draw[Rline] (rho3.east) --++ (.1,0);
	\node[draw=black,rectangle,rounded corners=0.1cm,minimum size=.42cm,,text=black,very thick] (rho4) at ($(rho3)+(0,-.35)$) {$\rho$};
	\draw[Lline] (rho4.west) --++ (-.1,0) --++ (-.2,.35) --++ (-0.1,0);
	\draw[Rline] (rho4.east) --++ (.1,0) node (R2) {};
	\draw[Bline] (L1.south) -- (L2.north);
	\draw[Bline] (R1.north) -- (R2.south);
\end{tikzpicture}
&
\begin{tikzpicture}[
	scale=1.5,
	baseline, 
	Lline/.style={-,blue,very thick,rounded corners},
	Rline/.style={-,red,very thick,rounded corners},
	Bline/.style={dashed,gray,thick},
	]
	\node[draw=black,rectangle,rounded corners=0.1cm,minimum size=.42cm,,text=black,very thick] (rho1) at (0,.525) {$\rho$};
	\draw[Lline] (rho1.west) --++ (-.1,0) --++ (-.2,-1.05) --++ (-0.1,0) node (L1) {};
	\draw[Rline] (rho1.east) --++ (.1,0) node (R1) {};
	\node[draw=black,rectangle,rounded corners=0.1cm,minimum size=.42cm,,text=black,very thick] (rho2) at ($(rho1)+(0,-.35)$) {$\rho$};
	\draw[Lline] (rho2.west) --++ (-.1,0) --++ (-.2,-.35) --++ (-0.1,0);
	\draw[Rline] (rho2.east) --++ (.1,0);
	\node[draw=black,rectangle,rounded corners=0.1cm,minimum size=.42cm,,text=black,very thick] (rho3) at ($(rho2)+(0,-.35)$) {$\rho$};
	\draw[Lline] (rho3.west) --++ (-.1,0) --++ (-.2,.35) --++ (-0.1,0);
	\draw[Rline] (rho3.east) --++ (.1,0);
	\node[draw=black,rectangle,rounded corners=0.1cm,minimum size=.42cm,,text=black,very thick] (rho4) at ($(rho3)+(0,-.35)$) {$\rho$};
	\draw[Lline] (rho4.west) --++ (-.1,0) --++ (-.2,1.05) --++ (-0.1,0) node (L2) {};
	\draw[Rline] (rho4.east) --++ (.1,0) node (R2) {};
	\draw[Bline] (L1.south) -- (L2.north);
	\draw[Bline] (R1.north) -- (R2.south);
\end{tikzpicture}
&
\begin{tikzpicture}[
	scale=1.5,
	baseline, 
	Lline/.style={-,blue,very thick,rounded corners},
	Rline/.style={-,red,very thick,rounded corners},
	Bline/.style={dashed,gray,thick},
	]
	\node[draw=black,rectangle,rounded corners=0.1cm,minimum size=.42cm,,text=black,very thick] (rho1) at (0,.525) {$\rho$};
	\draw[Lline] (rho1.west) --++ (-.1,0) --++ (-.2,-.35) --++ (-0.1,0);
	\draw[Rline] (rho1.east) --++ (.1,0) node (R1) {};
	\node[draw=black,rectangle,rounded corners=0.1cm,minimum size=.42cm,,text=black,very thick] (rho2) at ($(rho1)+(0,-.35)$) {$\rho$};
	\draw[Lline] (rho2.west) --++ (-.1,0) --++ (-.2,.35) --++ (-0.1,0) node (L1) {};
	\draw[Rline] (rho2.east) --++ (.1,0);
	\node[draw=black,rectangle,rounded corners=0.1cm,minimum size=.42cm,,text=black,very thick] (rho3) at ($(rho2)+(0,-.35)$) {$\rho$};
	\draw[Lline] (rho3.west) --++ (-.1,0) --++ (-.2,-.35) --++ (-0.1,0) node (L2) {};
	\draw[Rline] (rho3.east) --++ (.1,0);
	\node[draw=black,rectangle,rounded corners=0.1cm,minimum size=.42cm,,text=black,very thick] (rho4) at ($(rho3)+(0,-.35)$) {$\rho$};
	\draw[Lline] (rho4.west) --++ (-.1,0) --++ (-.2,.35) --++ (-0.1,0);
	\draw[Rline] (rho4.east) --++ (.1,0) node (R2) {};
	\draw[Bline] (L1.south) -- (L2.north);
	\draw[Bline] (R1.north) -- (R2.south);
\end{tikzpicture}
\\
\hline
\end{tabular}
\centering \caption{We use the tensor network to illustrate the four kinds of parties and their second and fourth moments. The boxes represent the subsystems of a generic $k$-partite state, and the two legs represent the row and column indices. The grey dashed lines represent periodic boundary condition. From the second line, one can find that the operators to estimate the second moments for all these four kinds of parties are the SWAP operators, which are represented by changing the order of two legs. Hence, $\tr(\mathcal{R}_\pi\mathcal{R}_\pi^\dagger)=\tr(\rho^2)$ for all $\pi\in\mathcal{S}_{2k}$. From the third line, one can find that the operators to estimate $\tr[(\mathcal{R}_\pi\mathcal{R}_\pi^\dagger)^2]$ for R-type parties are still SWAP operators, while the operators for T-type parties are cyclic permutation operators, which are represented by changing the order of the four legs cyclically. }
\label{tab:mainres}
\end{table}

Based on these results, we can prove that $M_{2n}^\pi$ can be measured directly by a joint observable.

\begin{theorem}\label{theorem:main res}
Given a $k$-partite state $\rho$ and the index permutation operation $\mathcal{R}_\pi$, the $2n$-th moment of $\mathcal{R}_{\pi}$, $M_{2n}^\pi:=\tr[(\mathcal{R}_{\pi}\mathcal{R}_{\pi}^\dagger)^n]$, can be estimated by observable measurement on 2n copies of $\rho$,
\begin{equation}\label{eq:main res}
\begin{split}
M_{2n}^\pi=\tr\left(O^\pi_{2n}\rho^{\otimes 2n}\right)=\frac{1}{2}\tr\left[\left(\bigotimes_{i=1}^k U^\pi_{i}+h.c.\right) \rho^{\otimes 2n}\right].
\end{split}
\end{equation}
For T1-type parties $U_{i}^\pi=\overrightarrow{\Pi}_i$ and for T2-type parties $U_{i}^\pi=\overleftarrow{\Pi}_i$.
Here $\overrightarrow{\Pi}$ and $\overleftarrow{\Pi}$ are the cyclic permutation operators in different directions, satisfying $\overrightarrow{\Pi}\ket{s_1,\cdots,s_{2n}}=\ket{s_{2n},s_1,\cdots,s_{2n-1}}$ and $\overleftarrow{\Pi}\ket{s_1,\cdots,s_{2n}}=\ket{s_2,\cdots,s_{2n},s_1}$. For R1-type parties $
U_{i}^\pi=\mbb{S}^{(2n,1)}_i\otimes \mbb{S}^{(2,3)}_i\otimes \cdots \otimes \mbb{S}^{(2n-2,2n-1)}_i$ and for R2-type parties $U_{i}^\pi=\mbb{S}^{(1,2)}_i\otimes \mbb{S}^{(3,4)}_i\otimes \cdots \otimes \mbb{S}^{(2n-1,2n)}_i$, where $\mbb{S}^{(u,v)}$ is the SWAP operator acting on the $u$-th and $v$-th copies.
\end{theorem}
We leave the proofs of theorems in Appendix \ref{subsec:proofofobservable}. The special cases of partial transposed and realigned moments for a two-qubit system have been discussed in Ref.~\cite{carteret2005noiseless,cai2008novel}.

Borrowing the ideas from \cite{ekert2002direct,knill1998power,cai2008novel}, by introducing an ancilla qubit, we can design a quantum circuit to measure $M_{2n}^\pi$ based on the control-unitary operations, see Fig.~\ref{fig:SWAP_test}. As the SWAP operators are the generators of the permutation group, all the control-unitary operators in this circuit can be decomposed into a polynomial number of the 3-qubit control-SWAP operators.
\begin{figure}[htbp]
    \centering
    \includegraphics[scale=0.17]{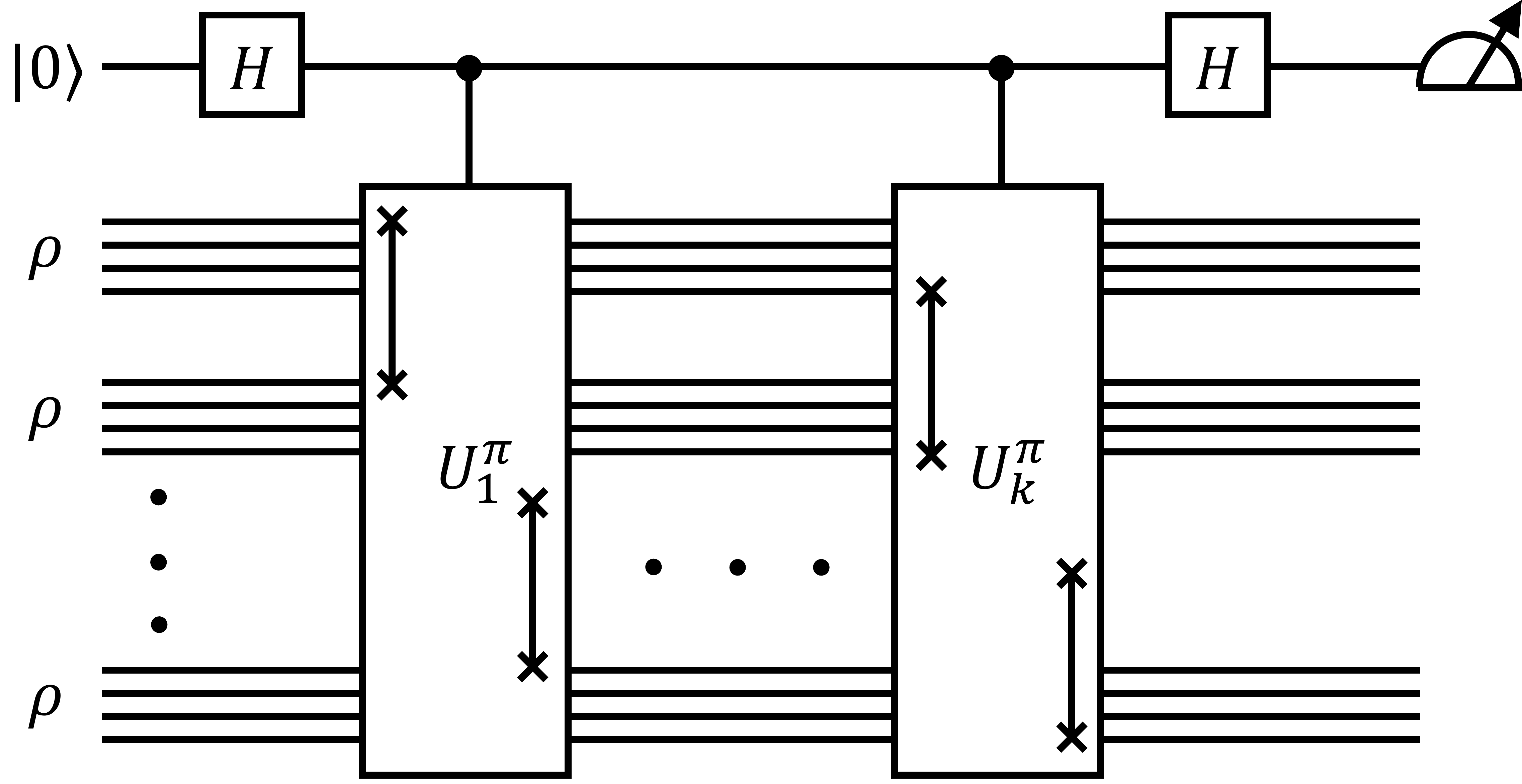}
    \caption{Quantum circuit for measuring $M_{2n}^\pi$. Inputs of this algorithm are $2n$ identical copies of $\rho$ and an ancilla qubit. The quantum gates in this circuit include Hadamard gates, labeled with $H$, and the control-$U_i^\pi$ gate. \comments{As $U_i^\pi$ can be decomposed to a series of SWAP operators, denoted by the connected $X$-form signs, this circuit can be realized by 3-qubit quantum gates.} The measurement of the ancilla qubit is in computational basis. }
    \label{fig:SWAP_test}
\end{figure}

However, the simultaneous preparation of $2n$ identical copies of $\rho$ is greatly challenging for state-of-the-art quantum devices. Fortunately, the recently developed techniques, shadow estimation \cite{aaronson2019shadow,huang2020predicting} and randomized measurements \cite{van2012measuring,brydges2019probing,elben2019toolbox}, provide means to measure these moments by local (single-qubit) or global (multi-qubit) single-copy operations. Practically speaking, global operations are still challenging, and local protocols are the ones commonly used in real experiments. Shadow estimation has a wide range of applications while inefficient in general. Randomized measurements has lower sample complexities, while it can only measure some specialized physical quantities. In Appendix \ref{sec:protocols}, we propose the measurement of the permutation moments using these two protocols and a hybrid one and analyze their sample complexities.

In the bipartite scenario, $M_n^{(1,2)}=\tr[(\mathcal{R}_{(1,2)})^n]=\tr[(\rho_{AB}^{T_A})^n]$ and $M_{2n}^{(2,3)}$ are key quantities that help to construct the weak-form PPT criteria \cite{elben2020mixed,yu2021optimal,neven2021symmetry} and the criteria proposed later in Eq.~\eqref{eq:criterionformulation} and Eq.~\eqref{eq:enhanced CCNR}, respectively. The local randomized measurements scheme is not applicable for the measurement of $M_n^{(1,2)}$ \cite{zhou2020Single}. The existing single-copy local protocol for measuring $M_n^{(1,2)}$ requires the shadow scheme. While, according to Theorem \ref{theorem:main res}, the observables for measuring $M_{2n}^{(2,3)}$ have a simple form. Thus, $M_{2n}^{(2,3)}$ can be measured through the local randomized measurements protocol and have a much lower sample complexity. We list the sample complexities of measuring $M_3^{(1,2)}$ and $M_{4}^{(2,3)}$ in Table \ref{tab:variance}.

\begin{table}[htbp]
    \centering
    \begin{tabular}{c|c|c}
    \hline
         & Global Protocol & Local Protocol \\
    \hline 
        $M_3^{(1,2)} $ & $O(D^{\frac{2}{3}})$\cite{zhou2020Single} & $O(D^2)$\cite{elben2020mixed} \\
    
        $M_4^{(2,3)}$ & $O(D^{\frac{1}{2}})$ & $O(D^{1.187})$\\
    \hline
    \end{tabular}
    \caption{
    This table shows the best-known sample complexities of measuring $M_3^{(1,2)}$ and  the complexities of protocols developed in Appendix \ref{subsec:randomprotocol} to measure $M_4^{(2,3)}$. $D$ is the dimension of the underlying Hilbert space. This table shows a nearly quadratic improvement in the local case.
    }
    
    \label{tab:variance}
\end{table}

To find the lower bound of $\norm{\mathcal{R}_\pi}=\sum_i\lambda_i$ using these moments, one needs to solve an optimization problem. The original optimization problem is extremely hard to solve because we have an exponentially large number of $\lambda_i$.  Adopting the Lagrange multiplier method, we can simplify the optimization to a problem of solving a set of polynomial equations.

\begin{theorem}\label{thm:LagrangeMultipler}
The minimum value of $\norm{\mathcal{R}_\pi}$ given $M_2^\pi$, ..., $M_{2n}^\pi$ is reached when there are at most $n$ non-zero $\lambda_i$s. Thus, denote the solution of this problem to be $E_{2n}^\pi(\rho)$, it is equal to the solution of the following optimization problem,
\begin{equation}
\begin{split}
\min_{q_1,\cdots,q_n\in\mathbb{N}}& E_{2n}^\pi(\rho)= q_1\lambda_1+q_2\lambda_2+\cdots+q_n\lambda_n\\
\mathrm{s.t.}& \ \sum_{i=1}^nq_i\lambda_i^2=M_2^\pi,\cdots, \ \sum_{i=1}^nq_i\lambda_i^{2n}=M_{2n}^\pi\\
& \ q_1+q_2+\cdots+q_n\le L,
\end{split}
\end{equation}
where $L$ is the number of the singular values of $\mathcal{R}_\pi$ and $q_i$ is the degeneracy of singular value $\lambda_i$.
\end{theorem}
As a special case, when we only know the value of $M_2^\pi$ and $M_4^\pi$, the minimum of $\sum_i\lambda_i$ has an analytical form \cite{berry2003bounds}
\begin{equation}\label{eq:E_4 definition}
\begin{split}
E_4^\pi(\rho)=\sqrt{\frac{q(qM_2^\pi+U)}{q+1}}+\sqrt{\frac{M_2^\pi-U}{q+1}},
\end{split}
\end{equation}
where $q=\lfloor \frac{(M_2^\pi)^2}{M_4^\pi}\rfloor$ and $U=\sqrt{q(q+1)M_4^\pi-q (M_2^\pi)^2}$.

Now, we can formally represent the moment-based permutation criteria as
\begin{equation}\label{eq:criterionformulation}
\begin{split}
    E_{2n}^\pi(\rho)\le 1 \ , \ \forall \pi\in\mathcal{S}_{2k} \ , \ n\in\mathbb{N}
\end{split}
\end{equation}
for all separable $k$-partite state $\rho$. In fact, $E_{2n}^\pi(\cdot)$ may not necessarily be the function of $\rho$. Adopting the bipartite entanglement criterion introduced in Ref.~\cite{zhang2008entanglement}, we get
\begin{equation}\label{eq:enhanced CCNR}
\begin{split}
E_{2n}^{(2,3)}(\rho_{AB}-\rho_A\otimes\rho_B)\le \sqrt{(1-\tr\rho_A^2)(1-\tr\rho_B^2)}
\end{split}
\end{equation}
for separable $\rho_{AB}$.

\emph{Bipartite entanglement detection.} --- Compared with existing entanglement detection schemes based on partial transposed moments \cite{elben2020mixed,yu2021optimal,neven2021symmetry}, this framework is not only a direct generalization to multipartite entanglement, but also enhances the detection capability in the bipartite scenario.  With the second and fourth moments only, Eq.~\eqref{eq:enhanced CCNR} can detect $3\times 3$-dimensional bound entanglement constructed using the unextendible product basis proposed in Ref.~\cite{bennet1999unextendible}. We leave the detailed discussion in Appendix \ref{subsec:localentdecay}.

The criterion of Eq.~\eqref{eq:enhanced CCNR} also performs well in practical physical systems. We study the local bipartite entanglement dynamics in a quantum system evolved under a long-range $XY$ Hamiltonian. Specifically, we choose a 10-qubit open boundary Ising model with the Hamiltonian of the form
\begin{equation}\label{eq:XY Hamiltonian}
\begin{split}
H_{XY}=\sum_{i<j}J_{ij}(\hat{\sigma}_i^+\hat{\sigma}_j^-+\hat{\sigma}_i^-\hat{\sigma}_j^+)+ B_z\sum_i\hat{\sigma}_i^z,
\end{split}
\end{equation}
where $\hat{\sigma}_i^z$, $\hat{\sigma}_i^+$, and $\hat{\sigma}_i^-$ are the spin-$\frac{1}{2}$ Pauli-$Z$, raising, and lowering operator acting on the $i$-th qubit; $J_{ij}=\frac{J_0}{|i-j|^\alpha}$ is the interaction strength following the power-law decay with $J_0$ and $\alpha$ set to be $420s^{-1}$ and $1.24$, respectively \cite{brydges2019probing}; $B_z$ stands for transverse field and is set to be $400s^{-1}$. This Hamiltonian has been realized in real physical systems \cite{jurcevic2014quasiparticle,brydges2019probing} and is often served as the benchmark of detection capabilities of entanglement criteria \cite{neven2021symmetry,elben2020mixed}.

\begin{figure}
    \centering
    \includegraphics[width=8cm]{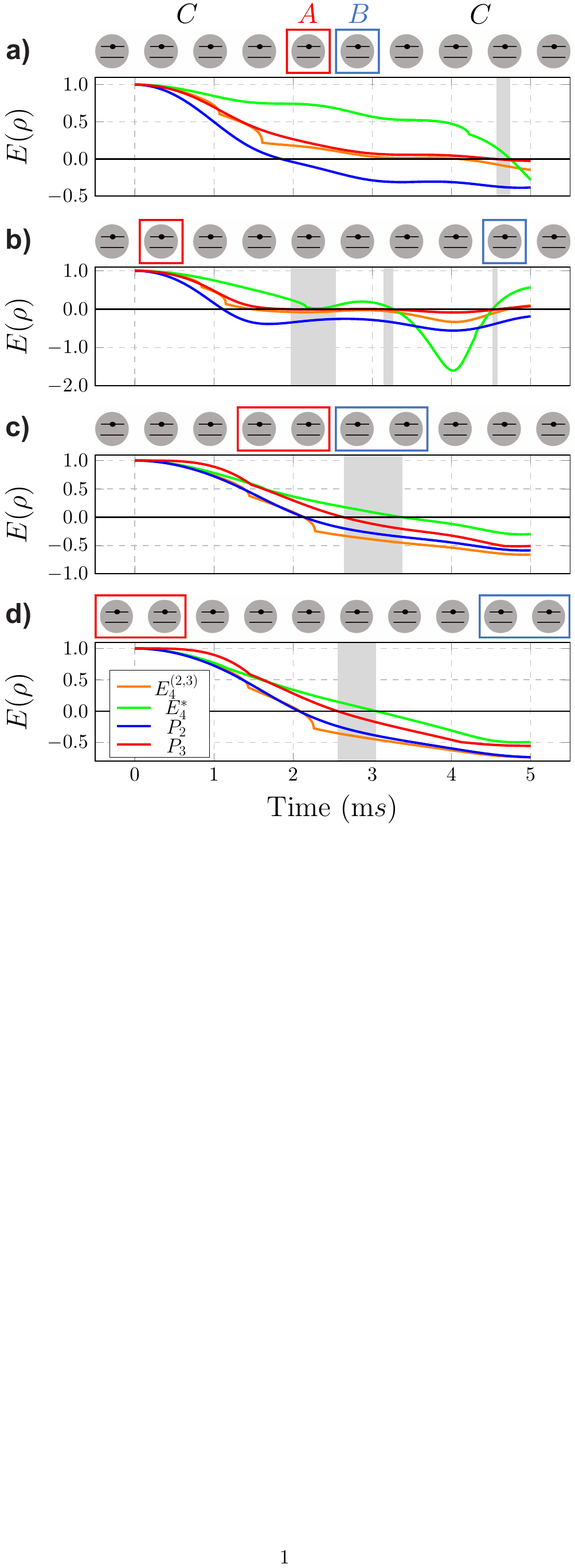}
\caption{Local entanglement decay in thermal system (Color Online). The entanglement dynamics of the local systems $A$ and $B$, which are marked by squares and initialized to be $\ket{\psi(t=0)}_{AB}=\frac{1}{\sqrt{2}}(\ket{0}^{\otimes N_{AB}}+\ket{1}^{\otimes N_{AB}})$. Qubits without squares are initialized to be the tensor product of $\ketbra{0}$, act as part $C$. The entanglement of $AB$ is detected when the value is above zero for each criterion. The grey areas represent the time periods in which the entanglement can only be detected by the $E_4^*$ criterion. }
    \label{fig:Quench Dynamics}
\end{figure}

The 10-qubit chain is divided into three parts, $A$, $B$ and $C$, where $A$ and $B$ constitute the local system we study, initialized to be $\frac{1}{\sqrt{2}}(\ket{0}^{\otimes N_{AB}}+\ket{1}^{\otimes N_{AB}})$. $C$ acts as the bath, which is initialized to be the tensor product of $\ket{0}$. We compare four implementable nonlinear criteria in investigating the entanglement dynamics of systems composed of $A$ and $B$. The first two criteria are \comments{moment-based permutation criteria based on $M_4^{(2,3)}$, }  Eq.~\eqref{eq:criterionformulation} and Eq.~\eqref{eq:enhanced CCNR}, when setting $\pi=(2,3)$ and $n=2$, labeled by $E_4^{(2,3)}$ and $E_4^*$, respectively. Others are the entropy criterion based on the comparison of the purities \cite{Horodecki2009entanglement}, labeled by $P_2$; and the weak-form PPT criterion based on $M_3^{(1,2)}=\tr[(\rho_{AB}^{T_A})^3]$ \cite{yu2021optimal,neven2021symmetry}, labeled by $P_3$. We define four quantities to represent these criteria which satisfy that $E(\rho)>0$ iff the entanglement is detected by the corresponding criterion. 

The numerical simulation results \cite{JOHANSSON2012qutip} are shown in Fig.~\ref{fig:Quench Dynamics}. One could find that \comments{the decay of the entanglement initially localized in $A$ and $B$ is observed by all four criteria. Besides, }the moment-based permutation criteria, especially $E_4^*$, have an obvious advantage since it detects entanglement while all others fail in various time periods and different choices of $A$ and $B$.

In addition to the strong detection capability, the key quantities in this framework, $E_{2n}^\pi(\rho)$, have clear mathematical meaning as they give the lower bounds of the permutation norms. The permutation norms, including entanglement negativity, can be treated as entanglement measures. We thus conjecture that these quantities can also be used as entanglement measures. In Appendix \ref{sec:physics}, we support this conjecture by showing that $E_4^{(2,3)}(\rho)$ can witness the entanglement scaling transition in a quantum dynamical phase transition \cite{abanin2019manybody,serbyn2013local,smith2016many,wu2016understanding} and the entanglement rainbow structure for the eigenstates of a thermal Hamiltonian \cite{ueda2020quantum,kim2014eth,cotler2021emergent}.

\emph{Multipartite entanglement detection.}---Another advantage of our framework lies in multipartite entanglement detection. There exist multipartite entangled states that are separable in any bipartition and thus cannot be detected by any criteria extended from the bipartite case \cite{jungnitsch2011taming}, including the PPT and CCNR criteria. 
Theorem~\ref{thm:LagrangeMultipler} provides us new means to design practical entanglement criteria for these states. We depict the sets of detectable multipartite entangled states of different criteria in Fig~\ref{fig:entlevel}.

An important example is also based on the unextendible product basis \cite{bennet1999unextendible}. Consider a three-qubit system and define four product pure states
\begin{equation}
\{\ket{\psi}_i\}_{i=1}^4=\{\ket{0,1,+},\ket{1,+,0},\ket{+,0,1},\ket{-,-,-}\},
\end{equation}
where $\ket{\pm}=(\ket{0}\pm\ket{1})/\sqrt{2}$. It has been proved that the state
\begin{equation}
    \rho = \frac{1}{4}\left(\mathbb{I}_8-\sum_{i=1}^4\ketbra{\psi_i}{\psi_i}\right)
\end{equation}
is separable in any bipartition and thus its detection needs a new kind of moment-based permutation criterion other than PPT and CCNR. We find that when setting $\pi=\tbinom{1,2,3,4,5,6}{1,3,2,4,5,6}$, realigning the first two parties and keeping the third party unchanged, the entanglement of this state can be detected using $E_8^\pi(\rho)$, which only requires four orders of moments. We leave some details of calculating $E_8^\pi(\rho)$ in Appendix \ref{subsec:multientdetection}.

\begin{figure}
    \centering
    \includegraphics[width=0.45\textwidth]{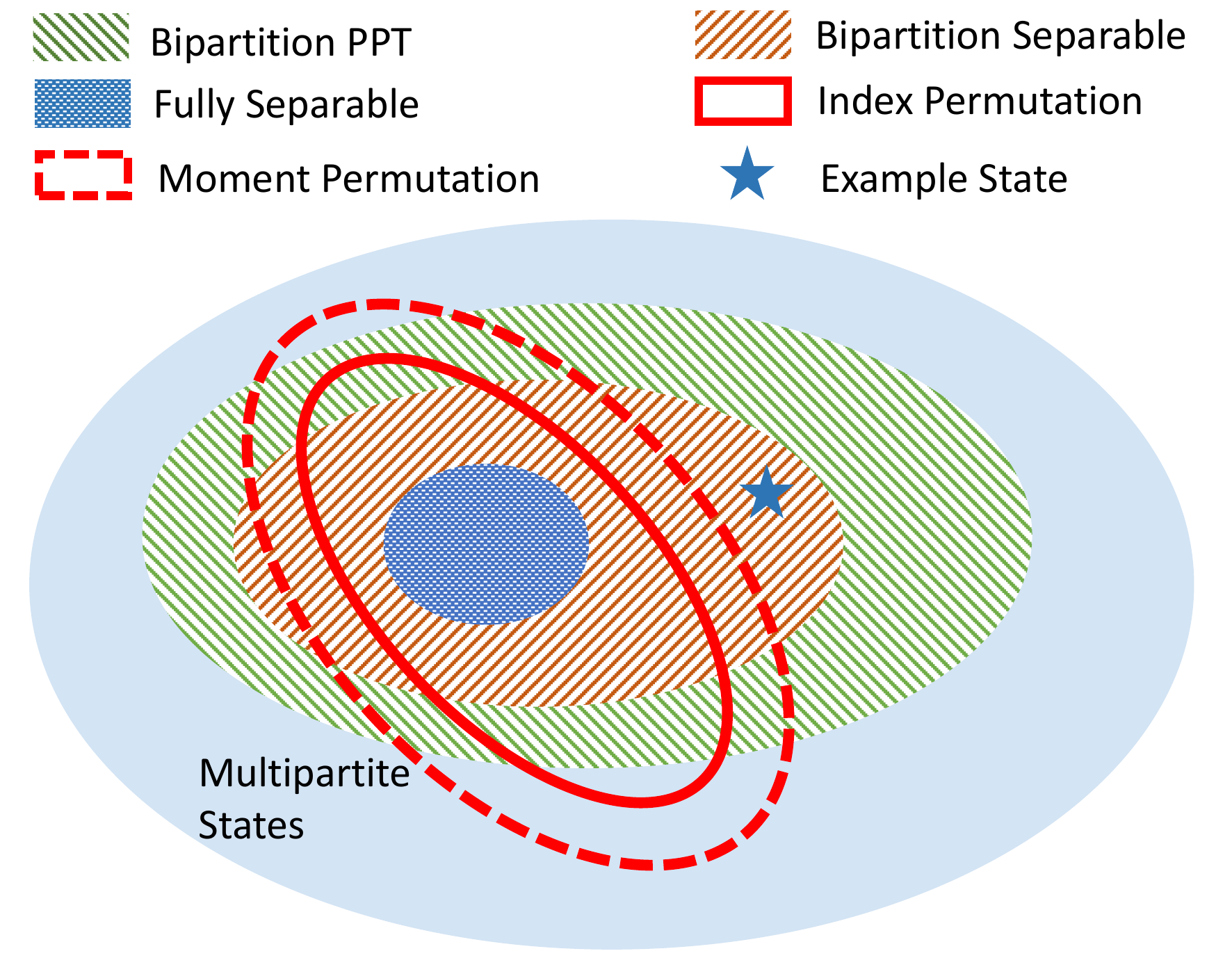}
\caption{Illustration of different sets of detectable states (Color Online). Bipartition PPT: PPT in any bipartition; Bipartition Separable: separable in any bipartition; Fully Separable: $\sum_ip_i\rho_{1}^i\otimes \cdots \otimes\rho_{k}^i$; Index Permutation: states that cannot be detected by an index permutation criterion other than the bipartite partial transposition; Moment Permutation: states that cannot be detected using finite numbers of permutation moments; Example State: a state that is separable in any bipartition while can be detected by a moment-based permutation criterion. When the number of moments increases, the dashed red circle will approach the solid red circle.
    }
    \label{fig:entlevel}
\end{figure}

For multipartite quantum systems, entanglement can have a rather complex entanglement structure \cite{Lu2018Structure}. At the same time, the tools for detecting entanglement structure are quite restrictive \cite{ren2021metrology}. In Appendix \ref{sec:structure}, we show that our framework can also be generalized to detect the multipartite entanglement structure.

\emph{Outlook.}---The techniques we developed in this work, including the moment measurements and bounding the lower-order moment using the higher-order moments, have many potential applications, like the positive map entanglement detection \cite{gunhe2009entanglement} and the trace distance estimation. Furthermore, it is also interesting to investigate how to generalize the framework to entanglement detection in continuous variable systems \cite{zhang2013CV}.

We thank Junjie Chen, Zhaohui Wei, and Xiaodong Yu for valuable discussions. This work was supported by the National Natural Science Foundation of China Grants No. 12174216 and No. 11875173 and the National Key Research and Development Program of China Grants No. 2019QY0702 and No. 2017YFA0303903.

\bibliographystyle{apsrev}

\appendix

\onecolumngrid
\newpage

\setcounter{theorem}{0}
\setcounter{table}{0}

\section{Preliminaries}\label{sec:Preliminary}
\subsection{Tensor Network Basis}
As our work is mainly based on index permutation criterion, much tensor calculation is needed to derive the results. So here we will introduce a graphical method to conduct the tensor calculation, the tensor network \cite{wood2011tensor}, which also plays a crucial role in quantum simulation. In the following context, we will frequently use the technique introduced in this section to do the tensor calculation.

In tensor network representation, a matrix is represented as a box with open legs, which correspond to uncontracted indices. Those left-oriented legs stand for row indices and right-oriented ones stand for column indices. Vectors and scalars are represented as boxes with legs of same orientation and boxes with no leg, respectively. The number of legs depends on the number of parties we are interested in. A $k$-partite state is represented by a tensor with $k$ pairs of legs, with each pair of legs standing for row and column indices of each party. Connecting legs stands for index contraction, like matrix production and taking trace:
\begin{eqnarray}
AB=
\begin{tabular}{c}
     \includegraphics[scale=0.15]{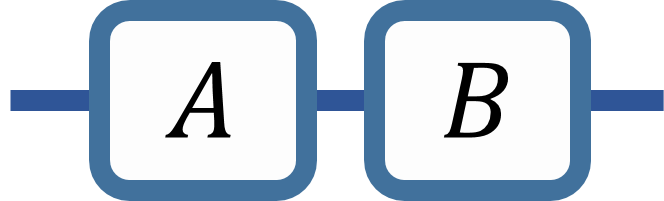} 
\end{tabular}
, \ \ \
\tr(A)=
\begin{tabular}{c}
     \includegraphics[scale=0.15]{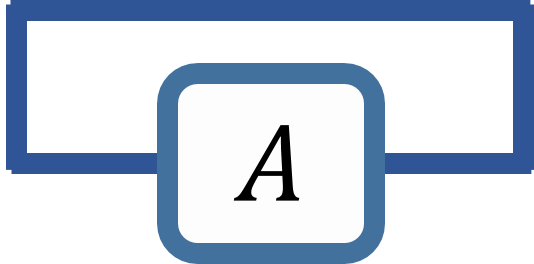} 
\end{tabular}.
\end{eqnarray}
For tensor production operation $A\otimes B$, where no indices are contracted, $A$ and $B$ are just put together with no leg connection:
\begin{eqnarray}
A\otimes B=
\begin{tabular}{c}
    \includegraphics[scale=0.15]{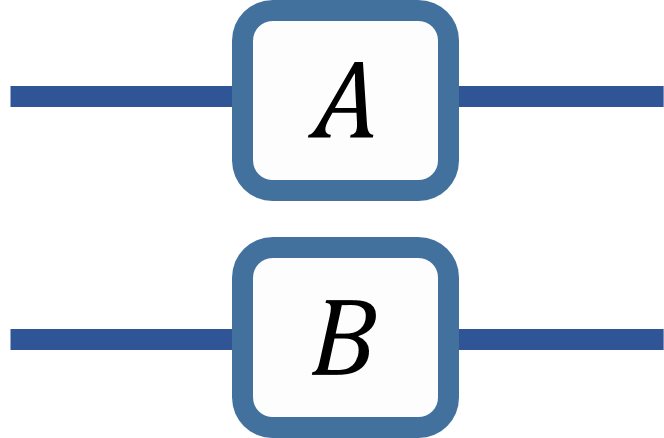}
\end{tabular}
\end{eqnarray}

Index permutation operations can be easily represented by changing the order of legs. Take partial transposition map and realignment map as examples,
\begin{eqnarray}
\rho_{AB}=
\begin{tabular}{c}
    \includegraphics[scale=0.15]{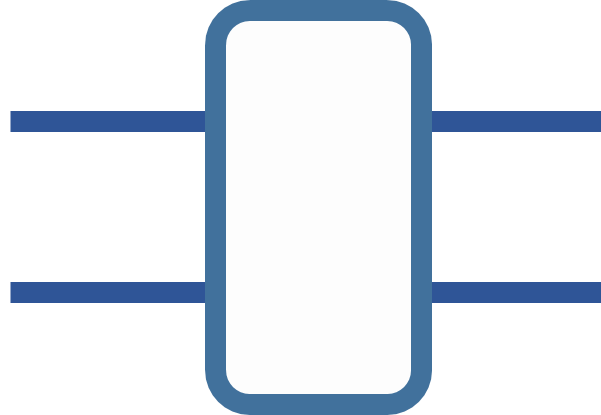}
\end{tabular}
, \ \ \
\mathcal{R}_{(2,3)}(\rho_{AB})=
\begin{tabular}{c}
    \includegraphics[scale=0.15]{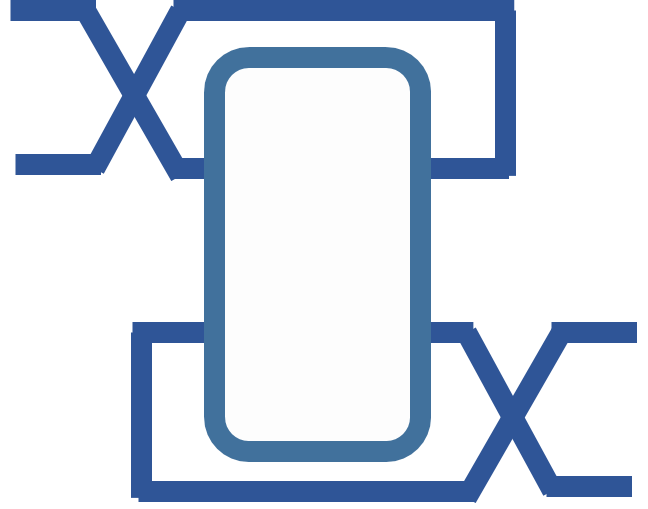}
\end{tabular}
, \ \ \
\rho_{AB}^{T_A}=\mathcal{R}_{(1,2)}(\rho_{AB})=
\begin{tabular}{c}
     \includegraphics[scale=0.15]{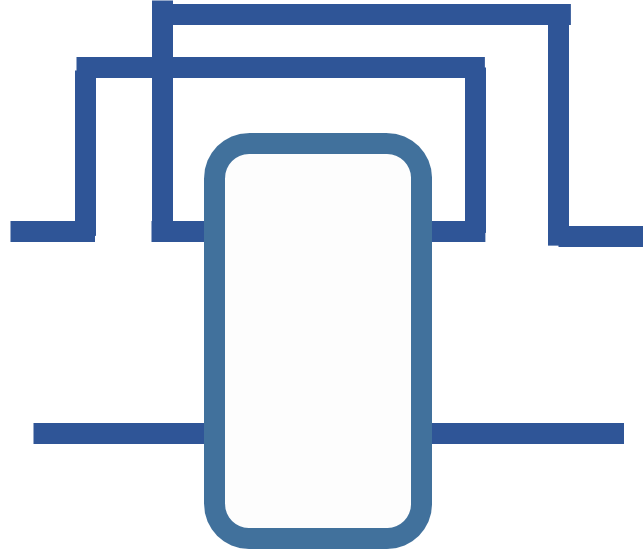}
\end{tabular}
.
\end{eqnarray}

In tensor network, permutation operators can be represented by changing the position of legs. SWAP operator is represented by exchanging two legs, cyclic operator is represented by sequentially moving each leg to the position of its neighboring leg. The operators we use to calculate the fourth moment of $\mathcal{R}_\pi$ can be graphically represented as:
\begin{eqnarray}
\mathbb{S}^{(1,2)}\otimes\mathbb{S}^{(3,4)}=
\begin{tabular}{c}
    \includegraphics[scale=0.15]{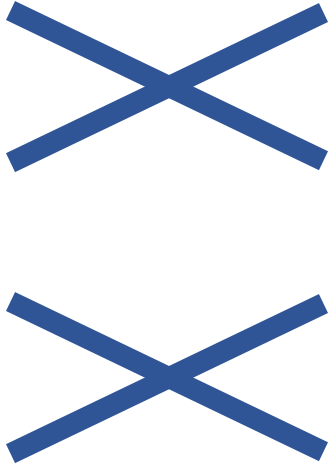}
\end{tabular}
, \ \
\mathbb{S}^{(2,3)}\otimes\mathbb{S}^{(4,1)}=
\begin{tabular}{c}
    \includegraphics[scale=0.15]{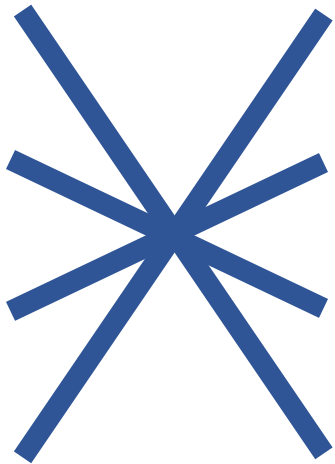}
\end{tabular}
, \ \
\overrightarrow{\Pi}=
\begin{tabular}{c}
    \includegraphics[scale=0.15]{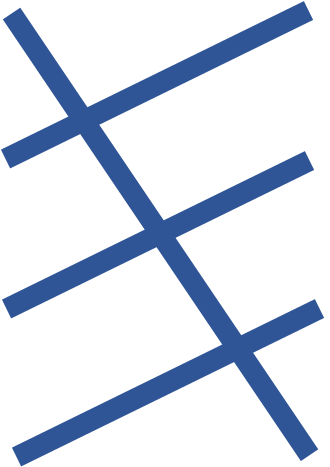}
\end{tabular}
, \ \
\overleftarrow{\Pi}=
\begin{tabular}{c}
    \includegraphics[scale=0.15]{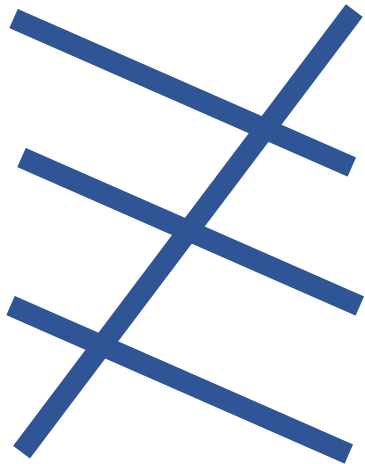}
\end{tabular}
.
\end{eqnarray}

\subsection{Random Unitary Basis}\label{subsec:random}
According to Schur-Weyl duality, random unitary is closely related to permutation operator, and is the basis of shadow estimation and randomized measurements, which will be discussed later. In unitary group, there exists a unique measure, which is called Haar measure, that satisfies
\begin{equation}
\begin{split}
\int_{\mathrm{Haar}}dU=1,\ \int_{\mathrm{Haar}}dUf(U)=\int_{\mathrm{Haar}}dUf(VU)=\int_{\mathrm{Haar}}dUf(UV)
\end{split}
\end{equation}
for arbitrary unitary $V$, where $U$ is integrated by Haar measure. Using Haar measure, one can define the $t$-fold twirling channel
\begin{equation}
\begin{split}
\Phi_t(O)=\int_{\mathrm{Haar}}dUU^{\otimes t}OU^{\dagger\otimes t},
\end{split}
\end{equation}
which equals to the linear combination of permutation operators \cite{gu2013moments,roberts2017chaos}
\begin{equation}
\begin{split}
\Phi_t(O)=\sum_{\pi,\sigma\in\mathcal{S}_t}C_{\pi,\sigma}\tr\left(\hat{W}_\pi O\right)\hat{W}_\sigma,
\end{split}
\end{equation}
where $\mathcal{S}_t$ is the permutation group of order $t$, $\pi$ and $\sigma$ are its elements, $C_{\pi,\sigma}$ is the element of Weigarten matrix, and $\hat{W}_\pi$ is the permutation operator corresponding to $\pi$.

However, to construct a $t$-fold twirling channel, Haar measure may not be necessary, one can only average over a unitary ensemble $\mathcal{E}_t$ with a finite number of elements
\begin{equation}
\begin{split}
\Phi_t(O)=\frac{1}{|\mathcal{E}_t|}\sum_{U\in\mathcal{E}_t}U^{\otimes t}O U^{\dagger \otimes t},
\end{split}
\end{equation}
where $|\mathcal{E}_t|$ denotes the number of elements in $\mathcal{E}_t$. All the unitary ensembles that satisfy this equation are called unitary $t$-design. The commonly-used Clifford group has been proved to be a unitary 3-design for multi-qubit systems \cite{zhu2017multiqubit}.

\section{Proof of Theorems}\label{sec:proofs}

\subsection{Observable of Permutation Moments}\label{subsec:proofofobservable}

\begin{theorem}
Given a $k$-partite state $\rho$ and the index permutation operation $\mathcal{R}_\pi$, the $2n$-th moment of $\mathcal{R}_{\pi}$, $M_{2n}^\pi:=\tr[(\mathcal{R}_{\pi}\mathcal{R}_{\pi}^\dagger)^n]$, can be estimated by observable measurement on 2n copies of $\rho$,
\begin{equation}\label{eq:main res}
\begin{split}
M_{2n}^\pi=\tr\left(O^\pi_{2n}\rho^{\otimes 2n}\right)=\frac{1}{2}\tr\left[\left(\bigotimes_{i=1}^k U^\pi_{i}+h.c.\right) \rho^{\otimes 2n}\right].
\end{split}
\end{equation}
For T1-type parties $U_{i}^\pi=\overrightarrow{\Pi}_i$ and for T2-type parties $U_{i}^\pi=\overleftarrow{\Pi}_i$.
Here $\overrightarrow{\Pi}$ and $\overleftarrow{\Pi}$ are the cyclic permutation operators in different direction, satisfying $\overrightarrow{\Pi}\ket{s_1,\cdots,s_{2n}}=\ket{s_{2n},s_1,\cdots,s_{2n-1}}$ and $\overleftarrow{\Pi}\ket{s_1,\cdots,s_{2n}}=\ket{s_2,\cdots,s_{2n},s_1}$. For R1-type parties $
U_{i}^\pi=\mbb{S}^{(2n,1)}_i\otimes \mbb{S}^{(2,3)}_i\otimes \cdots \otimes \mbb{S}^{(2n-2,2n-1)}_i$ and for R2-type parties $U_{i}^\pi=\mbb{S}^{(1,2)}_i\otimes \mbb{S}^{(3,4)}_i\otimes \cdots \otimes \mbb{S}^{(2n-1,2n)}_i$, where $\mbb{S}^{(u,v)}$ is the SWAP operator acting on the $u$-th and $v$-th copies.
\end{theorem}

\begin{proof}
The key of this proof is to figure out the index permutation rules in $\tr[(\mathcal{R}_\pi\mathcal{R}_\pi^\dagger)^n]$. First recall the definition of index permutation operation:
\begin{equation}\label{eq:A1}
\begin{split}
\left[\mathcal{R}_\pi\right]_{s_1s_2,\cdots,s_{2k-1}s_{2k}}=\rho_{s_{\pi(1)}s_{\pi(2)},\cdots,s_{\pi(2k-1)}s_{\pi(2k)}}.
\end{split}
\end{equation}
According to the linearity of $\mathcal{R}_\pi(\cdot)$ and hermitian of $\rho$,
\begin{equation}
\begin{split}
\mathcal{R}_\pi^\dagger=[\mathcal{R}_\pi(\rho)^*]^T=\mathcal{R}_\pi(\rho^*)^T=\mathcal{R}_\pi(\rho^T)^T,
\end{split}
\end{equation}
where $T$ denotes the transposition operation. Hence, the element of $\mathcal{R}_\pi(\rho)^\dagger$ is
\begin{equation}\label{eq:A3}
\begin{split}
\left[\mathcal{R}_\pi^\dagger\right]_{s_1s_2,\cdots,s_{2k-1}s_{2k}}&=[\mathcal{R}_\pi(\rho^T)]_{s_2s_1,\cdots,s_{2k}s_{2k-1}}\\
&=[\rho^T]_{s_{\pi'(1)}s_{\pi'(2)},\cdots,s_{\pi'(2k-1)}s_{\pi'(2k)}}\\
&=[\rho]_{s_{\pi'(2)}s_{\pi'(1)},\cdots,s_{\pi'(2k)}s_{\pi'(2k-1)}},
\end{split}
\end{equation}
where $\pi'(\cdot)$ denotes the corresponding row number or column number of $\pi(\cdot)$, if $\pi(\cdot)$ is odd, then $\pi'(\cdot)=\pi(\cdot)+1$; otherwise, $\pi'(\cdot)=\pi(\cdot)-1$. Make a comparison of Eq.~\eqref{eq:A1} and Eq.~\eqref{eq:A3}, we find that the $r$-th row index of $\mathcal{R}_\pi$ and the $r$-th column index of $\mathcal{R}_\pi^\dagger$ are the row (column) and column (row) indices from the same subsystem of $\rho$. And this is same for the $r$-th column index of $\mathcal{R}_\pi$ and $r$-th row index of $\mathcal{R}_\pi^\dagger$. Using tensor network representation, this conclusion tells us that if transversely drawing the tensors of $\mathcal{R}_\pi$ and $\mathcal{R}_\pi^\dagger$, we will find they are mirror symmetric. We take $\mathcal{R}_{(2,3)}$, $\mathcal{R}_{(2,3)}$ and a generic permutation matrix $\mathcal{R}_\pi$ as examples:
\begin{eqnarray}
\mathcal{R}_{(2,3)}^\dagger=&
\left(
\begin{tabular}{c}
    \includegraphics[scale=0.12]{R_rhoAB_.png}
\end{tabular}
\right)^\dagger=
\begin{tabular}{c}
    \includegraphics[scale=0.12]{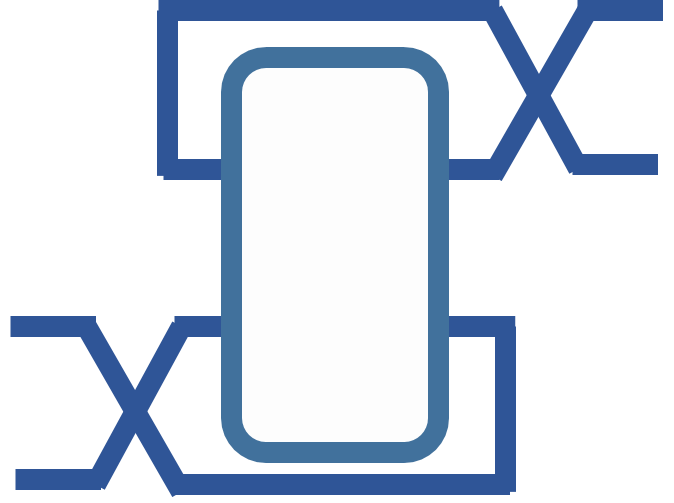}
\end{tabular}
, \ \ \
\mathcal{R}_{(1,2)}=
\left(
\begin{tabular}{c}
    \includegraphics[scale=0.15]{rhoABTA.png}
\end{tabular}
\right)^\dagger=
\begin{tabular}{c}
    \includegraphics[scale=0.15]{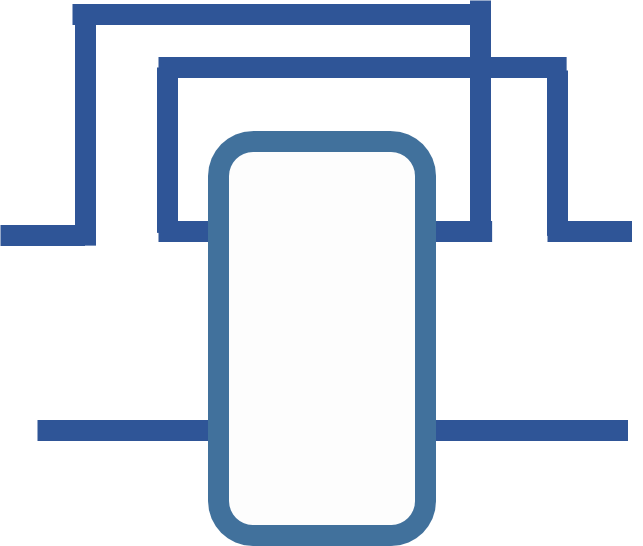}
\end{tabular}
,\\
&\mathcal{R}_\pi^\dagger=
\left(
\begin{tabular}{c}
    \includegraphics[scale=0.12]{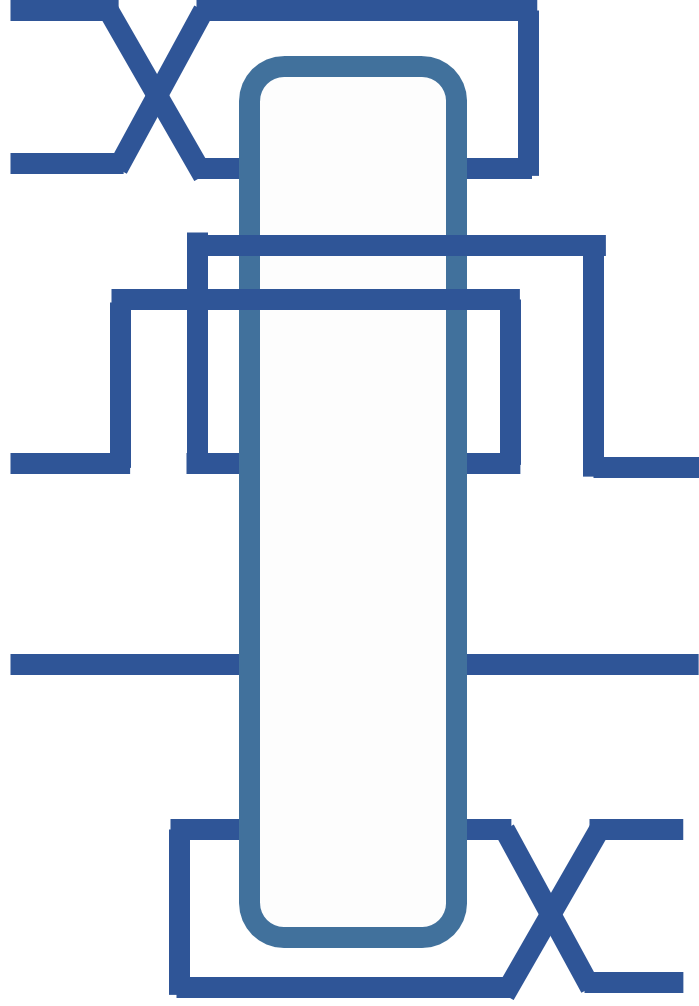}
\end{tabular}
\right)^\dagger=
\begin{tabular}{c}
    \includegraphics[scale=0.12]{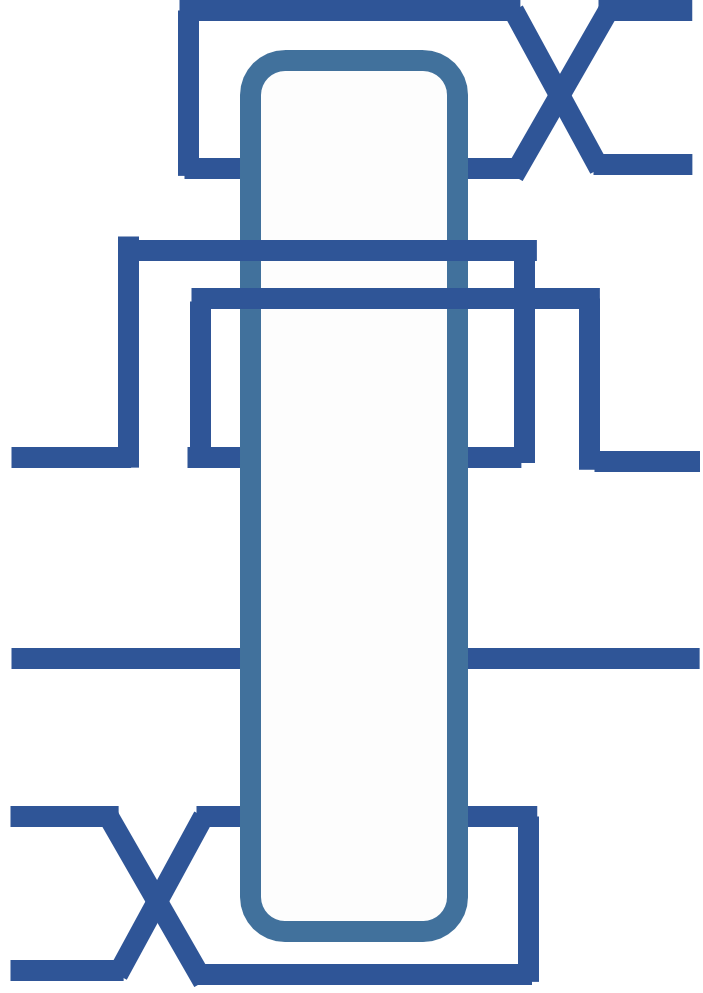}
\end{tabular}.
\end{eqnarray}

Because of this mirror symmetry property, when multiplying $\mathcal{R}_\pi$ and $\mathcal{R}_\pi^\dagger$, leg connections only occur in the same subsystem. Besides, for those R1-type subsystems, their two legs connect in the multiplication $\mathcal{R}_\pi\mathcal{R}_\pi^\dagger$, and for those R2-type subsystems, their two legs connect in $\mathcal{R}_\pi^\dagger\mathcal{R}_\pi$. While for T-type subsystem, because its two legs are on both sides of $\mathcal{R}_\pi$ and $\mathcal{R}_\pi^\dagger$, so it will connect with both of its neighboring $\rho$. Hence, in $\tr[(\mathcal{R}_\pi\mathcal{R}_\pi^\dagger)^n]$, one would find that the R-type parties of two neighboring $\rho$ are either connected by legs or have no connection at all, depending on their exact types; all the T-type parties are connected by their legs, while the directions of connection are different for T1-type and T2-type parties. This connection rule is equivalent to measuring SWAP operators, $\mbb{S}$, on connected neighboring R-type subsystems and cyclic permutation operators $\overrightarrow{\Pi}$ or $\overleftarrow{\Pi}$ on $2n$ copies of T-type subsystems. 

To summarize, 
\begin{equation}\label{eq:M2nunitary}
\begin{split}
M_{2n}^\pi=\tr\left[\left(\bigotimes_{i=1}^kU_i^\pi\right)\rho^{\otimes 2n}\right]
\end{split}
\end{equation}
where $U^\pi_i=\mathbb{S}_i^{(2,3)}\otimes \cdots\otimes\mathbb{S}_i^{(2n,1)}$ for R1-type parties, $U^\pi_i=\mathbb{S}_i^{(1,2)}\otimes \cdots\otimes\mathbb{S}_i^{(2n-1,2n)}$ for R2-type parties, $U^\pi_i=\overrightarrow{\Pi}_i$ for T1-type parties, and $U^\pi_i=\overleftarrow{\Pi}_i$ for T2-type parties. In fact, the definitions of $U_i^\pi$ for R1-type and R2-type parties can be exchanged, and the definitions of $U_i^\pi$ for T1-type and T2-type parties can also be exchanged. This can be proved by the permutation invariant properties of $M_{2n}^\pi$. 

Define $\Pi_{\mathrm{inv}}$ to be the $2n$-th order inversely permutation operator, which satisfies $\Pi_{\mathrm{inv}}[\rho^{(1)}\otimes\rho^{(2)}\cdots\otimes\rho^{(2n)}]\Pi_{\mathrm{inv}}^\dagger=\rho^{(2n)}\otimes\rho^{(2n-1)}\cdots\otimes\rho^{(1)}$. Because the $2n$ copies of $\rho$ in $M_{2n}^\pi$ are identical, we have
\begin{equation}
\begin{split}
M_{2n}^\pi=\tr\left[\left(\bigotimes_{i=1}^kU_i^\pi\right)\Pi_{\mathrm{inv}}\left(\rho^{\otimes 2n}\right)\Pi_{\mathrm{inv}}^\dagger\right]=\tr\left[\Pi_{\mathrm{inv}}^\dagger\left(\bigotimes_{i=1}^kU_i^\pi\right)\Pi_{\mathrm{inv}}\left(\rho^{\otimes 2n}\right)\right]=\tr\left[\left(\bigotimes_{i=1}^k\Pi_{\mathrm{inv},i}^\dagger U_i^\pi\Pi_{\mathrm{inv},i}\right)\left(\rho^{\otimes 2n}\right)\right].
\end{split}
\end{equation}
For R1-type and R2-type parties, $\Pi_{\mathrm{inv},i}^\dagger U_i^\pi\Pi_{\mathrm{inv},i}=U_i^\pi$, while for T1-type parties, $\Pi_{\mathrm{inv},i}^\dagger \overrightarrow{\Pi}_i\Pi_{\mathrm{inv},i}=\overrightarrow{\Pi}_i^\dagger=\overleftarrow{\Pi}_i$, and for T2-type parties, $\Pi_{\mathrm{inv},i}^\dagger \overleftarrow{\Pi}_i\Pi_{\mathrm{inv},i}=\overleftarrow{\Pi}_i^\dagger=\overrightarrow{\Pi}_i$. If we replace $\Pi_{\mathrm{inv},i}$ with $\overrightarrow{\Pi}$,
\begin{equation}
\begin{split}
M_{2n}^\pi=\tr\left[\left(\bigotimes_{i=1}^kU_i^\pi\right)\overrightarrow{\Pi}\left(\rho^{\otimes 2n}\right)\overrightarrow{\Pi}^\dagger\right]=\tr\left[\overrightarrow{\Pi}^\dagger\left(\bigotimes_{i=1}^kU_i^\pi\right)\overrightarrow{\Pi}\left(\rho^{\otimes 2n}\right)\right]=\tr\left[\left(\bigotimes_{i=1}^k\overrightarrow{\Pi}_i^\dagger U_i^\pi\overrightarrow{\Pi}_i\right)\left(\rho^{\otimes 2n}\right)\right].
\end{split}
\end{equation}
In this scenario, the definitions of $U_i^\pi$ of T1-type and T2-type parties will not change. While for R1-type parties, $\overrightarrow{\Pi}^\dagger_i[\mathbb{S}_i^{(2,3)}\otimes \cdots\otimes\mathbb{S}_i^{(2n,1)}]\overrightarrow{\Pi}_i=\mathbb{S}_i^{(1,2)}\otimes \cdots\otimes\mathbb{S}_i^{(2n-1,2n)}$, and for R2-type parties, $\overrightarrow{\Pi}^\dagger_i[\mathbb{S}_i^{(1,2)}\otimes \cdots\otimes\mathbb{S}_i^{(2n-1,2n)}]\overrightarrow{\Pi}_i=\mathbb{S}_i^{(2,3)}\otimes \cdots\otimes\mathbb{S}_i^{(2n,1)}$. Similarly, if we adopt $\overrightarrow{\Pi}\Pi_{\mathrm{inv}}$ to rotate these $2n$ copies of $\rho$, we can prove that the definitions of $U_i^\pi$ for R1-type and R2-type parties, and the definitions for T1-type and T2-type parties can be exchanged simultaneously.

Therefore, we can rewrite Eq.~\eqref{eq:M2nunitary} into observable form
\begin{equation}
\begin{split}
    M_{2n}^\pi=\frac{1}{2}\tr\left[\left(\bigotimes_{i=1}^kU_i^\pi+h.c.\right)\rho^{\otimes 2n}\right],
\end{split}
\end{equation}
where $U_i^\pi$ follows the same definition in Eq.\eqref{eq:M2nunitary}.
\end{proof}
Take a four-partite state $\rho$ as an example. The four parties of $\rho$ are all different types with respect to $\mathcal{R}_\pi(\cdot)$, T1-type, T2-type, R1-type, and R2-type respectively. It can be represented using tensor network as
\begin{eqnarray}
\mathcal{R}_\pi(\rho)=
\begin{tabular}{c}
     \includegraphics[scale=0.12]{Rpirho.png}
\end{tabular}
.
\end{eqnarray}
Because $\mathcal{R}_\pi$ and $\mathcal{R}_\pi^\dagger$ are mirror symmetric, the fourth moment, $M_4^\pi=\tr[(\mathcal{R}_\pi\mathcal{R}_\pi^\dagger)^2]$ can be graphically represented as
\begin{eqnarray}
M_4^\pi=\tr\left(
\begin{tabular}{c}
     \includegraphics[scale=0.12]{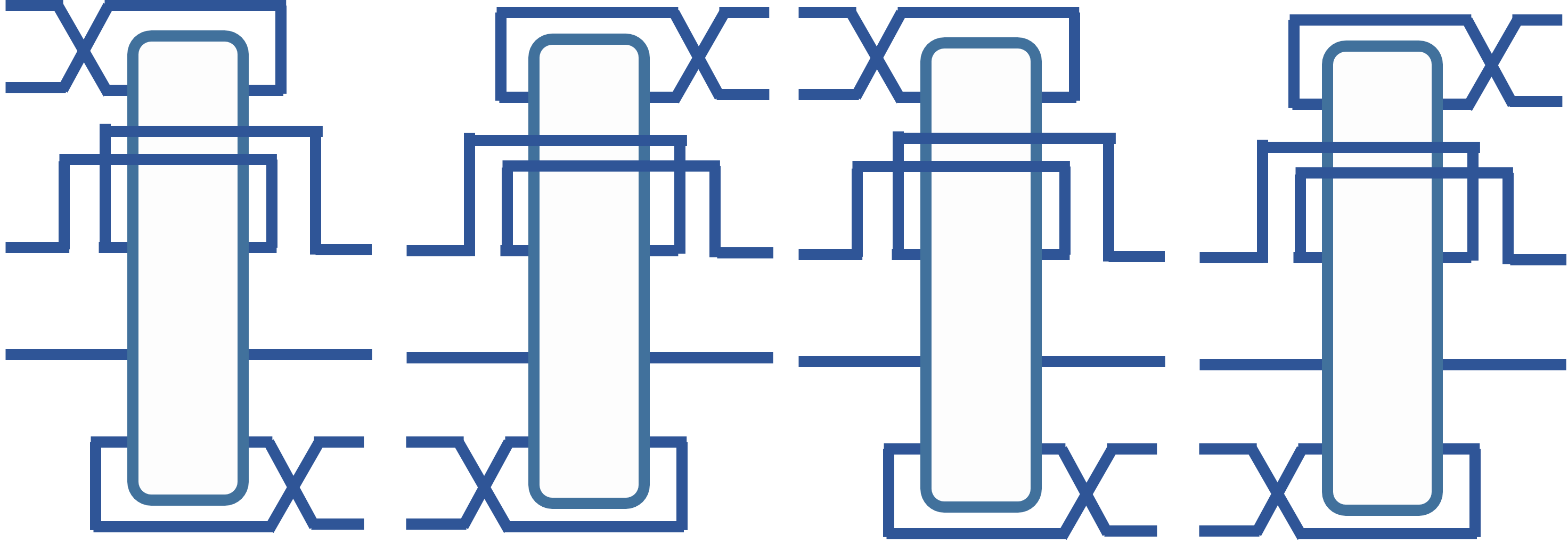}
\end{tabular}
\right)
=
\begin{tabular}{c}
     \includegraphics[scale=0.12]{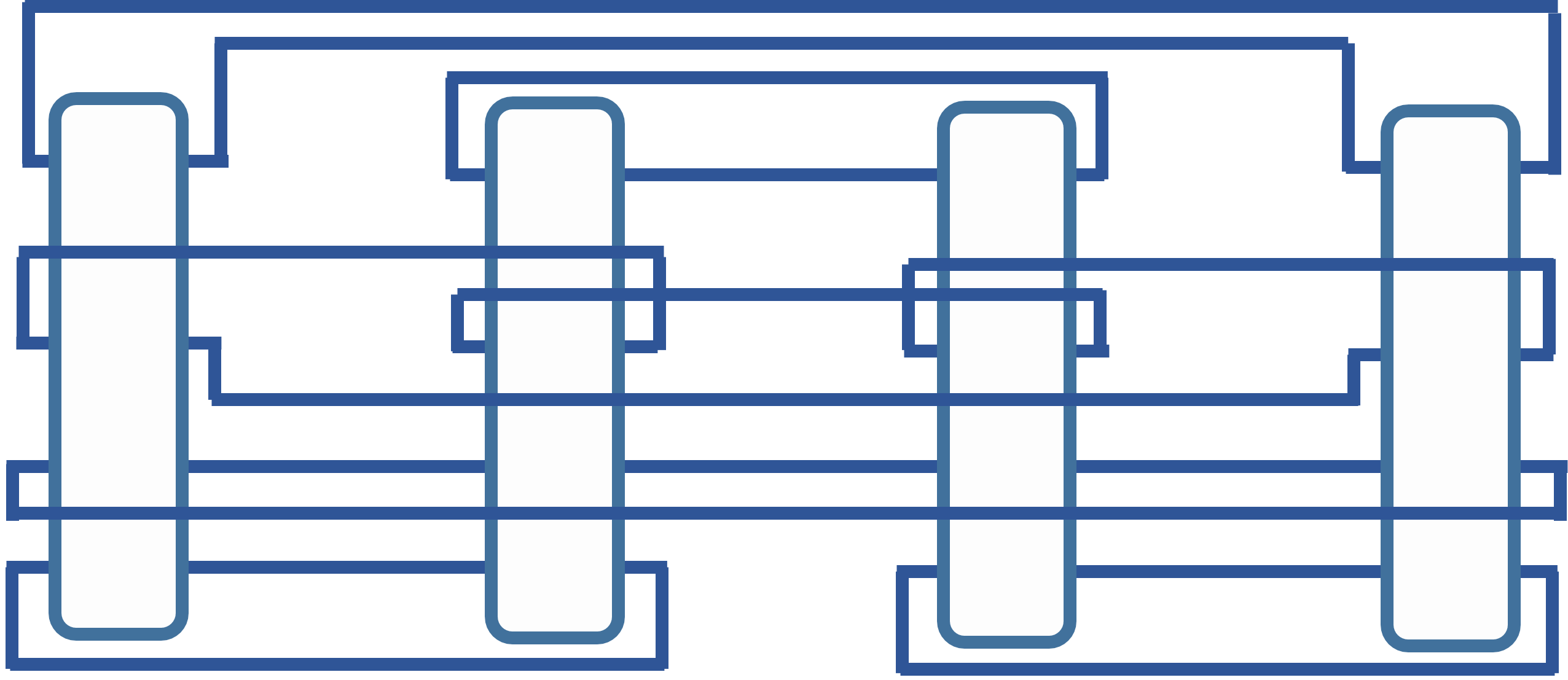}
\end{tabular}
.
\end{eqnarray}
As we said, one can find that the connection of legs belonging to T-type parties is among the four copies of $\rho$, while the connection of R-type legs is between two neighboring $\rho$. Hence, $M_4^\pi$ is equivalent to measuring 
\begin{equation}
\begin{split}
    O^\pi=\frac{1}{2}\left\{\left[\mathbb{S}_1^{(4,1)}\otimes\mathbb{S}_1^{(2,3)}\right]\otimes\overleftarrow{\Pi}_2\otimes\overrightarrow{\Pi}_3\left[\mathbb{S}_4^{(1,2)}\otimes\mathbb{S}_4^{(3,4)}\right]+h.c.\right\}
\end{split}
\end{equation}
on $4$ copies of $\rho$.

\subsection{Optimization Problem}\label{subsec:optimization}
\begin{theorem}\label{theorem:Lagrange Multipler}
The minimum value of $\norm{\mathcal{R}_\pi}$ given $M_2^\pi$, ..., $M_{2n}^\pi$ is reached when there are at most $n$ non-zero $\lambda_i$s. Thus, denote the solution of this problem to be $E_{2n}^\pi(\rho)$, the minimum value of $\sum_i\lambda_i$ is the solution of the following optimization problem,
\begin{equation}
\begin{split}
\min_{q_1,\cdots,q_n\in\mathbb{N}}& E_{2n}^\pi(\rho)= q_1\lambda_1+q_2\lambda_2+\cdots+q_n\lambda_n\\
\mathrm{s.t.}& \ \sum_{i=1}^nq_i\lambda_i^2=M_2^\pi,\cdots, \ \sum_{i=1}^nq_i\lambda_i^{2n}=M_{2n}^\pi\\
& \ q_1+q_2+\cdots+q_n\le L
\end{split}
\end{equation}
\end{theorem}
\begin{proof}
The original optimization problem:
\begin{equation}\label{eq:biopt}
\begin{split}
\min_{\{\lambda_i\}} \ \ &\sum_{i=1}^L \lambda_i\\
\text{subject to} \ \ &\sum_{i=1}^L \lambda_i^{2}=M_{2}^\pi \\
&\vdots\\
&\sum_{i=1}^L \lambda_i^{2n}=M_{2n}^\pi \\
& \lambda_1\ge \lambda_2\ge\cdots\ge \lambda_L\ge 0.
\end{split}
\end{equation}
Here, $L$ is the number of the singular values. Let $\lambda_i=x_L^2+\cdots+x_i^2$ with $x_1,\cdots,x_L\in \mathbb{R}.$ The above optimization problem can be written as  
\begin{equation}
\begin{split}
\min_{\{x_i\}} \ \ &\sum_{i=1}^L\left( x_L^2+\cdots+x_i^2\right)\\
\text{subject to} \ \ &\sum_{i=1}^L (x_L^2+\cdots+x_i^2)^{2}=M_{2}^\pi \\
&\vdots\\
&\sum_{i=1}^L (x_L^2+\cdots+x_i^2)^{2n}=M_{2n}^\pi \\
&x_1,\cdots,x_L\in \mathbb{R}.
\end{split}
\end{equation}
Define the Lagrange function as 
\begin{equation}
\begin{split}
\mathcal{L}&=\sum_{i=1}^L  (x_L^2+\cdots+x_i^2)+\alpha_1\Big(\sum_{i=1}^L (x_L^2+\cdots+x_i^2)^{2}-M_{2}^\pi\Big)\\
&+\cdots +\alpha_n\left(\sum_{i=1}^L (x_L^2+\cdots+x_i^2)^{2n}-M_{2n}^\pi\right).
\end{split}
\end{equation}
Take partial derivative with respect to all variables and the minimum value in the optimization problem can be achieved at the points where all derivatives are equal to $0.$ Therefore, we have
\begin{equation}
\begin{split}
0=&\frac{\partial \mathcal{L}}{\partial x_L } =2L x_L+4\alpha_1 x_L\sum_{i=1}^L (x_L^2+\cdots+x_i^2)+\cdots+4n\alpha_n x_L\sum_{i=1}^L (x_L^2+\cdots+x_i^2)^{2n-1},\\
&\vdots\\
0=&\frac{\partial \mathcal{L}}{\partial x_{1} }=2x_1+4\alpha_1 x_1 (x_L^2+\cdots+x_1^2)+\cdots+4n\alpha_n x_1(x_L^2+\cdots+x_1^2)^{2n-1},
\end{split}
\end{equation}
substituting $\lambda_i=x_L^2+\cdots+x_i^2$ into this equation,
\begin{equation}
\begin{split}
x_1\left[1+2\alpha_1\lambda_1+\cdots+2n\alpha_n\lambda_1^{2n-1}\right]=&0\\
&\vdots\\
x_L\left[L+2\alpha_1\sum_{i=1}^L\lambda_i+\cdots+2n\alpha_n\sum_{i=1}^L\lambda_i^{2n-1}\right]=&0.
\end{split}
\end{equation}
These above equations indicate that the extreme points satisfy	
\begin{equation}\label{eq:origin solution}
\begin{split}
\lambda_L=&0\quad \text{or}  \quad L+2\alpha_1\sum_{i=1}^L \lambda_i+\cdots+2n\alpha_n\sum_{i=1}^L \lambda_i^{2n-1}=0,\\
\lambda_L=&\lambda_{L-1}\quad \text{or}  \quad L-1+2\alpha_1\sum_{i=1}^{L-1} \lambda_i+\cdots+2n\alpha_n\sum_{i=1}^{L-1} \lambda_i^{2n-1}=0,\\
&\vdots\\
\lambda_2=&\lambda_{1}\quad \text{or}  \quad 1+2\alpha_1\lambda_1+\cdots+2n\alpha_n\lambda_1^{2n-1}=0.
\end{split}
\end{equation}
After simple analysis, there's only two possible kinds of $\lambda_i$ that satisfies the above equations: $\lambda_i=0$ or are the roots of the high-degree equation
\begin{equation}\label{eq:extreme}
2n\alpha_n\lambda^{2n-1}+2(n-1)\alpha_{n-1}\lambda^{2n-3}\cdots+2\alpha_1\lambda+1=0.
\end{equation}
Notice that all even-degree terms are zero and 
\begin{equation}
2n\alpha_n\lambda^{2n-1}+2(n-1)\alpha_{n-1}\lambda^{2n-3}\cdots+2\alpha_1\lambda
\end{equation}
is an odd function. As a consequence, Eq.~\eqref{eq:extreme} has at most $n$ positive roots. This implies that the set $\{\lambda_i\}$ have at most $n$ different positive values and the rest are $0$. Assume the $n$ different positive values are $\lambda_1\ge \cdots\ge \lambda_n\ge 0$. Denote the number of $\lambda_i$ as $q_i$ and the sum satisfies $q_1+\cdots +q_n\le L$. The optimization problem can be reduced to a new simpler optimization problem.

\begin{equation}
\begin{split}
\min_{\{q_i\in\mathbb{N}\}} \ \ &q_1 \lambda_1+\cdots+q_n \lambda_n\\
\text{subject to} \ \ &\sum_{i=1}^n q_i\lambda_i^{2}=M_{2}^\pi \\
&\vdots\\
&\sum_{i=1}^n q_i\lambda_i^{2n}=M_{2n}^\pi \\
& \lambda_1\ge \lambda_2\ge\cdots\ge \lambda_n\ge 0\\
& q_1+\cdots q_n\le L.
\end{split}
\end{equation}
\end{proof}

Particularly, we give an analytical solution in the case $n=2$. The optimization problem is
\begin{equation}\label{eq:optimizationM2M4}
\begin{split}
\min_{\{\lambda_i\}} \ \ &\sum_{i=1}^L \lambda_i\\
\text{subject to} \ \ &\sum_{i=1}^L \lambda_i^{2}=M_{2}^\pi \\
&\sum_{i=1}^L \lambda_i^{4}=M_{4}^\pi \\
& \lambda_1\ge \lambda_2\ge\cdots\ge \lambda_L\ge 0.
\end{split}
\end{equation}
To solve this problem, we need to adopt a theorem in \cite{berry2003bounds}.
\begin{fact}
Suppose $\{p_1,p_2,\cdots,p_{L}\}$ is a probability distribution, namely,  $p_i\geq 0$ and $\sum_i p_i=1.$ The sum $H_f=\sum_i f(p_i)$ is called an entropy measure, if $f(\cdot)$ is a function satisfying the following conditions,
\begin{enumerate}
    \item f(0)=0,
    \item $f(\cdot)$ is strictly convex or strictly concave,
    \item the first derivative $f^{\prime}$ exists and is continuous in the interval $(0,1)$.
\end{enumerate}

Let $f(\cdot)$ and $g(\cdot)$ be two functions satisfying the above conditions. Thus, $H_f=\sum_i f(p_i)$ and $H_g=\sum_i g(p_i)$ are two entropy measures. Define $\tilde{f}^{\prime}(g^{\prime})=f^{\prime}(p(g^{\prime}))$, which is formulated by treating $f^{\prime}=\frac{df}{dp}$ as the function of $g^{\prime}=\frac{dg}{dp}$. If $\tilde{f}^{\prime}(g^{\prime})$ is a strictly convex function of $g^{\prime}$, then the solution of the optimization problem
\begin{equation}\label{eq:optimizationentropy}
\begin{split}
\min_{\{p_i\ge 0\}} \ \ &H_f=\sum_{i=1}^L f(p_i)\\
subject \ to \ \ &\sum_{i=1}^L p_i=1\\
& \sum_{i=1}^Lg(p_i)=H_g
\end{split}
\end{equation}
is obtained when
\begin{equation}
\begin{split}
\{p_i\}=\{p_1,\cdots,p_1,p_2,0,\cdots,0\},\qquad p_2=1-p_1\left \lfloor \frac{1}{p_1}  \right \rfloor< p_1.
\end{split}
\end{equation}
If $\tilde{f}^{\prime}(g^{\prime})$ is a strictly concave function of $g^{\prime}$, the minimum is obtained when
\begin{equation}
\{p_i\}=\{p_1,p_2,\cdots,p_2\},\qquad p_2=\frac{1-p_1}{d-1}\leq p_1.
\end{equation}

\end{fact}
Substituting $p_i=\lambda_i^2$ , $f(p)=\sqrt{p}$  and $g(p)=p^2$ and ignore the normalization constant, Eq.~\eqref{eq:optimizationentropy} can be converted to our target optimization problem, Eq.~\eqref{eq:optimizationM2M4}. It is easy to verify that the two functions, $f(p)=\sqrt{p}$  and $g(p)=p^2$, satisfy the required three conditions, and $\tilde{f}^{\prime}(g^{\prime})=1/(\sqrt{2g^{\prime}})$ is strictly convex. Therefore, the minimum value is achieved when 
\begin{equation}
\{\lambda_i\}=\{\lambda_1,\cdots,\lambda_1,\lambda_2,0,\cdots,0\}.
\end{equation}
Suppose the number of $\lambda_1$ is $q$, then institute it into the constraint conditions, one gets
\begin{equation}
\begin{split}
\lambda_1 &=\frac{1}{q\sqrt{q+1}}\sqrt{q(qM_2^\pi+U)}
,\\
\lambda_2 &=\frac{1}{\sqrt{q+1}}\sqrt{M_2^\pi-U},
\end{split}
\end{equation}
with $U=\sqrt{q(q+1)M_4^\pi-q(M_2^\pi)^2}$. To make this a valid solution, there is only one possible value of $q$, which is $q=\left \lfloor \frac{(M_2^\pi)^2}{M_4} \right \rfloor$. Therefore, the solution of the target optimization problem Eq.~\eqref{eq:optimizationM2M4} is  
\begin{equation}
\begin{split}
    E_4^\pi=\frac{1}{\sqrt{q+1}}\sqrt{q(qM_2^\pi+U)}+\frac{1}{\sqrt{q+1}}\sqrt{M_2^\pi-U}.
\end{split}
\end{equation}

\subsection{Entanglement Structure Detection}\label{sec:structure}
In the multipartite system, when relaxing the restriction of LOCC, states can have much more complicated entanglement structures \cite{Lu2018Structure} than in the bipartite scenario. Concepts like entanglement depth and intactness naturally arise. A $k$-partite state $\rho$ is called $t$-separable iff it can be written as 
\begin{equation}
\begin{split}
\rho=\sum_{g_1\cup g_2\cup\cdots\cup g_t=[k]}\sum_i p_{\{g\},i}\rho_{g_1}^i\otimes \rho_{g_2}^i\otimes\cdots\otimes\rho_{g_t}^i,
\end{split}
\end{equation}
where $g_1$, ..., $g_t$ are $t$ disjoint non-empty sets, $\{p_{\{g\},i}\}$ are probabilities satisfying the normalization condition, $\sum_{\{g\},i}p_{\{g\},i}=1$. The largest number of $t$ for a given $k$-partite state $\rho$ is called the \emph{entanglement intactness} of $\rho$.

To generalize our moment-based permutation criteria to entanglement structure detection, we need to first develop the norm-based permutation entanglement structure criteria, and then prove that they can also be estimated by the permutation moments. Adopting the Ky Fan matrix norm \cite{horn1994topics}, we find two indicators for entanglement intactness.

\begin{theorem}
Given a $k$-partite state $\rho$ with intactness $t$, it satisfies
\begin{equation}
\begin{split}
    G_R(\rho)&=\sum_{g\subsetneq [k]}\norm{\mathcal{R}_{(2,3)}(\rho_{g,\bar{g}})}_{d^2}\le(2^k-2^t)d+(2^t-2),\\
    G_T(\rho)&=\sum_{g\subsetneq [k]}\norm{\mathcal{R}_{(1,2)}(\rho_{g,\bar{g}})}_{d^2}\le(2^k-2^t)d+(2^t-2),
\end{split}
\end{equation}
where $\norm{\cdot}_{d^2}$ is the Ky Fan $d^2$ matrix norm, defined by the sum of $d^2$ largest singular values of the matrix; $d$ is an integer that is not greater than the dimension of the smallest party; $\rho_{g,\bar{g}}$ is the bipartite state constructed by treating $g$ and $\bar{g}$ as the two parties of $\rho$, where $\bar{g}$ is the complement of $g$; $\mathcal{R}_{(2,3)}$ and $\mathcal{R}_{(1,2)}$ are the bipartite index permutation operations acting on the indices of parties $g$ and $\bar{g}$.
\end{theorem}

When $k=3$, this criterion gives the genuine tripartite entanglement criterion proposed in Ref.~\cite{Li2017genuine}.

\begin{proof}
Here, we use $G_R(\rho)$ as an example to prove our theorem, the proof for $G_T(\rho)$ is quite similar. According to the convexity of Ky Fan norm, these two functions are all convex:
\begin{equation}
\begin{split}
G_R(a\rho_1+b\rho_2)&=\sum_{g\subsetneq [k]}\norm{\mathcal{R}_{(2,3)}\left[a(\rho_1)_{g,\bar{g}}+b(\rho_2)_{g,\bar{g}}\right]}_{d^2}\\
&\le\sum_{g\subsetneq [k]}a\norm{\mathcal{R}_{(2,3)}\left[(\rho_1)_{g,\bar{g}}\right]}_{d^2}+b\norm{\mathcal{R}_{(2,3)}\left[(\rho_2)_{g,\bar{g}}\right]}_{d^2}\\
&=a G_R(\rho_1)+b G_R(\rho_2).
\end{split}
\end{equation}
Hence, for a state with intactness $t$, $\rho=\sum_{i,\psi_i\in t-int}p_i\ketbra{\psi_i}{\psi_i}$, 
\begin{equation}
\begin{split}
G_R(\rho)\le \max_{\psi\in t-int}G_R(\psi),
\end{split}
\end{equation}
where we use $t-int$ to denote the set of states that have intactness $t$. For pure state with intactness $t$, $\ket{\psi}=\bigotimes_{i=1}^t\ket{\psi_i}$, there's $2^t-2$ non-trivial partition $g|\bar{g}$ that ensure each $\ket{\psi_i}$ is in either $g$ or $\bar{g}$. In this case, $g$ and $\bar{g}$ are two separate parts, $\norm{\mathcal{R}_{(2,3)}(\psi_{g,\bar{g}})}_{d^2}=1$ because $\tr\left\{[\mathcal{R}_{(2,3)}(\psi_{g,\bar{g}})\mathcal{R}_{(2,3)}(\psi_{g,\bar{g}})^\dagger]^n\right\}=1$ for all $n$. In other cases, there exists some $\psi_i$ that distributed both in $g$ and $\bar{g}$, hence $g$ and $\bar{g}$ are two entangled parties. Also use singular values to denote the norm, $\norm{\mathcal{R}_{(2,3)}(\psi_{g,\bar{g}})}_{d^2}=\sum_{i=1}^{d^2}\lambda_i$, where $\lambda_i$ are arranged in decreasing order, the purity condition $\tr(\psi^2)=1$ gives a constrain that $\sum_{i=1}^{d^2}\lambda_i^2\le 1$. Recalling that singular values are all non-negative, hence the maximum value of $\sum_{i=1}^{d^2}\lambda_i$ under this purity constrain is $d$.

Considering other constraints in addition to purity constraint, the maximum will be less than $d$. In fact, the exact maximum value depends on the number of $\psi_i$ that is not fully contained in $g$ or $\bar{g}$. We will leave the calculation of the exact upper bound in our future work. Therefore, $(2^k-2^t)d+(2^t-2)$ generally cannot be achieved by some $k$-partite states with intactness $t$. However, this bound is nontrivial because $k$-partite GHZ state reaches $(2^k-2)d$ which is larger than $(2^k-2^t)d+(2^t-2)$ for any $t\ge 2$. 

\end{proof}

The lower bound of Ky Fan norm can also be constructed using moments. Here we only consider the case knowing the second and fourth moments. The optimization problem for entanglement structure detection is similar as one in Sec.~\ref{subsec:optimization}, and can be simplified as 
\begin{equation}\label{eq:gmeopt}
\begin{split}
\min_{\lambda_i} \ \ &\sum_{i=1}^l \lambda_i\\
\text{subject to} \ \ &\sum_{i=1}^L \lambda_i^{2}=M_{2} \\
&\sum_{i=1}^L \lambda_i^{4}=M_{4} \\
& \lambda_1\ge \lambda_2\ge\cdots\ge \lambda_L\ge 0\\
& l\le L.
\end{split}
\end{equation}
Using the similar strategy to simplify this optimization problem, it can be obtained that the extreme points satisfy	
\begin{equation}
\begin{split}
\lambda_L&=0\quad \text{or}  \quad l+2\alpha\sum_{i=1}^L \lambda_i+4\beta\sum_{i=1}^L \lambda_i^{3}=0,\\
&\vdots\\
\lambda_{l+1}&=\lambda_{l}\quad \text{or}  \quad l+2\alpha\sum_{i=1}^l \lambda_i+4\beta\sum_{i=1}^l \lambda_i^{3}=0,\\
\lambda_l&=\lambda_{l-1}\quad \text{or}  \quad l-1+2\alpha\sum_{i=1}^{l-1} \lambda_i+4\beta\sum_{i=1}^{l-1} \lambda_i^{3}=0,\\
&\vdots\\
\lambda_2&=\lambda_{1}\quad \text{or}  \quad 1+2\alpha \lambda_1+4\beta\lambda_1^{3}=0,
\end{split}
\end{equation}
where $\alpha$ and $\beta$ are Lagrange coefficients. Compared with the original optimization problem, Eq.~\eqref{eq:origin solution}, the above equations are different in constant terms of the first $(L-l)$ lines. We will see in the following analysis that this difference will make the problem much harder than the original problem. We discuss the solution in different cases. 

\begin{figure}
    \centering
    \includegraphics[width=15cm]{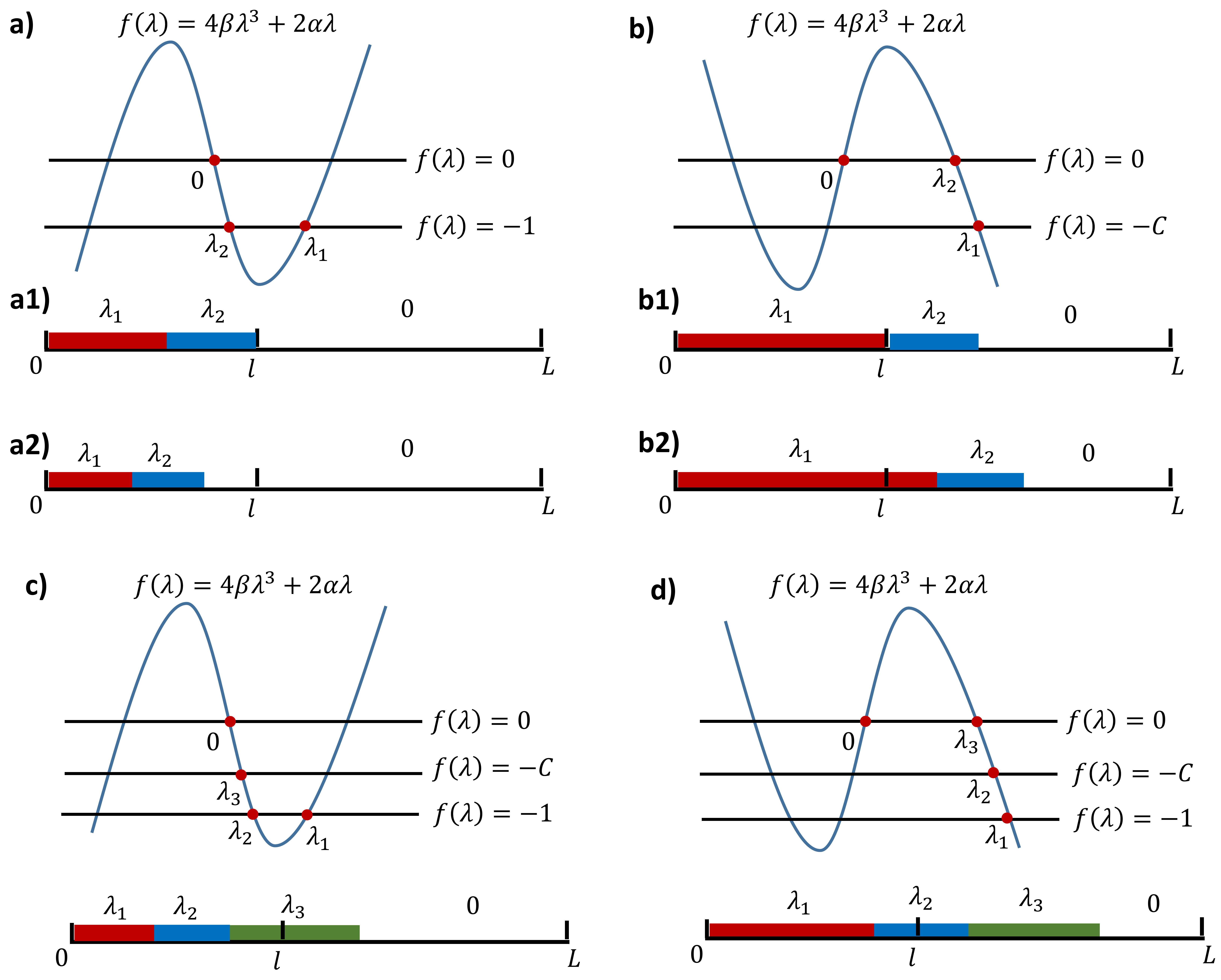}
    \caption{Different cases in entanglement structure optimization problem.}
    \label{fig:gmecase}
\end{figure}

1).  $\lambda_l\neq \lambda_{l+1}$.

For $i\leq l,$ $\lambda_i$ is the root of the equation $4\beta \lambda^3+2\alpha \lambda =-1$.  For $i>l$,  $\lambda_i$ is the root of the equation $4\beta \lambda^3+2\alpha \lambda=0$ and obviously, $0$ is one root. For the convenience of discussion, we draw the distribution diagrams of roots Fig.~\ref{fig:gmecase}a1) and Fig.~\ref{fig:gmecase}a2), and there are two possible cases, depending on the sign of $\alpha$ and $\beta$. For the first case, as shown in  Fig~.\ref{fig:gmecase}a1), $\lambda_i$ must equal to 0 for all $i>l$ because the possible nonzero solution is greater than $\lambda_1$ and $\lambda_2$. Hence, the original optimization problem becomes a much simpler minimization problem,
\begin{equation}
\begin{split}
\min \ \ &f=q_1\lambda_1+q_2\lambda_2 \\
\text{subject to} \ \ &q_1\lambda_1^2+q_2\lambda_2^2=M_2\\
&q_1\lambda_1^4+q_2\lambda_2^4 =M_4,\\
&\lambda_1\ge \lambda_2> 0,\quad q_1,q_2\in\mathbb{N}, \quad q_1+q_2= l.
\end{split}
\end{equation}
For the second case, as shown in Fig~.\ref{fig:gmecase}a2), $\lambda_i$ only has one possible value $\lambda_1$ for all $i\le l$, while has a possible positive value $\lambda_2$ for $i>l$. Hence, the original optimization problem becomes,
\begin{equation}
\begin{split}
\min \ \ &f=l\lambda_1+q_2\lambda_2 \\
\text{subject to} \ \ &l\lambda_1^2+q_2\lambda_2^2=M_2\\
&l\lambda_1^4+q_2\lambda_2^4 =M_4,\\
&\lambda_1\ge \lambda_2> 0,\quad q_2\in\mathbb{N},\\
&q_2\leq L-l.
\end{split}
\end{equation}
In fact, it is possible to find the analytical solutions for these two cases, by the ordinary derivation process. However, in the following cases, finding analytical solutions is not likely to be done.

2).  $\lambda_l= \lambda_{l+1}$ and there exists an positive integer $u$ such that $u<l$, $\lambda_u\neq \lambda_{u+1}$, $\lambda_{u+1}=\lambda_{u+2}=\cdots=\lambda_L=0$.
Then $\lambda_1,\cdots,\lambda_u$ are roots of the equation $4\beta \lambda^3+2\alpha \lambda+1 =0$, the distribution is shown in Fig.~\ref{fig:gmecase}a2) which is quite similar to the one shown in Fig.~\ref{fig:gmecase}a1), and can be unified with that case to get a more general minimization problem:
\begin{equation}\label{eq:gmecase1}
\begin{split}
\min \ \ &f=q_1\lambda_1+q_2\lambda_2 \\
\text{subject to} \ \ &q_1\lambda_1^2+q_2\lambda_2^2=M_2\\
&q_1\lambda_1^4+q_2\lambda_2^4 =M_4,\\
&\lambda_1\ge \lambda_2> 0,\quad q_1,q_2\in\mathbb{N},\\
&q_1+q_2\le l.
\end{split}
\end{equation}

3). If  $\lambda_1=\cdots=\lambda_{l+1}=\cdots=\lambda_q\neq\lambda_{q+1}$ and $q\ge l+1$, then $\lambda_1,\cdots,\lambda_{q}$ are roots of the equation $4\beta \lambda^3+2\alpha \lambda+\frac{l}{q}=0$ and $\lambda_{q+1},\cdots,\lambda_L$ are roots of the equation $4\beta \lambda^3+2\alpha \lambda=0$. As shown in Fig~.\ref{fig:gmecase}b2), there is only one possible case.  This case can be unified with the case shown in Fig.~\ref{fig:gmecase}b1), to give a more general optimization problem:
\begin{equation}\label{eq:gmecase2}
\begin{split}
\min \ \ &f=q_1\lambda_1+q_2\lambda_2 \\
\text{subject to} \ \ &q_1\lambda_1^2+q_2\lambda_2^2=M_2\\
&q_1\lambda_1^4+q_2\lambda_2^4 =M_4,\\
&\lambda_1\ge \lambda_2> 0,\quad,q_1,q_2\in\mathbb{N},\\
&l\le q_1,\quad q_1+q_2\le L.
\end{split}
\end{equation}

4).  $\lambda_l= \lambda_{l+1}$ and there exists integers $u$ and $v$ satisfying $0<u<l$ and $v>l+1$, such that $\lambda_{u}\neq \lambda_{u+1}=\cdots=\lambda_l=\lambda_{l+1}=\cdots=\lambda_{v}\neq\lambda_{v+1}=0$. Then $\lambda_1,\cdots,\lambda_{u}$ are roots of the equation $4\beta \lambda^3+2\alpha \lambda =-1$ and $\lambda_{u+1},\cdots,\lambda_{v}$ are roots of the equation $4\beta \lambda^3+2\alpha \lambda=-\frac{l-u}{v-u}$. The last terms $\lambda_{v+1},\cdots,\lambda_L$ are roots of $4\beta \lambda^3+2\alpha \lambda=0$. As shown in Fig~.\ref{fig:gmecase}c) and Fig.~\ref{fig:gmecase}d), there are also two cases dependent on the sign of $\alpha$ and $\beta$. Now we discuss them in details.

For the first case, as shown in Fig~.\ref{fig:gmecase}c), $\lambda_i$ has two possible values for $i\le u$ and one possible value for $u<i\le v$. In this scenario, the optimization problem becomes
\begin{equation}
\begin{split}
\min \ \ &f=q_1\lambda_1+q_2\lambda_2 +q_3\lambda_3\\
\text{subject to} \ \ &q_1\lambda_1^2+q_2\lambda_2^2+q_3\lambda_3^2=M_2\\
&q_1\lambda_1^4+q_2\lambda_2^4 +q_3\lambda_3^4=M_4,\\
&4\beta \lambda_1^3+2\alpha \lambda_1+1=0,\\
&4\beta \lambda_2^3+2\alpha \lambda_2+1=0,\\
&4\beta \lambda_3^3+2\alpha \lambda_3+\frac{l-q_1-q_2}{q_3}=0,\\
&\lambda_1\ge \lambda_2\geq\lambda_3 >0,\quad q_1,q_2,q_3\in\mathbb{N},\\
&\quad q_1+q_2<l,\quad l<q_1+q_2+q_3\le L.
\end{split}
\end{equation}
This optimization problem can be further simplified to 
\begin{equation}\label{eq:gmecase3}
\begin{split}
\min \ \ &f=q_1\lambda_1+q_2\lambda_2 +q_3\lambda_3\\
\text{subject to} \ \ &q_1\lambda_1^2+q_2\lambda_2^2+q_3\lambda_3^2=M_2\\
&q_1\lambda_1^4+q_2\lambda_2^4 +q_3\lambda_3^4=M_4,\\
&\frac{\lambda_3^3}{\lambda_1\lambda_2(\lambda_1+\lambda_2)}-\frac{(\lambda_1^3-\lambda_2^3)\lambda_3}{\lambda_1\lambda_2(\lambda_1^2-\lambda_2^2)}+\frac{l-q_1-q_2}{q_3}=0,\\
&\lambda_1\ge \lambda_2\geq\lambda_3 >0,\quad q_1,q_2,q_3\in\mathbb{N},\\
&\quad q_1+q_2<l,\quad l<q_1+q_2+q_3\le L.
\end{split}
\end{equation}
For the second case, as seen in Fig~.\ref{fig:gmecase}d), $\lambda_i$ has one possible value in each of the three intervals, $1\le i\le u$, $u+1\le i\le v$ and $v+1\le i\le L$. Then the optimization problem becomes
\begin{equation}\label{eq:gmecase4}
\begin{split}
\min \ \ &f=q_1\lambda_1+q_2\lambda_2 +q_3\lambda_3\\
\text{subject to} \ \ &q_1\lambda_1^2+q_2\lambda_2^2+q_3\lambda_3^2=M_2\\
&q_1\lambda_1^4+q_2\lambda_2^4 +q_3\lambda_3^4=M_4,\\
&-\frac{\lambda_2^3}{\lambda_1(\lambda_1^2-\lambda_3^2)}+\frac{\lambda_2\lambda_3^2}{\lambda_1(\lambda_1^2-\lambda_3^2)}+\frac{l-q_1}{q_2}=0,\\
&\lambda_1\ge \lambda_2\geq\lambda_3 >0,\quad q_1,q_2,q_3\in\mathbb{N},\\
& q_1<l , l< q_1+q_2<L, q_1+q_2+q_3\le L.
\end{split}
\end{equation}

Therefore, to solve the entanglement structure optimization problem, Eq.~\eqref{eq:gmeopt}, what we need to do is to find the minimum result of the four much easier optimization problems, Eq.~\eqref{eq:gmecase1}, Eq.~\eqref{eq:gmecase2}, Eq.~\eqref{eq:gmecase3} and Eq.~\eqref{eq:gmecase4}. Although it is also hard to find the analytical result, the complexity of this problem has been greatly reduced.

\section{Measurement Protocols of Permutation Moments and The Statistical Analysis}\label{sec:protocols}
\subsection{Shadow Estimation and Statistical Analysis}\label{subsec:shadow}
Shadow tomography \cite{huang2020predicting} is a systematic framework which can help us to benchmark the property of an unknown quantum system with only partial knowledge. Compared with former protocols, shadow tomography is much more resource-saving than full tomography and it does not need joint operations among multiple copies to estimate nonlinear functions.  Shadow tomography consists of two phases: quantum measurement and data postprocessing. In quantum measurement phase, we adopt $U$ to rotate $N$-qubit state $\rho$ to $U\rho U^\dagger$, and measure it in computational basis $\{\ket{\vec{s}} =\ket{s_1,\cdots,s_N}\}$. The measurement results and the rotation unitary together construct unbiased estimators of $\rho$:
\begin{equation}\label{eq:global shadow}
\begin{split}
\hat{\rho}=(2^N+1)U^\dagger\ketbra{\vec{s}}{\vec{s}}U-\mathbb{I}_{2^N}
\end{split}
\end{equation}
for $U$ randomly chosen from $N$-qubit Clifford group, and
\begin{equation}\label{eq:local shadow}
\begin{split}
\hat{\rho}=\bigotimes_{i=1}^N\left(3u_i^\dagger\ketbra{s_i}{s_i}u_i-\mathbb{I}_2\right)
\end{split}
\end{equation}
for $U=\bigotimes_{i=1}^Nu_i$ with each $u_i$ randomly chosen from single qubit Clifford group.

Recalling that $\mathcal{R}_\pi$ is just the index rearrangement of $\rho$, so shadow tomography can also construct an unbiased estimator of it
\begin{equation}
\begin{split}
\hat{\mathcal{R}}_\pi=\mathcal{R}_\pi(\hat{\rho}), \mathbb{E}[\hat{\mathcal{R}}_\pi]=\mathcal{R}_\pi(\mathbb{E}\hat{\rho})=\mathcal{R}_\pi.
\end{split}
\end{equation}
Hence, the unbiased estimator for $M_{2n}^\pi$ can be constructed using these $\hat{\mathcal{R}}_\pi$
\begin{equation}
\begin{split}
\hat{M}_{2n}^\pi=\frac{(M-2n)!}{M!}\sum_{\{i_1,\cdots, i_{2n}\}\subset[M]}\sum_{\sigma\in\mathcal{S}_{2n}}\tr[\hat{\mathcal{R}}_{\pi\sigma(i_1)}\cdots\hat{\mathcal{R}}_{\pi\sigma(i_{2n})}^\dagger],
\end{split}
\end{equation}
where $M$ is the number of the $\hat{\mathcal{R}}_\pi$ we prepared using shadow tomography, $\mathcal{S}_{2n}$ is the 2n-order permutation group, and $\hat{\mathcal{R}}_{\pi\sigma(i)}$ is the $\sigma(i)$-th copy of $\hat{\mathcal{R}}_\pi$. The unbiaseness of $\hat{M}_{2n}$ is easy to be proved:
\begin{equation}
\begin{split}
\mbb{E}\left(\hat{M}_{2n}^\pi\right)&=\frac{(M-2n)!}{M!}\sum_{\{i_1,\cdots, i_{2n}\}\subset[M]}\sum_{\sigma\in\mathcal{S}_{2n}}\tr[\mbb{E}\left(\hat{\mathcal{R}}_{\pi\sigma(i_1)}\right)\cdots\mbb{E}\left(\hat{\mathcal{R}}_{\pi\sigma(i_{2n})}^\dagger\right)]\\
&=\frac{(M-2n)!}{M!}\sum_{\{i_1,\cdots, i_{2n}\}\subset[M]}\sum_{\sigma\in\mathcal{S}_{2n}}\tr\left[\mathcal{R}_\pi\cdots\mathcal{R}_\pi^\dagger\right]\\
&=\frac{(M-2n)!}{M!}\times\frac{M!}{(M-2n)!}\tr\left[\left(\mathcal{R}_\pi\mathcal{R}_\pi^\dagger\right)^n\right]=M_{2n}^\pi.
\end{split}
\end{equation}

In fact, there exists many other protocols to measure nonlinear functions of state, like randomized measurements\cite{brydges2019probing,elben2019toolbox}, which has been proved to be more efficient than shadow estimation in some cases \cite{elben2020mixed,zhou2020Single,liu2021characterizing}. However, unlike shadow estimation, randomized measurements cannot always be easily conducted by qubit level operation, especially for observables like higher order permutation operators. Therefore, except for some special cases we will discuss in Sec.~\ref{subsec:random}, shadow estimation is a practical choice to measure the permutation moments.

In the following context, we take $M_4^{(2,3)}(\rho_{AB})$ as an example to illustrate the variance of different protocols, which is closely related to the bipartite entanglement criterion proposed in this work, and the conclusion can be easily generalized to other permutation moments. To analyze the error scaling of estimators constructed by shadow estimation, here we adopt the conclusions in \cite{huang2020predicting}:
\begin{fact}\label{fact:shadow variance}
The variance of predicting $\hat{O}=\tr(O\hat{\rho})$ on locally constructed $\hat{\rho}$, Eq.~\eqref{eq:local shadow}, scales like 
\begin{equation}
\begin{split}
\mathrm{Var}(\hat{O})\leq 2^{\mathrm{locality}(\hat{O})}\tr(\hat{O}^2).
\end{split}
\end{equation}
For globally constructed $\hat{\rho}$, Eq.~\eqref{eq:global shadow}, the variance scales like
\begin{equation}
\begin{split}
\mathrm{Var}(\hat{O})\leq 3\tr(\hat{O}^2).
\end{split}
\end{equation}
\end{fact}

Given $M$ measurement results of shadow estimation, the unbiased estimator we constructed to predict $M_4^{(2,3)}$ is
\begin{equation}\label{eq:M4def}
\begin{split}
\hat{M}_4^{(2,3)}=\frac{4(M-4)!}{M!}\sum_{\{i,j,k,l\}\subset[M]}\sum_{\sigma\in\mathcal{S}_3}\tr\left[\hat{R}_{i}\hat{R}_{\sigma(j)}^\dagger\hat{R}_{\sigma(k)}\hat{R}_{\sigma(l)}^\dagger\right],
\end{split}
\end{equation}
where for simplicity, we use $\hat{R}$ to represent $\hat{\mathcal{R}}_{(2,3)}(\rho_{AB})$. Here, because of the properties of trace function that it is invariant under permutation and taking Hermite conjugate, $\tr[\hat{R}_j\hat{R}_i^\dagger\hat{R}_k\hat{R}_l^\dagger]=\tr[\hat{R}_k\hat{R}_l^\dagger\hat{R}_j\hat{R}_i^\dagger]=\tr[\hat{R}_i\hat{R}_j^\dagger\hat{R}_l\hat{R}_k^\dagger]$, we can always set $\hat{R}_i$ in the first place. Combining Fact.~\ref{fact:shadow variance} and the definition Eq.~\eqref{eq:M4def}, one can prove that:

\begin{proposition}\label{proposition:M4shadowvariance}
According to Chebyshev's equation, to make sure our estimation of $M_4$ is accurately enough, which is to say, the estimator defined in Eq.\eqref{eq:M4def} satisfies $|\hat{M}_4^{(2,3)}-M_4^{(2,3)}|\le\epsilon$ with probability at least $1-\delta$, the number of snapshots needed is
\begin{equation}
\begin{split}
M=O\left(\frac{D}{\epsilon^{1/2}\delta^{1/4}}\right)
\end{split}
\end{equation}
for global shadow protocol, and 
\begin{equation}
\begin{split}
M=O\left(\frac{D^2}{\epsilon^{1/2}\delta^{1/4}}\right)
\end{split}
\end{equation}
for local shadow protocol, where $D$ stands for the system dimension. 
\end{proposition}

\begin{proof}
Following the definition of variance, we have
\begin{equation}
\begin{split}
\mathrm{Var}(\hat{M}_4^{(2,3)})=\mathbb{E}\left[\left(\hat{M}_4^{(2,3)}\right)^2\right]-\mathbb{E}\left[\hat{M}_4^{(2,3)}\right]^2=\mathbb{E}\left[\left(\hat{M}_4^{(2,3)}\right)^2\right]-\left(M_4^{(2,3)}\right)^2.
\end{split}
\end{equation}
substituting Eq.~\eqref{eq:M4def}, we get
\begin{equation}
\begin{split}
\left(\hat{M}_4^{(2,3)}\right)^2=\left(\frac{4(M-4)!}{M!}\right)^2\sum_{\{i,j,k,l\}\subset[M]}\sum_{\{i',j',k',l'\}\subset[M]}\left\{\sum_{\sigma\in\mathcal{S}_3}\tr\left[\hat{R}_{i}\hat{R}_{\sigma(j)}^\dagger\hat{R}_{\sigma(k)}\hat{R}_{\sigma(l)}^\dagger\right]\right\}\left\{\sum_{\pi\in\mathcal{S}_3}\tr\left[\hat{R}_{i'}\hat{R}_{\pi(j')}^\dagger\hat{R}_{\pi(k')}\hat{R}_{\pi(l')}^\dagger\right]\right\}
\end{split}
\end{equation}
To benefit our calculation, we can divide the summation terms in the R.H.S of the last equation into five groups, according to the number of equal indices between $\{i,j,k,l\}$ and $\{i',j',k',l'\}$. If there are $0\leq c\leq 4$ pairs of same indices, the total number of such terms are $\binom{M}{4}\binom{4}{c}\binom{M-4}{4-c}$, hence
\begin{equation}
\begin{split}
\mathbb{E}\left[\left(\hat{M}_4^{(2,3)}\right)^2\right]&=\left(\frac{4(M-4)!}{M!}\right)^2\binom{M}{4}\binom{4}{0}\binom{M-4}{4}36\tr\left(RR^\dagger RR^\dagger\right)^2\\
&+\left(\frac{4(M-4)!}{M!}\right)^2\binom{M}{4}\binom{4}{1}\binom{M-4}{3}36\mathbb{E}\left[\tr\left(RR^\dagger R\hat{R}^\dagger\right)^2\right]\\
&+\left(\frac{4(M-4)!}{M!}\right)^2\binom{M}{4}\binom{4}{2}\binom{M-4}{2}4\mathbb{E}\left\{\left[\tr(RR^\dagger\hat{R}_1\hat{R}^\dagger_2)+\tr(RR^\dagger\hat{R}_2\hat{R}^\dagger_1)+\tr(R\hat{R}^\dagger_1R\hat{R}^\dagger_2)\right]^2\right\}\\
&+\left(\frac{4(M-4)!}{M!}\right)^2\binom{M}{4}\binom{4}{3}\binom{M-4}{1}\mathbb{E}\left\{\left[\tr(R\hat{R}^\dagger_1\hat{R}_2\hat{R}^\dagger_3)+\tr(R\hat{R}^\dagger_1\hat{R}_3\hat{R}^\dagger_2)+\tr(R\hat{R}^\dagger_2\hat{R}_1\hat{R}^\dagger_3)+\cdots\right]\right\}^2\\
&\left(\frac{4(M-4)!}{M!}\right)^2\binom{M}{4}\binom{4}{4}\binom{M-4}{0}\mathbb{E}\left\{\left[\tr(\hat{R}_1\hat{R}^\dagger_2\hat{R}_3\hat{R}^\dagger_4)+\tr(\hat{R}_1\hat{R}^\dagger_2\hat{R}_4\hat{R}^\dagger_3)+\tr(\hat{R}_1\hat{R}^\dagger_3\hat{R}_2\hat{R}^\dagger_4)+\cdots\right]^2\right\},
\end{split}
\end{equation}
where $R=\mathcal{R}_{(2,3)}(\rho_{AB})$, and $\hat{R}_1$, $\hat{R}_2$, $\hat{R}_3$, and $\hat{R}_4$ are independent unbiased snapshots of $R$. Hence,
\begin{equation}\label{eq:M_4variance}
\begin{split}
\mathrm{Var}\left[\left(\hat{M}_4^{(2,3)}\right)\right]&=36\left(\frac{4(M-4)!}{M!}\right)^2\binom{M}{4}\binom{4}{1}\binom{M-4}{3}\mathrm{Var}\left[\tr(RR^\dagger R\hat{R}^\dagger)\right]\\
&+4\left(\frac{4(M-4)!}{M!}\right)^2\binom{M}{4}\binom{4}{2}\binom{M-4}{2}\mathrm{Var}\left[\tr(RR^\dagger\hat{R}_1\hat{R}^\dagger_2)+\tr(RR^\dagger\hat{R}_2\hat{R}^\dagger_1)+\tr(R\hat{R}^\dagger_1R\hat{R}^\dagger_2)\right]\\
&+\left(\frac{4(M-4)!}{M!}\right)^2\binom{M}{4}\binom{4}{3}\binom{M-4}{1}\mathrm{Var}\left[\tr(R\hat{R}^\dagger_1\hat{R}_2\hat{R}^\dagger_3)+\tr(R\hat{R}^\dagger_1\hat{R}_3\hat{R}^\dagger_2)+\tr(R\hat{R}^\dagger_2\hat{R}_1\hat{R}^\dagger_3)+\cdots\right]\\
&+\left(\frac{4(M-4)!}{M!}\right)^2\binom{M}{4}\binom{4}{4}\binom{M-4}{0}\mathrm{Var}\left[\tr(\hat{R}_1\hat{R}^\dagger_2\hat{R}_3\hat{R}^\dagger_4)+\tr(\hat{R}_1\hat{R}^\dagger_2\hat{R}_4\hat{R}^\dagger_3)+\tr(\hat{R}_1\hat{R}^\dagger_3\hat{R}_2\hat{R}^\dagger_4)+\cdots\right]\\
&\leq \frac{C_1}{M}\mathrm{Var}\left[\tr(RR^\dagger R\hat{R}^\dagger)\right]+\frac{C_2}{M^2}\mathrm{Var}\left[\tr(RR^\dagger\hat{R}_1\hat{R}^\dagger_2)+\tr(RR^\dagger\hat{R}_2\hat{R}^\dagger_1)+\tr(R\hat{R}^\dagger_1R\hat{R}^\dagger_2)\right]\\
&+\frac{C_3}{M^3}\mathrm{Var}\left[\tr(R\hat{R}^\dagger_1\hat{R}_2\hat{R}^\dagger_3)+\tr(R\hat{R}^\dagger_1\hat{R}_3\hat{R}^\dagger_2)+\tr(R\hat{R}^\dagger_2\hat{R}_1\hat{R}^\dagger_3)+\cdots\right]\\
&+\frac{C_4}{M^4}\mathrm{Var}\left[\tr(\hat{R}_1\hat{R}^\dagger_2\hat{R}_3\hat{R}^\dagger_4)+\tr(\hat{R}_1\hat{R}^\dagger_2\hat{R}_4\hat{R}^\dagger_3)+\tr(\hat{R}_1\hat{R}^\dagger_3\hat{R}_2\hat{R}^\dagger_4)+\cdots\right]\\
&=\frac{C_1}{M}\mathrm{Var}\left\{\tr[\hat{O}_1\hat{\rho}]\right\}+\frac{C_2}{M^2}\mathrm{Var}\left\{\tr[\hat{O}_2(\hat{\rho}_1\otimes\hat{\rho}_2)]\right\}+\frac{C_3}{M^3}\mathrm{Var}\left\{\tr[\hat{O}_3(\hat{\rho}_1\otimes\hat{\rho}_2\otimes\hat{\rho}_3)]\right\}\\
&+\frac{C_4}{M^4}\mathrm{Var}\left\{\tr[\hat{O}_4(\hat{\rho}_1\otimes\hat{\rho}_2\otimes\hat{\rho}_3\otimes\hat{\rho}_4)]\right\}.
\end{split}
\end{equation}
Here, $C_1$, $C_2$, $C_3$ and $C_4$ are constants independent of $M$. To apply Fact \ref{fact:shadow variance}, we need to figure out the exact form of $\hat{O}_1$, $\hat{O}_2$, $\hat{O}_3$ and $\hat{O}_4$. Denote $\hat{O}$ to be the target observable $\hat{O}=\mathbb{S}_A^{(1,2)}\otimes\mathbb{S}_A^{(3,4)}\otimes\mathbb{S}_B^{(2,3)}\otimes\mathbb{S}_B^{(4,1)}$, then the fourth term can be written as
\begin{equation}
\begin{split}
&\tr(\hat{R}_1\hat{R}_2^\dagger\hat{R}_3\hat{R}_4^\dagger)+\tr(\hat{R}_1\hat{R}_2^\dagger\hat{R}_4\hat{R}_3^\dagger)+\tr(\hat{R}_1\hat{R}_3^\dagger\hat{R}_2\hat{R}_4^\dagger)+\cdots\\
=&\tr[\hat{O}(\hat{\rho}_1\otimes\hat{\rho}_2\otimes\hat{\rho}_3\otimes\hat{\rho}_4)]+\tr[\hat{O}(\hat{\rho}_1\otimes\hat{\rho}_2\otimes\hat{\rho}_4\otimes\hat{\rho}_3)]+\tr[\hat{O}(\hat{\rho}_1\otimes\hat{\rho}_3\otimes\hat{\rho}_2\otimes\hat{\rho}_4)]+\cdots.
\end{split}
\end{equation}
To write these six terms into one, we need to define some permutation operators, $\hat{W}_1=\mathbb{I}^{\otimes 4}$, $\hat{W}_2=\mathbb{I}^{\otimes 2}\otimes\mathbb{S}^{(3,4)}$, $\hat{W}_3=\mathbb{I}^{\otimes 2}\otimes\mathbb{S}^{(2,3)}$, $\hat{W}_4=\mathbb{I}^{\otimes 2}\otimes\mathbb{S}^{(2,4)}$, $\hat{W}_5=\mathbb{I}\otimes\overrightarrow{\Pi}^{(2,3,4)}$ and $\hat{W}_6=\mathbb{I}\otimes\overleftarrow{\Pi}^{(2,3,4)}$. Hence
\begin{equation}
\begin{split}
&\tr[\hat{O}(\hat{\rho}_1\otimes\hat{\rho}_2\otimes\hat{\rho}_3\otimes\hat{\rho}_4)]+\tr[\hat{O}(\hat{\rho}_1\otimes\hat{\rho}_2\otimes\hat{\rho}_4\otimes\hat{\rho}_3)]+\tr[\hat{O}(\hat{\rho}_1\otimes\hat{\rho}_3\otimes\hat{\rho}_2\otimes\hat{\rho}_4)]+\cdots\\
=&\tr[\hat{O}\hat{W}_1(\hat{\rho}_1\otimes\hat{\rho}_2\otimes\hat{\rho}_3\otimes\hat{\rho}_4)\hat{W}_1^\dagger]+\tr[\hat{O}\hat{W}_2(\hat{\rho}_1\otimes\hat{\rho}_2\otimes\hat{\rho}_3\otimes\hat{\rho}_4)\hat{W}_2^\dagger]+\tr[\hat{O}\hat{W}_3(\hat{\rho}_1\otimes\hat{\rho}_2\otimes\hat{\rho}_3\otimes\hat{\rho}_4)\hat{W}_3^\dagger]+\cdots\\
=&\tr[(\hat{W}_1^\dagger\hat{O}\hat{W}_1+\hat{W}_2^\dagger\hat{O}\hat{W}_2+\hat{W}_3^\dagger\hat{O}\hat{W}_3+\cdots)(\hat{\rho}_1\otimes\hat{\rho}_2\otimes\hat{\rho}_3\otimes\hat{\rho}_4)]\\
=&\tr[\hat{O}_4(\hat{\rho}_1\otimes\hat{\rho}_2\otimes\hat{\rho}_3\otimes\hat{\rho}_4)],
\end{split}
\end{equation}
where $\hat{O}_4$ can be graphically represented as
\begin{eqnarray}
\hat{O}_4=
\begin{tabular}{c}
     \includegraphics[scale=0.08]{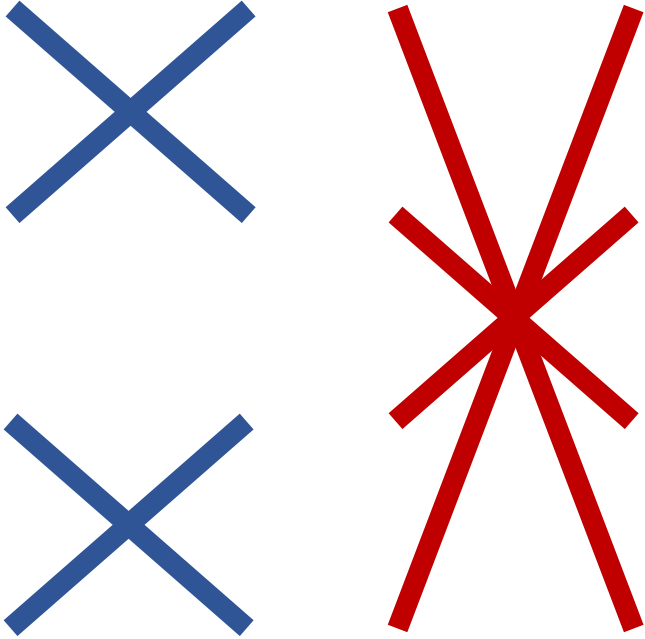} 
\end{tabular}
+
\begin{tabular}{c}
     \includegraphics[scale=0.08]{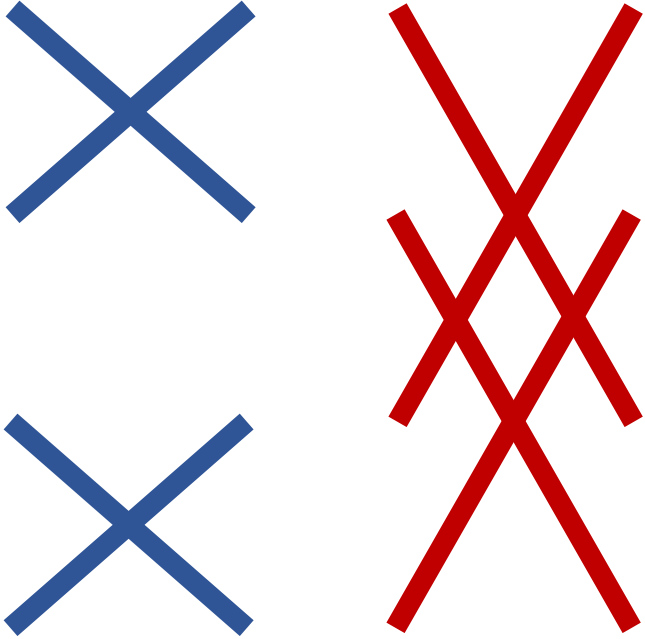} 
\end{tabular}
+
\begin{tabular}{c}
     \includegraphics[scale=0.08]{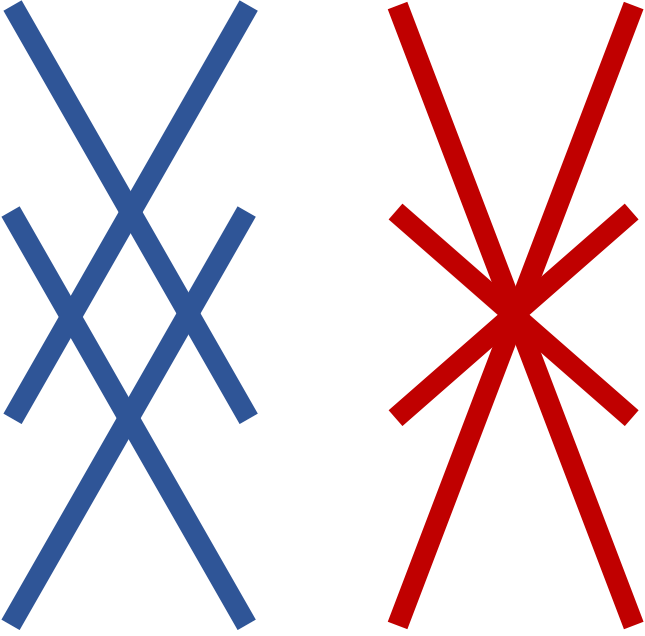} 
\end{tabular}
+
\begin{tabular}{c}
     \includegraphics[scale=0.08]{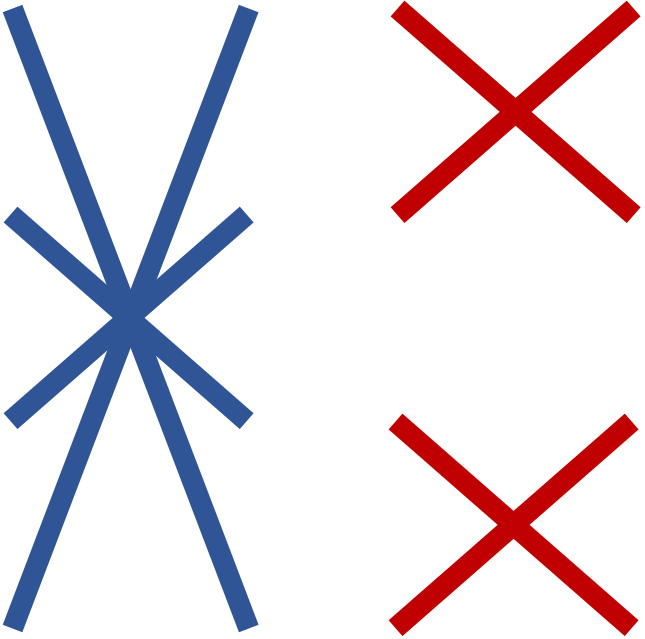} 
\end{tabular}
+
\begin{tabular}{c}
     \includegraphics[scale=0.08]{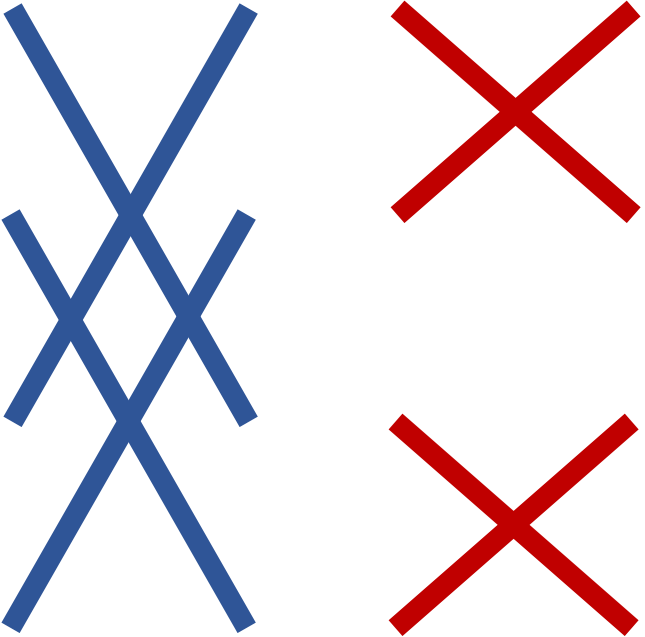} 
\end{tabular}
+
\begin{tabular}{c}
     \includegraphics[scale=0.08]{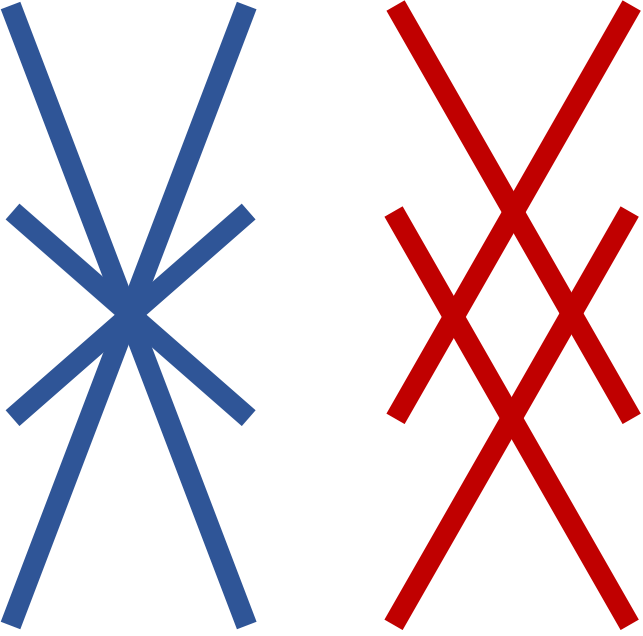} 
\end{tabular},
\end{eqnarray}
where the cross sign is the SWAP operator, and different colors stands for different parties, then
\begin{equation}
\begin{split}
\tr(\hat{O}_4^2)=\sum_{i,j=1}^6\tr[\hat{W}_i^\dagger\hat{O}\hat{W}_i\hat{W}_j^\dagger\hat{O}\hat{W}_j].
\end{split}
\end{equation}
We summarize the calculation in Table \ref{table:O_4calculation}. According to the mathematical property of operator SWAP and identity, $\tr(\mathbb{S})=d$ and $\tr(\mathbb{I}^{\otimes 2})=d^2$, it is easy to find that 
\begin{eqnarray}\label{eq:O_4square}
\tr[\hat{O}_4^2]\leq 36\tr[\hat{O}^2]=
36\tr\left[\left(
\begin{tabular}{c}
     \includegraphics[scale=0.08]{O_4_1.png}
\end{tabular}
\right)\left(
\begin{tabular}{c}
     \includegraphics[scale=0.08]{O_4_1.png}
\end{tabular}
\right)\right]
=
36\tr\left(
\begin{tabular}{c}
     \includegraphics[scale=0.08]{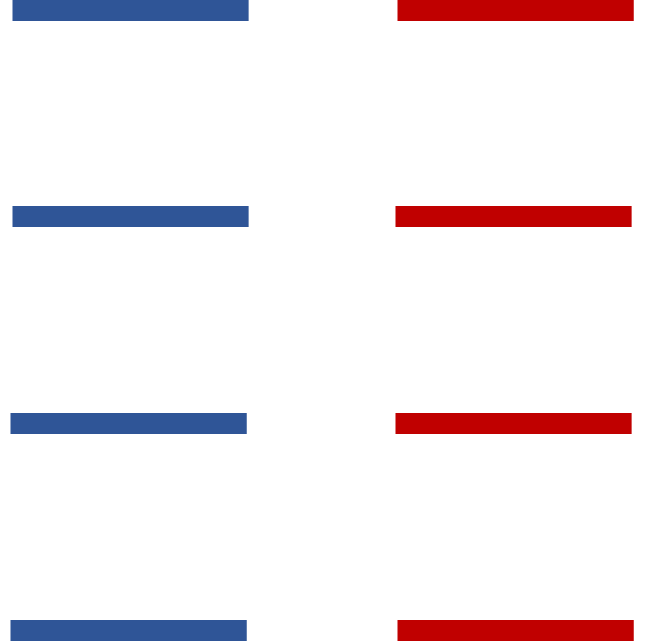}
\end{tabular}
\right)
=36d_A^4d_B^4=36D^4
\end{eqnarray}

\begin{table}
    \centering
    \setlength{\tabcolsep}{4mm}{
    \begin{tabular}{c||c|c|c|c|c|c}
    \hline\hline
         \rule{0pt}{28pt}
          & \includegraphics[scale=0.08]{O_4_1} & \includegraphics[scale=0.08]{O_4_2} &\includegraphics[scale=0.08]{O_4_3}&\includegraphics[scale=0.08]{O_4_4}&\includegraphics[scale=0.08]{O_4_5}&\includegraphics[scale=0.08]{O_4_6} \\
         \hline\hline
         \rule{0pt}{28pt}
         \includegraphics[scale=0.08]{O_4_1}&\includegraphics[scale=0.08]{O_4_01} &\includegraphics[scale=0.08]{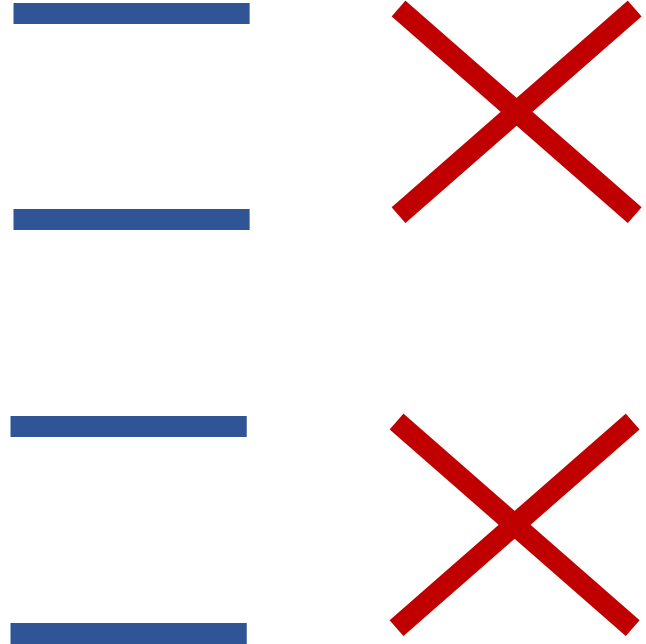} &\includegraphics[scale=0.08]{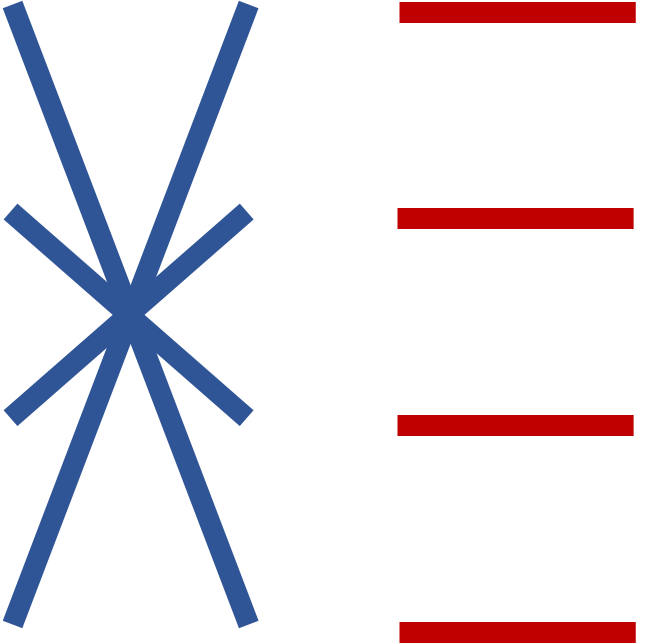} &\includegraphics[scale=0.08]{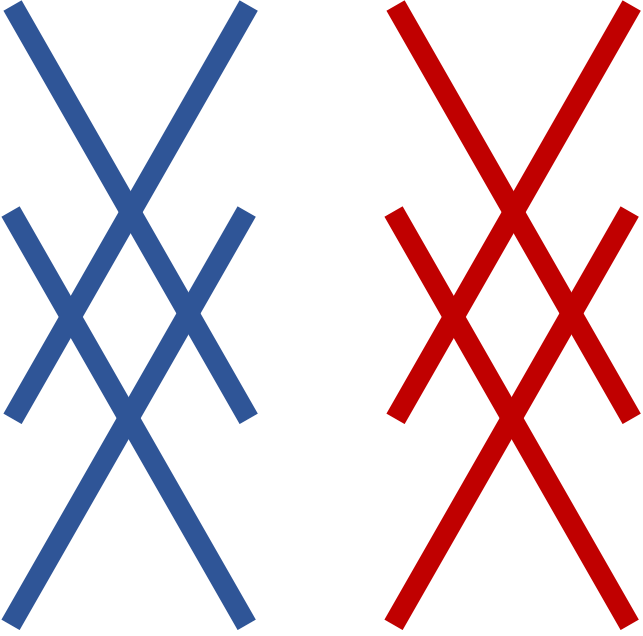} &\includegraphics[scale=0.08]{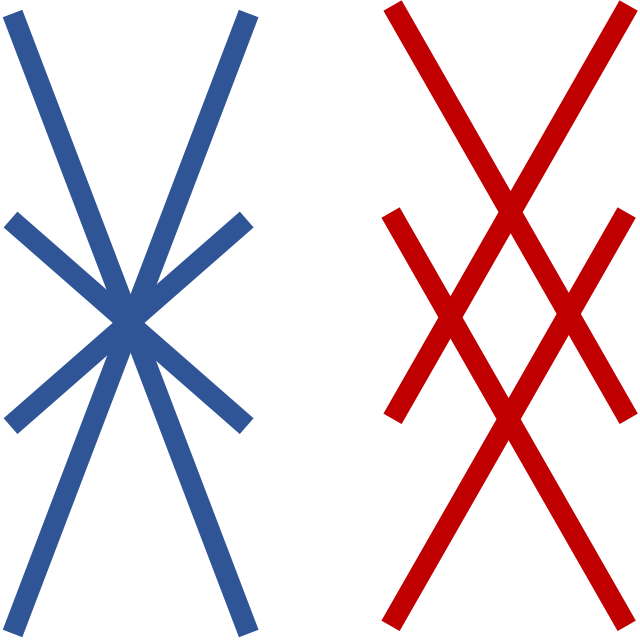} &\includegraphics[scale=0.08]{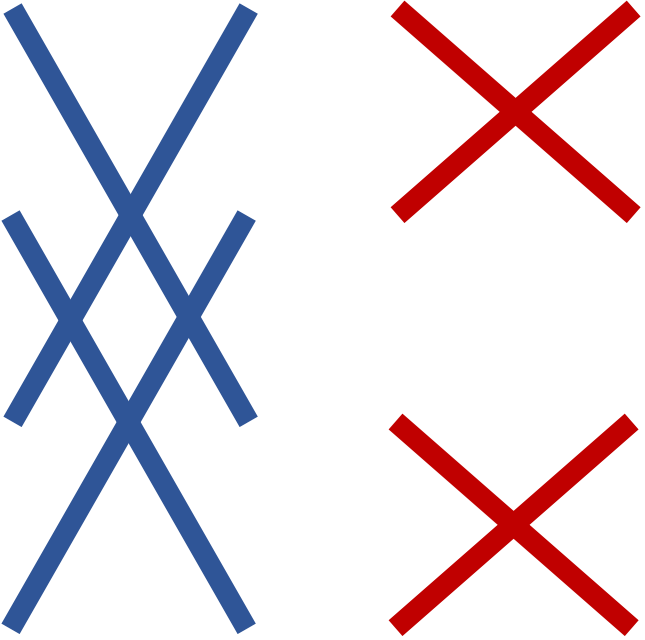} \\
         \hline
         \rule{0pt}{28pt}
         \includegraphics[scale=0.08]{O_4_2}&\includegraphics[scale=0.08]{O_4_02} &\includegraphics[scale=0.08]{O_4_01} &\includegraphics[scale=0.08]{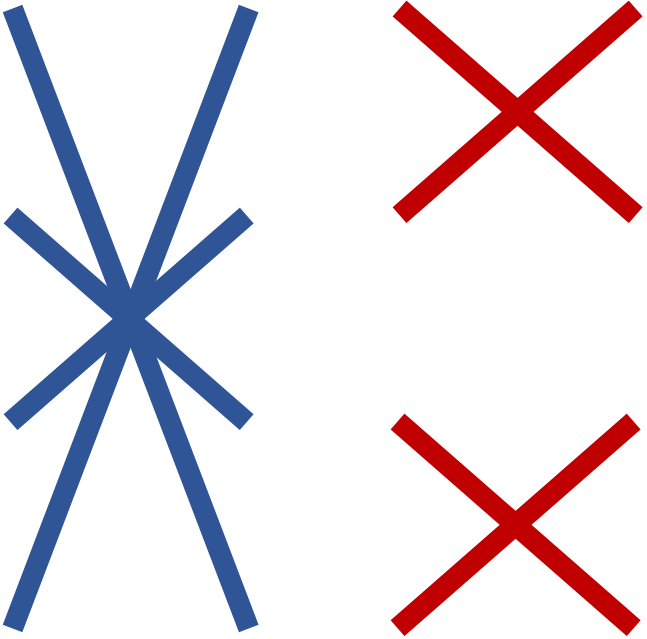} &\includegraphics[scale=0.08]{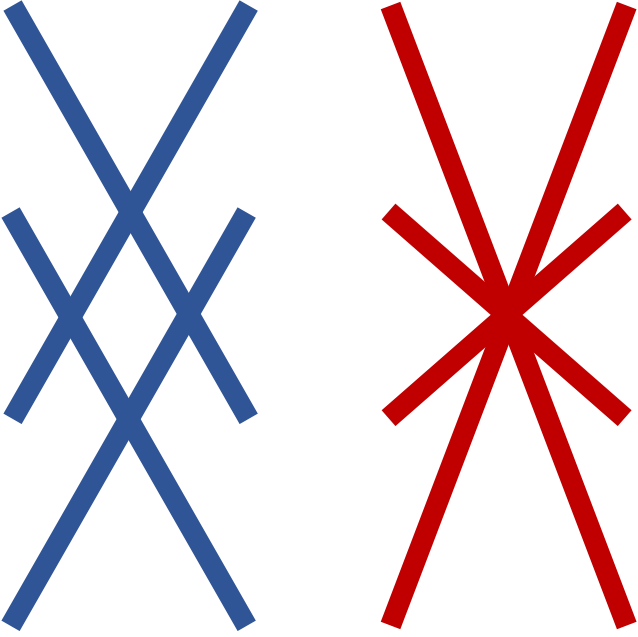} &\includegraphics[scale=0.08]{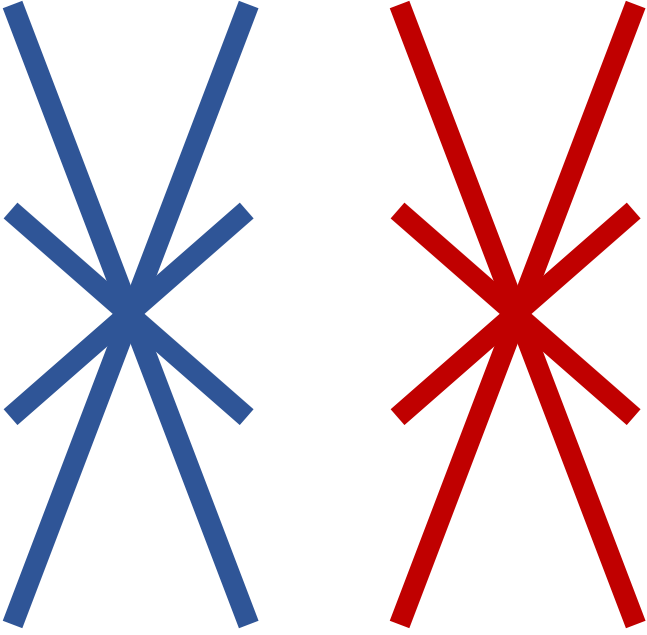} &\includegraphics[scale=0.08]{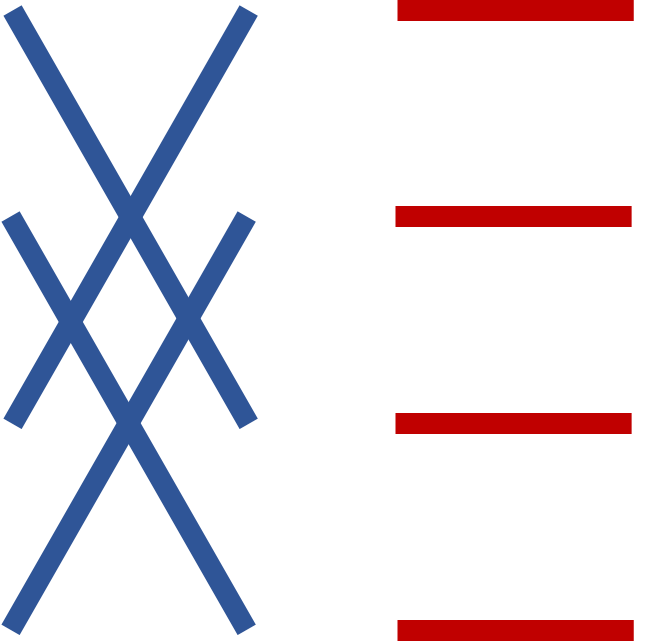} \\
         \hline
         \rule{0pt}{28pt}
         \includegraphics[scale=0.08]{O_4_3}&\includegraphics[scale=0.08]{O_4_03} &\includegraphics[scale=0.08]{O_4_11} &\includegraphics[scale=0.08]{O_4_01} &\includegraphics[scale=0.08]{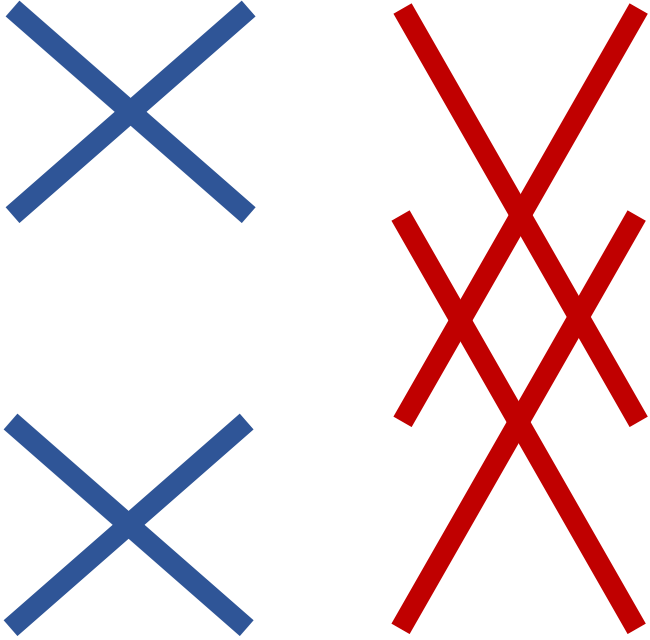} &\includegraphics[scale=0.08]{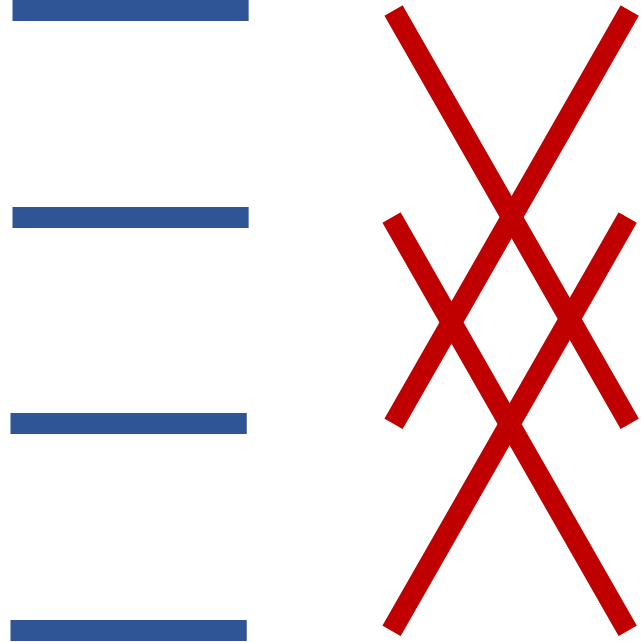} &\includegraphics[scale=0.08]{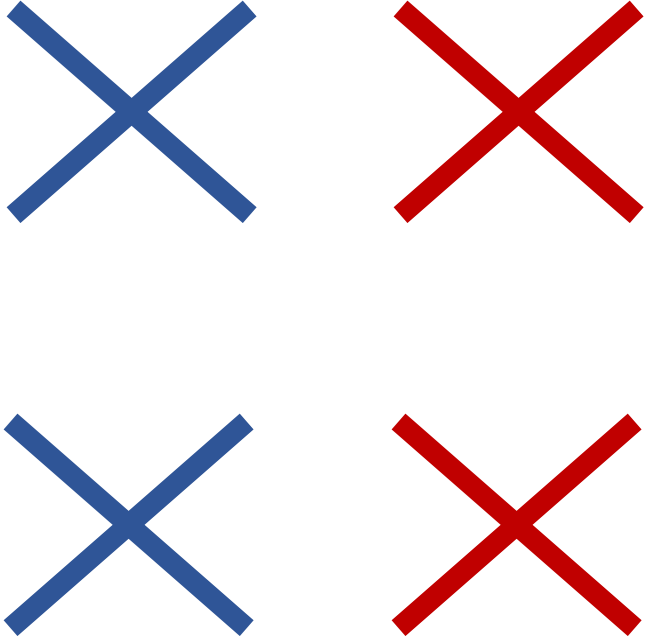} \\
         \hline
         \rule{0pt}{28pt}
         \includegraphics[scale=0.08]{O_4_4}&\includegraphics[scale=0.08]{O_4_04} &\includegraphics[scale=0.08]{O_4_12} &\includegraphics[scale=0.08]{O_4_21} &\includegraphics[scale=0.08]{O_4_01} &\includegraphics[scale=0.08]{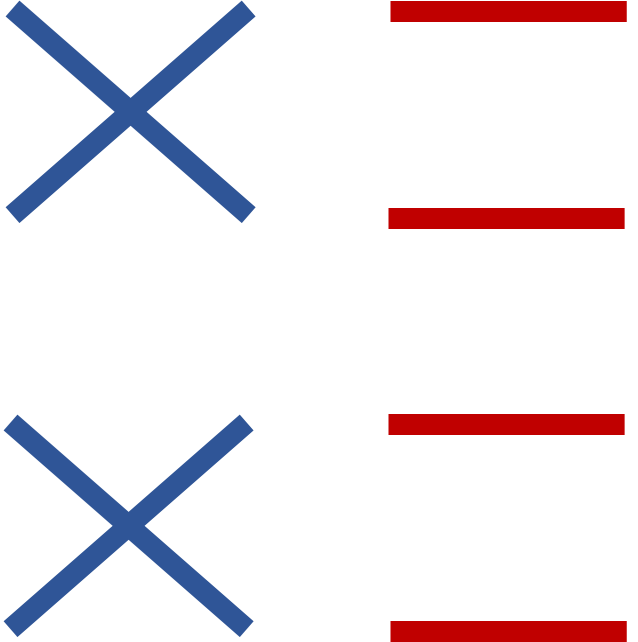} &\includegraphics[scale=0.08]{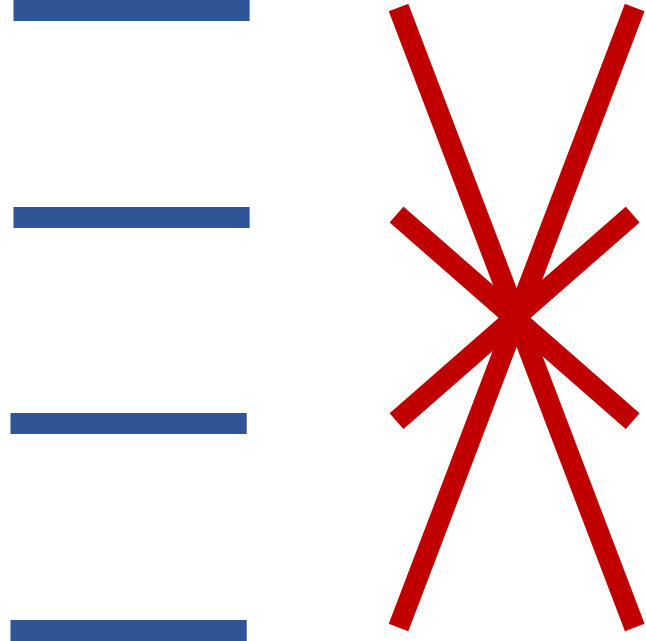} \\
         \hline
         \rule{0pt}{28pt}
         \includegraphics[scale=0.08]{O_4_5}&\includegraphics[scale=0.08]{O_4_05} &\includegraphics[scale=0.08]{O_4_13} &\includegraphics[scale=0.08]{O_4_22} &\includegraphics[scale=0.08]{O_4_31} &\includegraphics[scale=0.08]{O_4_01} &\includegraphics[scale=0.08]{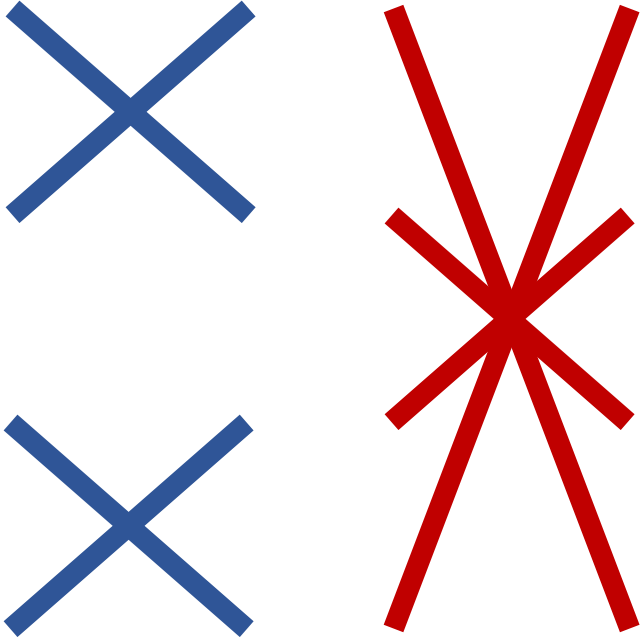} \\
         \hline
         \rule{0pt}{28pt}
         \includegraphics[scale=0.08]{O_4_6}&\includegraphics[scale=0.08]{O_4_06} &\includegraphics[scale=0.08]{O_4_14} &\includegraphics[scale=0.08]{O_4_23} &\includegraphics[scale=0.08]{O_4_32} &\includegraphics[scale=0.08]{O_4_41} &\includegraphics[scale=0.08]{O_4_01} \\
         \hline\hline
    \end{tabular}}
    \caption{ Multiplication of the terms in $\hat{O}_4^2$. In this table, we use blue and red colors to represent the two parties. ``x'' form signs stand for SWAP operators, and the horizontal lines stand for identity operators.  }
    \label{table:O_4calculation}
\end{table}

The other three operators $\hat{O}_1$, $\hat{O}_2$ and $\hat{O}_3$ can be constructed using $\hat{O}_4$. Start from the third term,
\begin{equation}
\begin{split}
&\tr(R\hat{R}^\dagger_1\hat{R}_2\hat{R}^\dagger_3)+\tr(R\hat{R}^\dagger_1\hat{R}_3\hat{R}^\dagger_2)+\tr(R\hat{R}^\dagger_2\hat{R}_1\hat{R}^\dagger_3)+\cdots\\
=&\tr[\hat{O}_4(\rho\otimes\hat{\rho}_1\otimes\hat{\rho}_2\otimes\hat{\rho}_3)]\\
=&\tr[\hat{O}_4(\rho\otimes\mathbb{I}^{\otimes 3})(\mathbb{I}\otimes\hat{\rho}_1\otimes\hat{\rho}_2\otimes\hat{\rho}_3)]\\
=&\tr\left\{\tr_1[\hat{O}_4(\rho\otimes\mathbb{I}^{\otimes 3})](\hat{\rho}_1\otimes\hat{\rho}_2\otimes\hat{\rho}_3)\right\}\\
=&\tr[\hat{O}_3(\hat{\rho}_1\otimes\hat{\rho}_2\otimes\hat{\rho}_3)],
\end{split}
\end{equation}
so that
\begin{equation}
\begin{split}
\hat{O}_3=\tr_1[\hat{O}_4(\rho\otimes\mathbb{I}^{\otimes 3})].
\end{split}
\end{equation}
In order to make the graphic representation more convenient, when analyzing $\hat{O}_3$ $\hat{O}_2$ and $\hat{O}_1$, we will treat $\rho$ \textbf{as if} it is the tensor product of two reduced density matrices, $\rho=\rho_A\otimes\rho_B$, so that we can use two separated boxes to represent it. But one needs to remember that it is actually not the case. Therefore,
\begin{eqnarray}
\hat{O}_3&=\begin{tabular}{c}
     \includegraphics[scale=0.1]{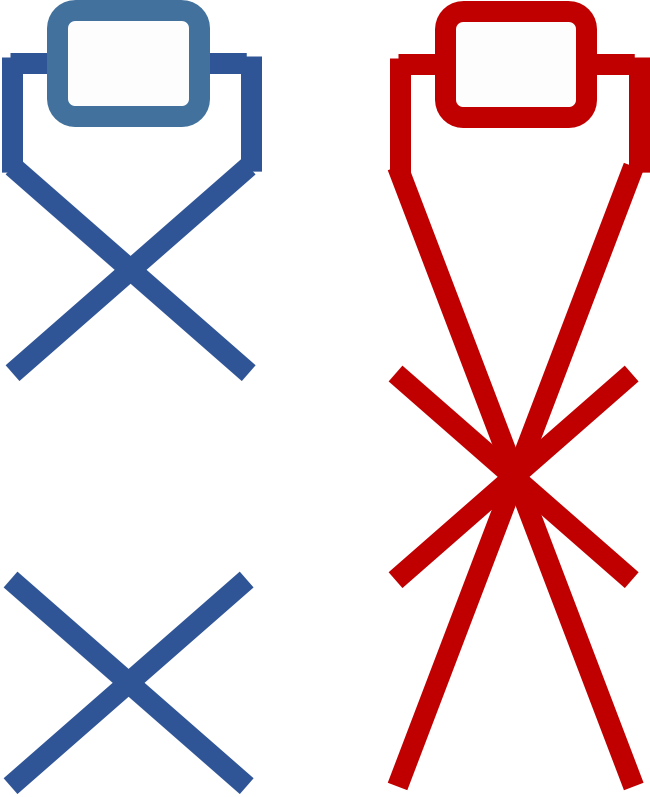} 
\end{tabular} 
+
\begin{tabular}{c}
     \includegraphics[scale=0.1]{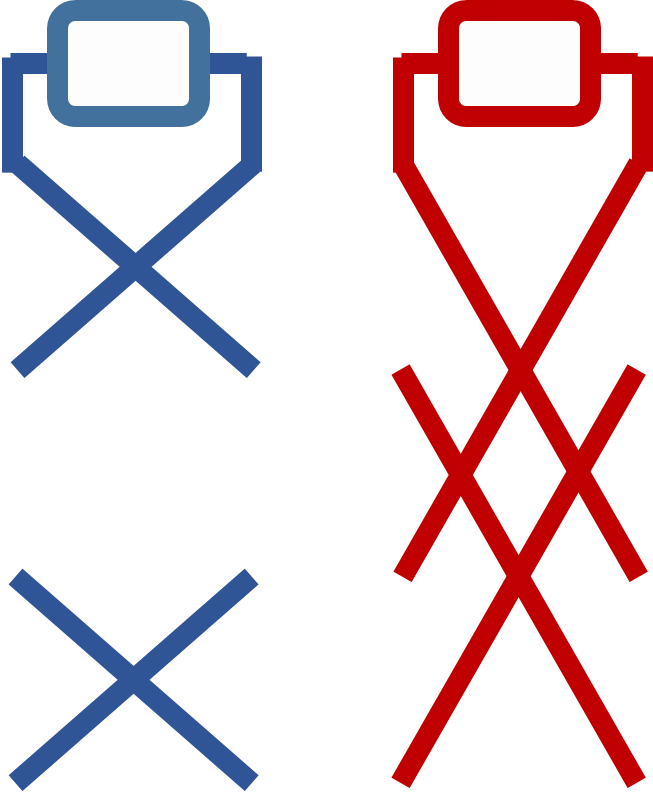} 
\end{tabular}
+
\begin{tabular}{c}
     \includegraphics[scale=0.1]{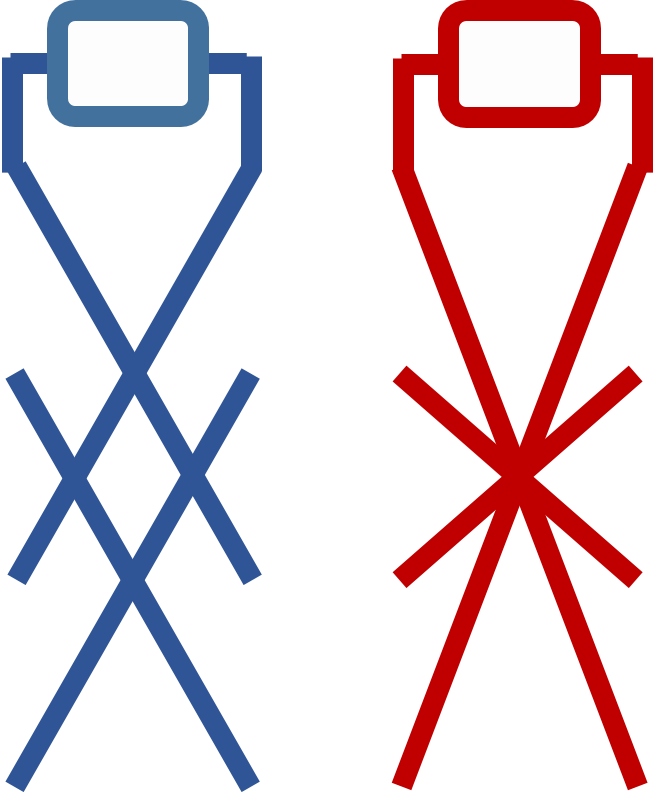} 
\end{tabular}
+
\begin{tabular}{c}
     \includegraphics[scale=0.1]{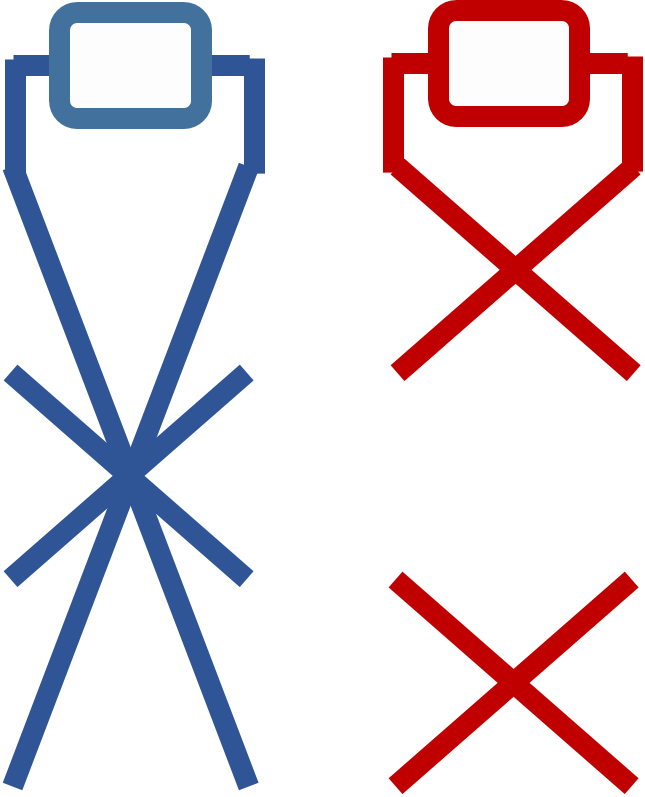} 
\end{tabular}
+
\begin{tabular}{c}
     \includegraphics[scale=0.1]{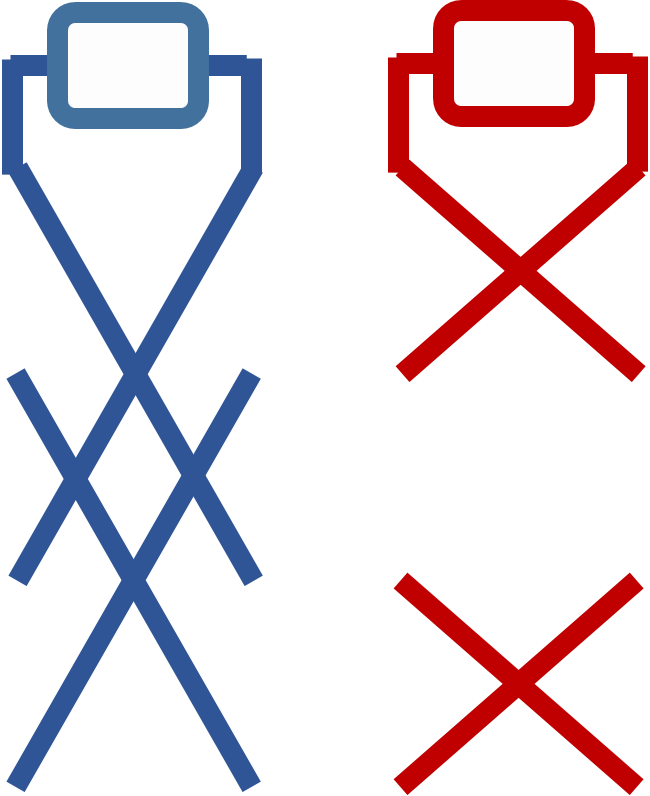} 
\end{tabular}
+
\begin{tabular}{c}
     \includegraphics[scale=0.1]{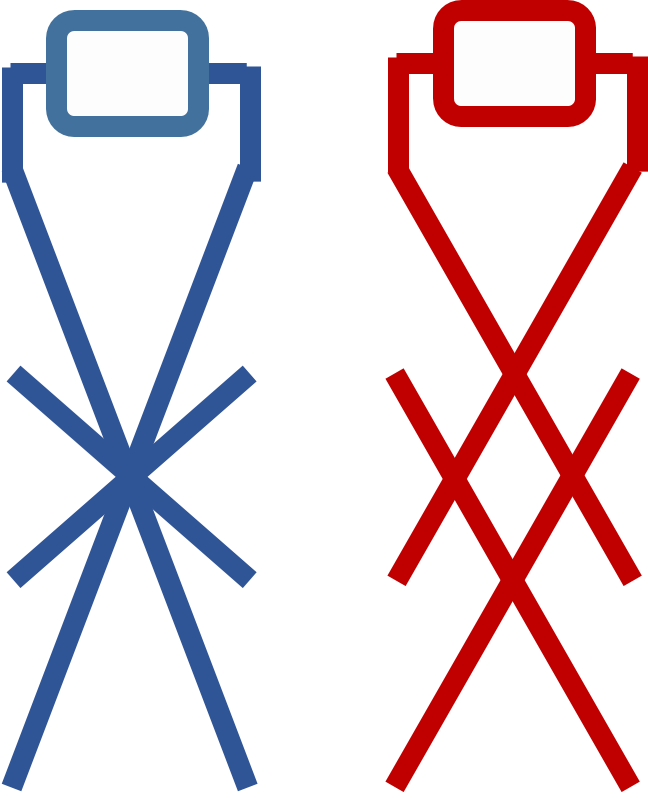} 
\end{tabular}
\\
&=
\begin{tabular}{c}
     \includegraphics[scale=0.1]{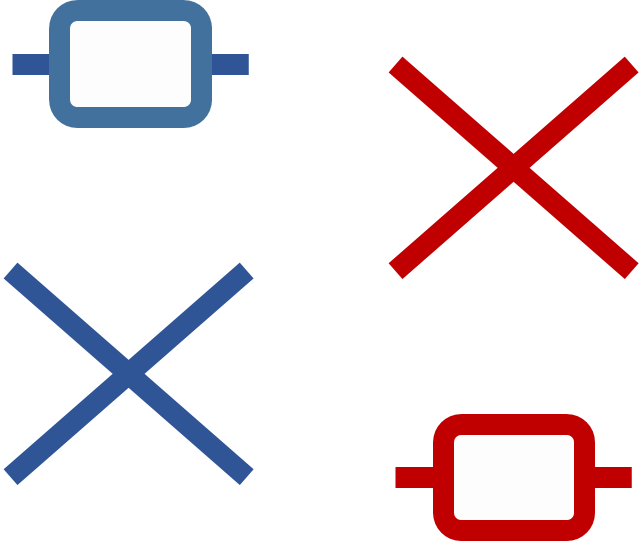} 
\end{tabular}
+
\begin{tabular}{c}
     \includegraphics[scale=0.1]{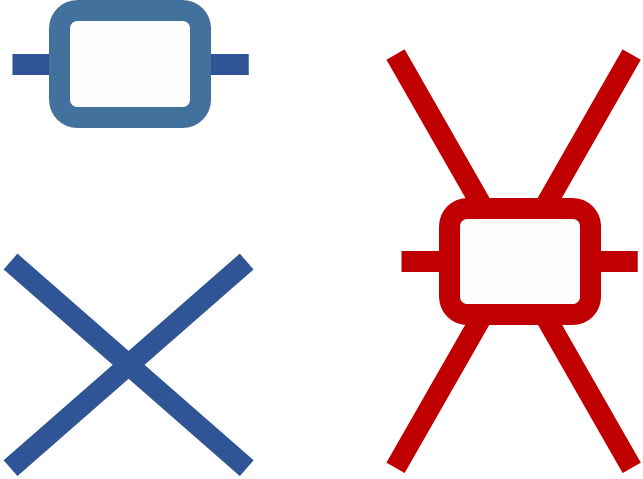} 
\end{tabular}
+
\begin{tabular}{c}
     \includegraphics[scale=0.1]{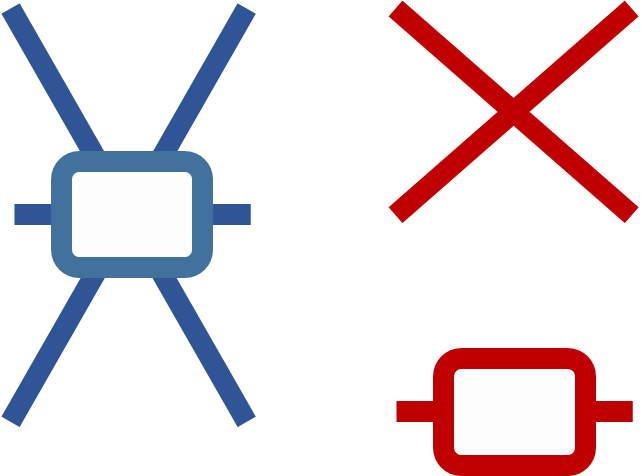} 
\end{tabular}
+
\begin{tabular}{c}
     \includegraphics[scale=0.1]{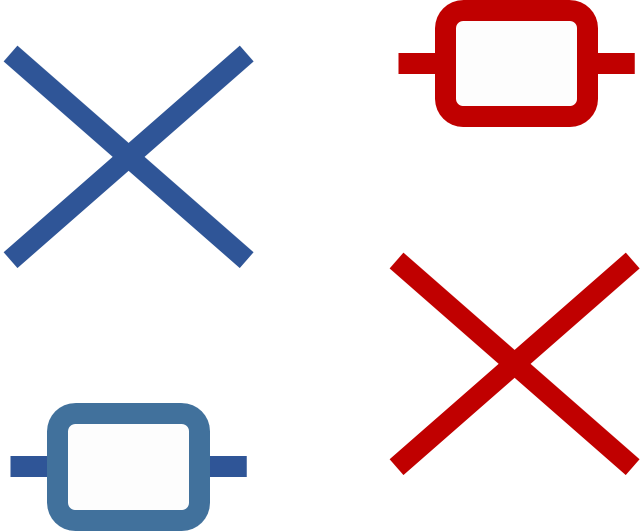} 
\end{tabular}
+
\begin{tabular}{c}
     \includegraphics[scale=0.1]{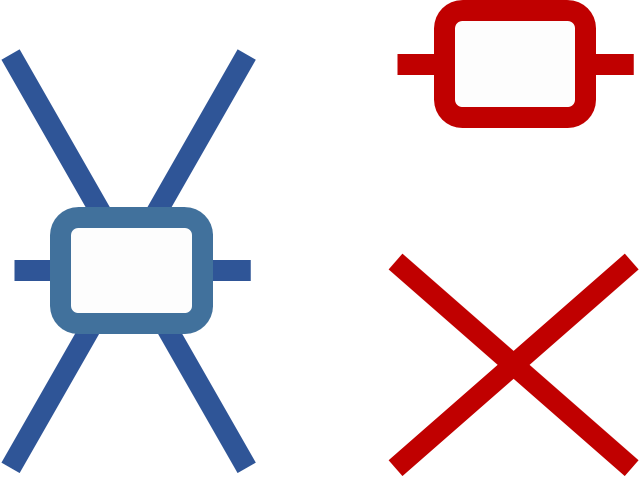} 
\end{tabular}
+
\begin{tabular}{c}
     \includegraphics[scale=0.1]{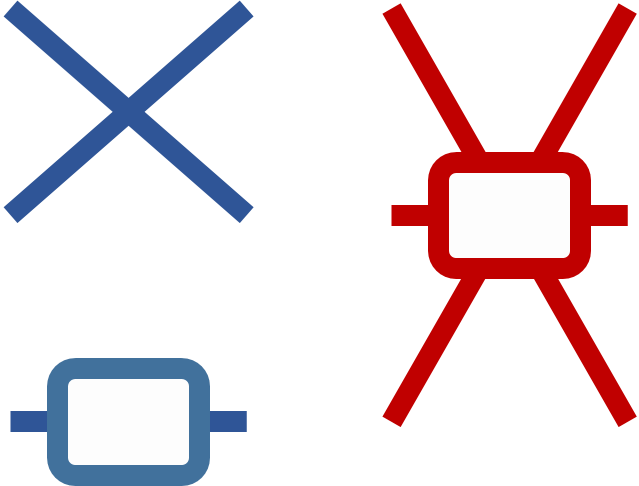} 
\end{tabular},
\end{eqnarray}
where we use the colored boxes to represent $\rho_A$ and $\rho_B$. To find out $\tr(\hat{O}_3^2)$, we make a similar table, Table.~\ref{table:O_3calculation}. By definition, taking trace of the terms in this table gives
\begin{eqnarray}
\tr\left(
\begin{tabular}{c}
     \includegraphics[scale=0.1]{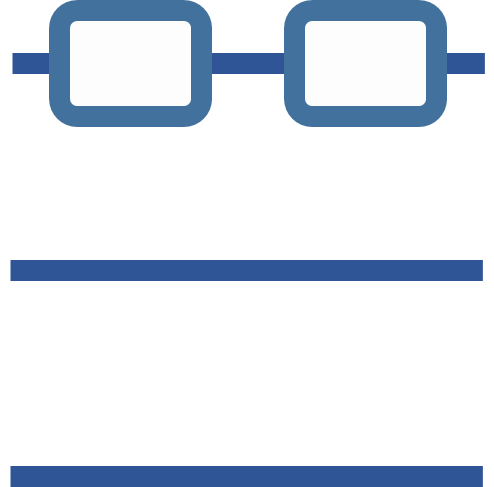} 
\end{tabular}
\right)
=
\tr\left(
\begin{tabular}{c}
     \includegraphics[scale=0.1]{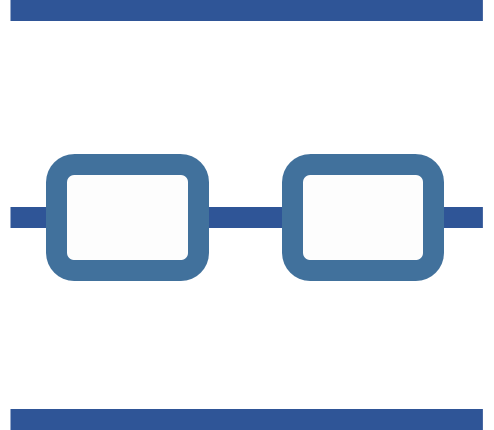} 
\end{tabular}
\right)
=
\tr\left(
\begin{tabular}{c}
     \includegraphics[scale=0.1]{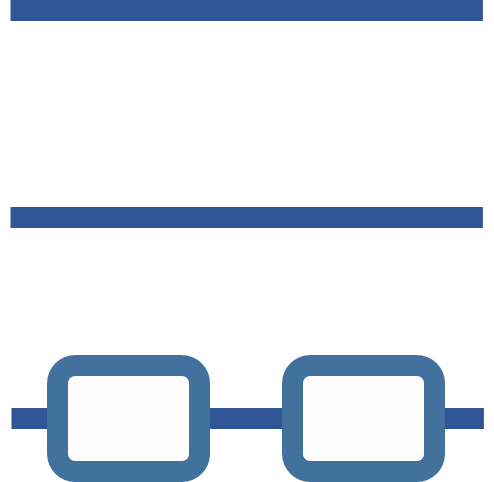} 
\end{tabular}
\right)=d_A^2\tr(\rho_A^2)
\end{eqnarray}
and
\begin{eqnarray}
\tr\left(
\begin{tabular}{c}
     \includegraphics[scale=0.1]{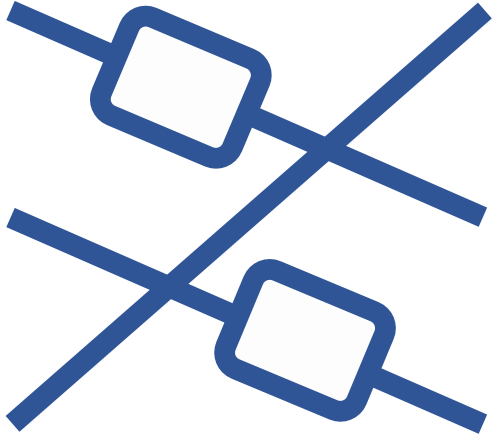} 
\end{tabular}
\right)
=
\tr\left(
\begin{tabular}{c}
     \includegraphics[scale=0.1]{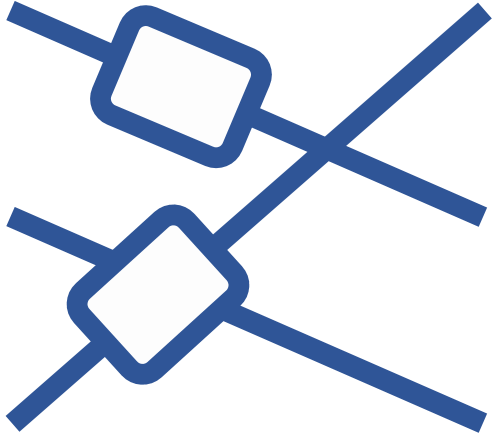} 
\end{tabular}
\right)
=
\tr\left(
\begin{tabular}{c}
     \includegraphics[scale=0.1]{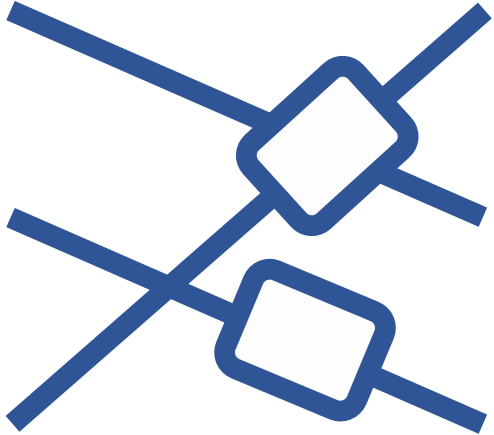} 
\end{tabular}
\right)
=
\tr\left(
\begin{tabular}{c}
     \includegraphics[scale=0.1]{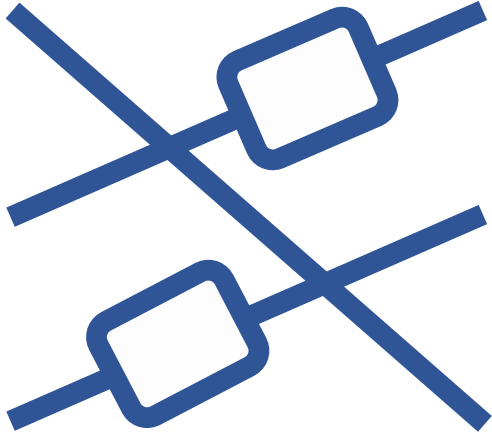} 
\end{tabular}
\right)
=
\tr\left(
\begin{tabular}{c}
     \includegraphics[scale=0.1]{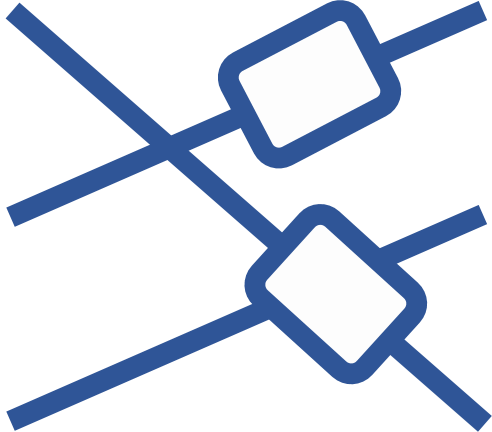} 
\end{tabular}
\right)
=
\tr\left(
\begin{tabular}{c}
     \includegraphics[scale=0.1]{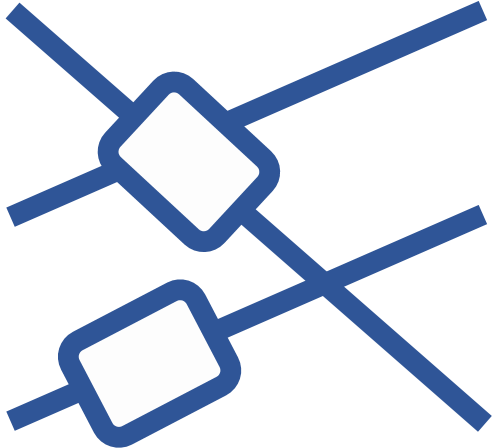} 
\end{tabular}
\right)
=\tr(\rho_A^2).
\end{eqnarray}
Hence
\begin{eqnarray}\label{eq:O_3square}
\tr[\hat{O}_3^2]\leq 36\tr\left[
\left(
\begin{tabular}{c}
     \includegraphics[scale=0.08]{O_3_1.png}
\end{tabular}
\right)
\left(
\begin{tabular}{c}
     \includegraphics[scale=0.08]{O_3_1.png}
\end{tabular}
\right)
\right]
=36\tr\left(
\begin{tabular}{c}
     \includegraphics[scale=0.08]{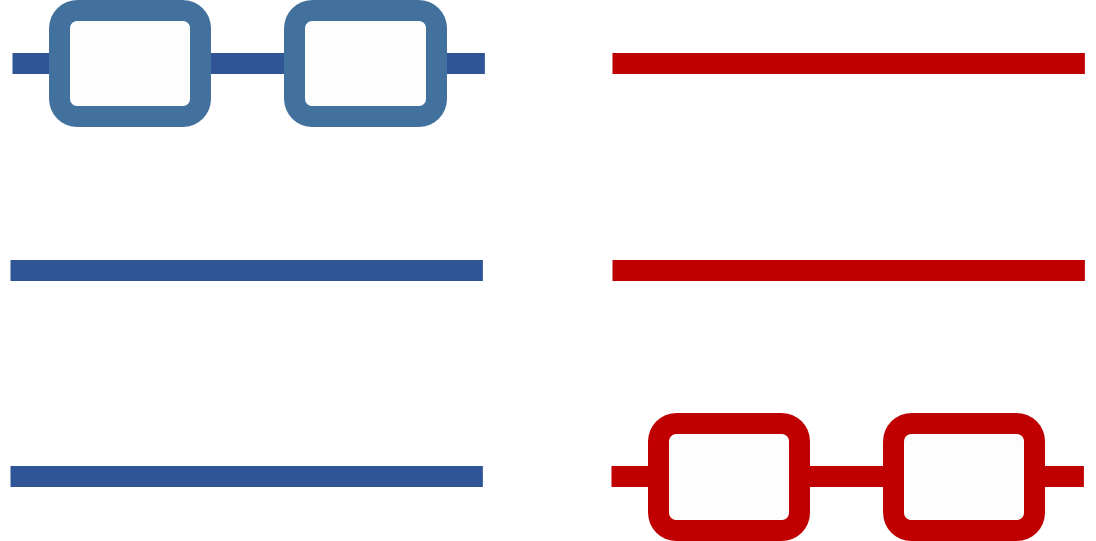}
\end{tabular}
\right)
=36d_A^2d_B^2\tr(\rho^2)=36D^2\tr(\rho^2)
\end{eqnarray}

\begin{table}
    \centering
    \setlength{\tabcolsep}{3mm}{
    \begin{tabular}{c||c|c|c|c|c|c}
    \hline\hline
         \rule{0pt}{28pt}
          & \includegraphics[scale=0.1]{O_3_1} & \includegraphics[scale=0.1]{O_3_2} &\includegraphics[scale=0.1]{O_3_3}&\includegraphics[scale=0.1]{O_3_4}&\includegraphics[scale=0.1]{O_3_5}&\includegraphics[scale=0.1]{O_3_6} \\
         \hline\hline
         \rule{0pt}{28pt}
         \includegraphics[scale=0.1]{O_3_1}&\includegraphics[scale=0.1]{O_3_11} &\includegraphics[scale=0.1]{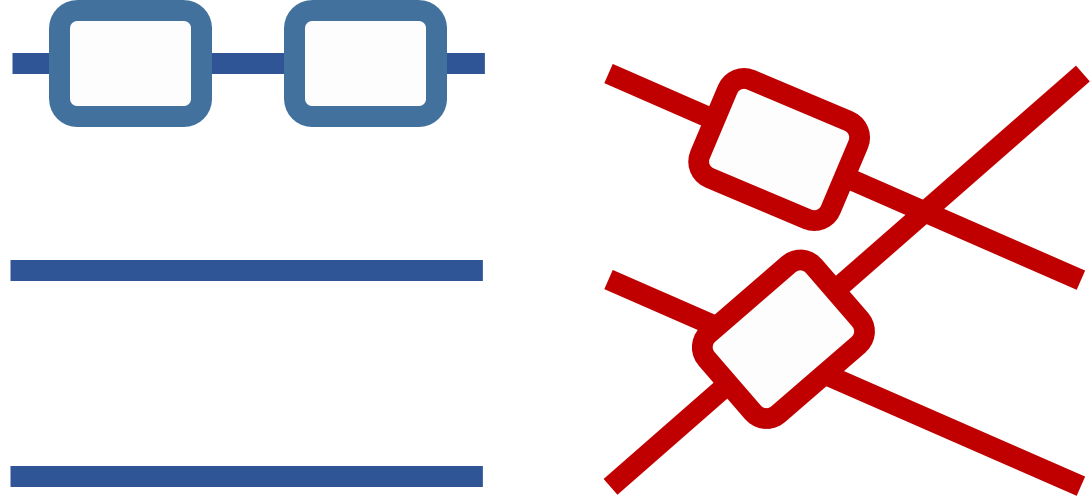} &\includegraphics[scale=0.1]{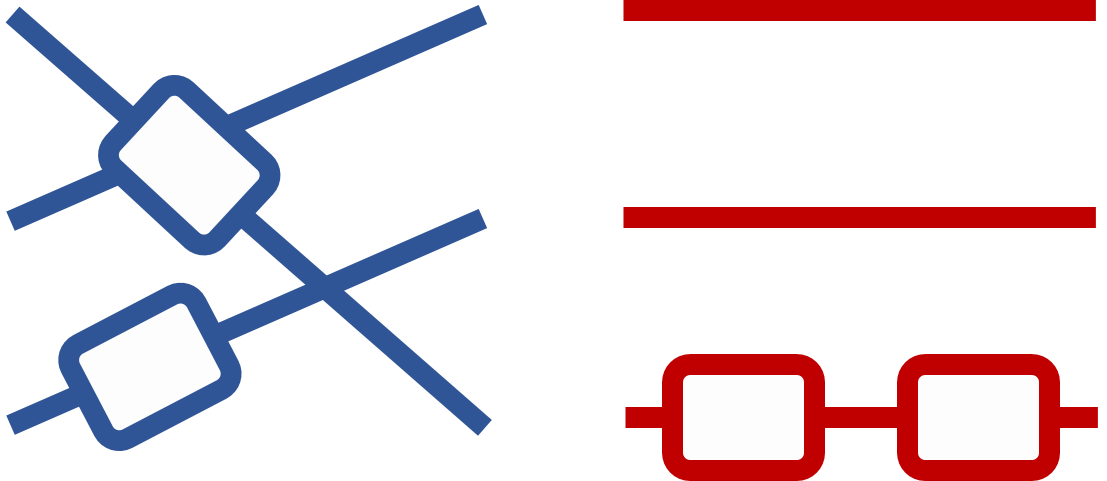} &\includegraphics[scale=0.1]{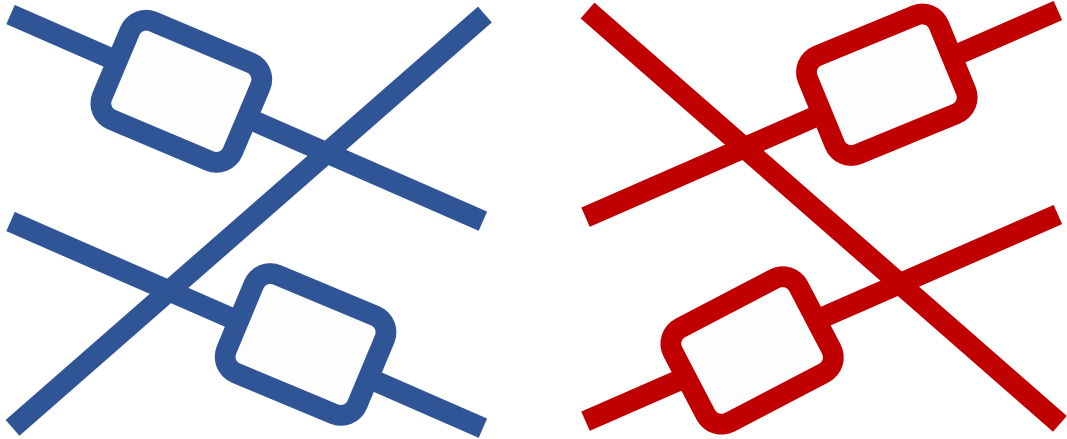} &\includegraphics[scale=0.1]{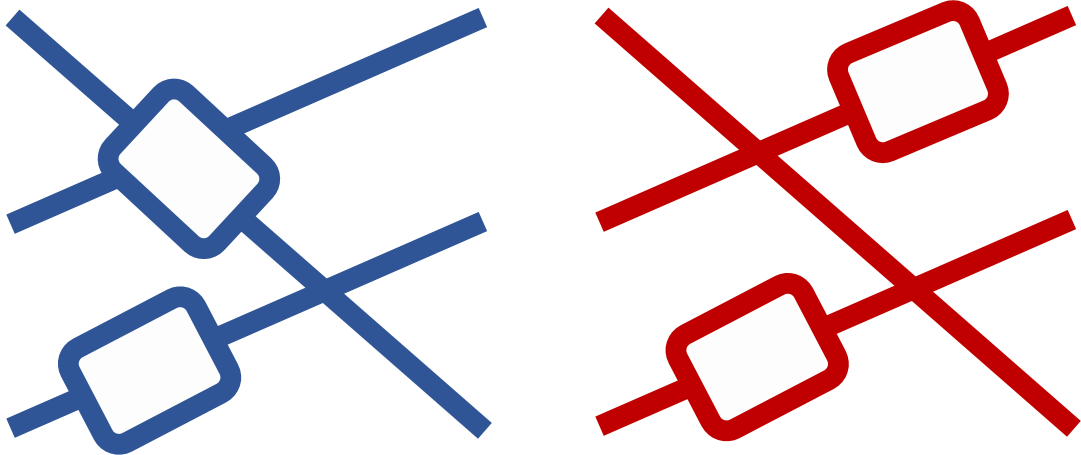} &\includegraphics[scale=0.1]{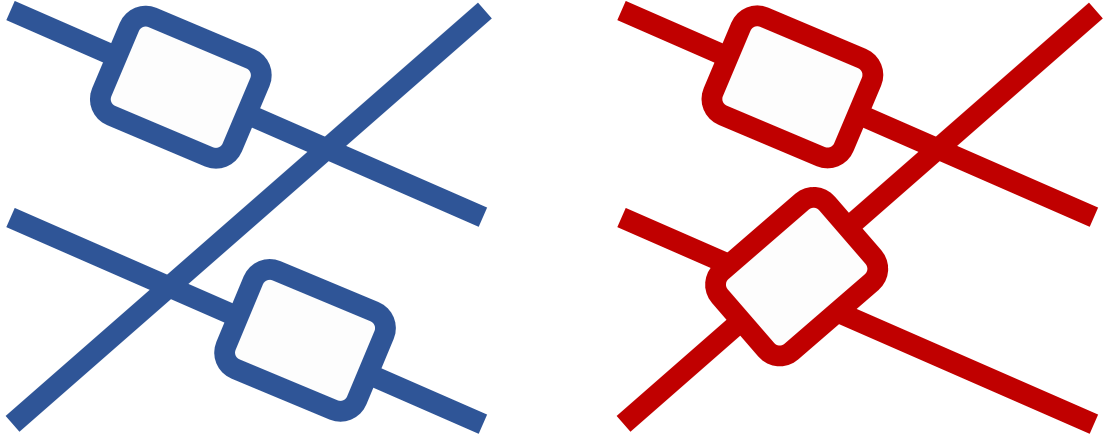} \\
         \hline
         \rule{0pt}{28pt}
         \includegraphics[scale=0.1]{O_3_2}&\includegraphics[scale=0.1]{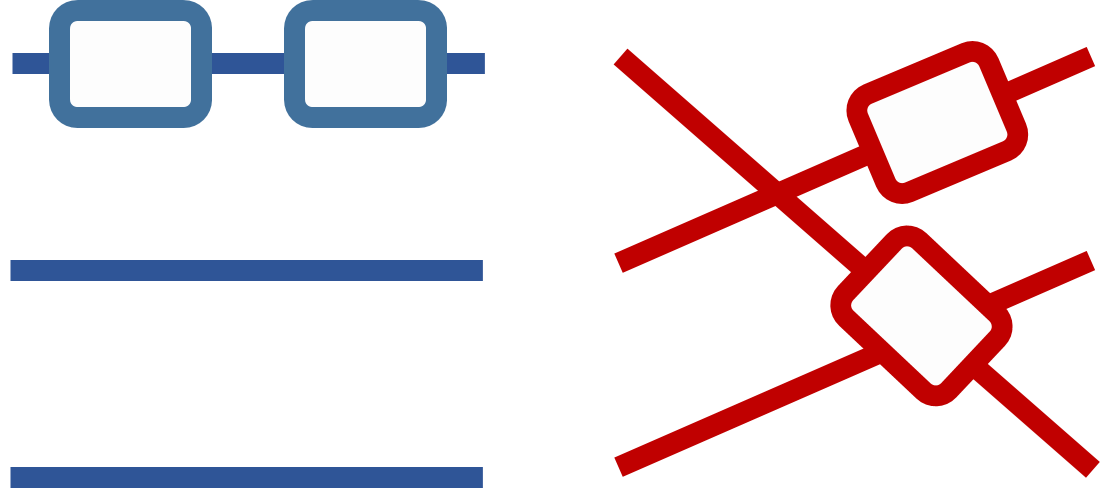} &\includegraphics[scale=0.1]{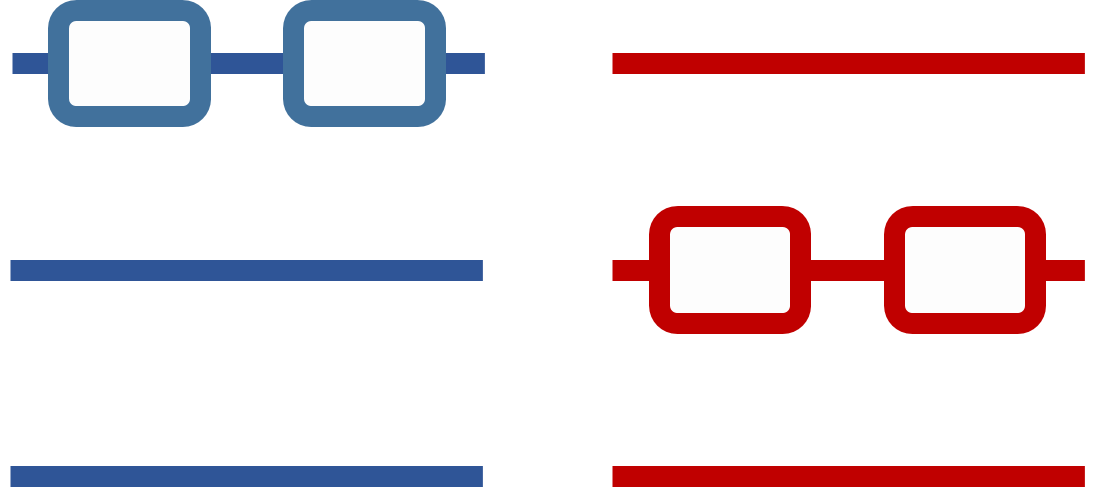} &\includegraphics[scale=0.1]{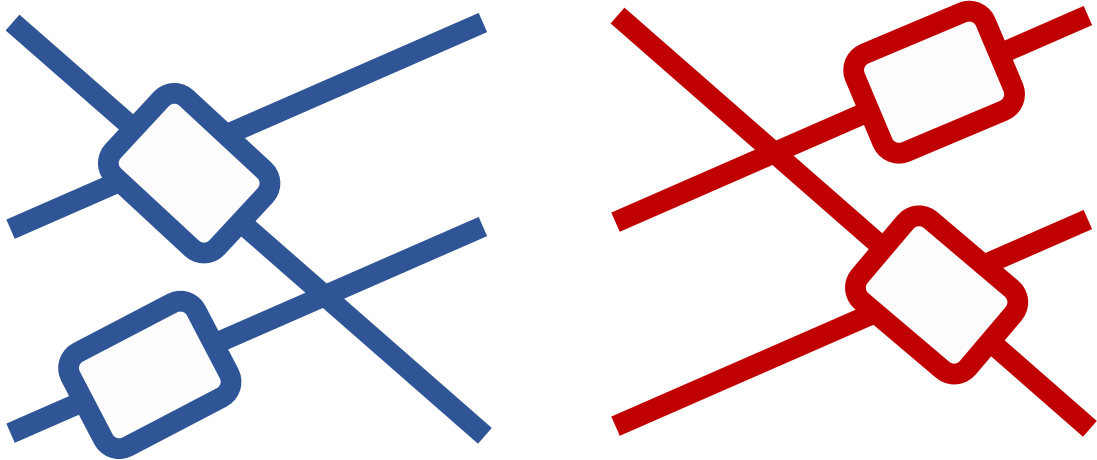} &\includegraphics[scale=0.1]{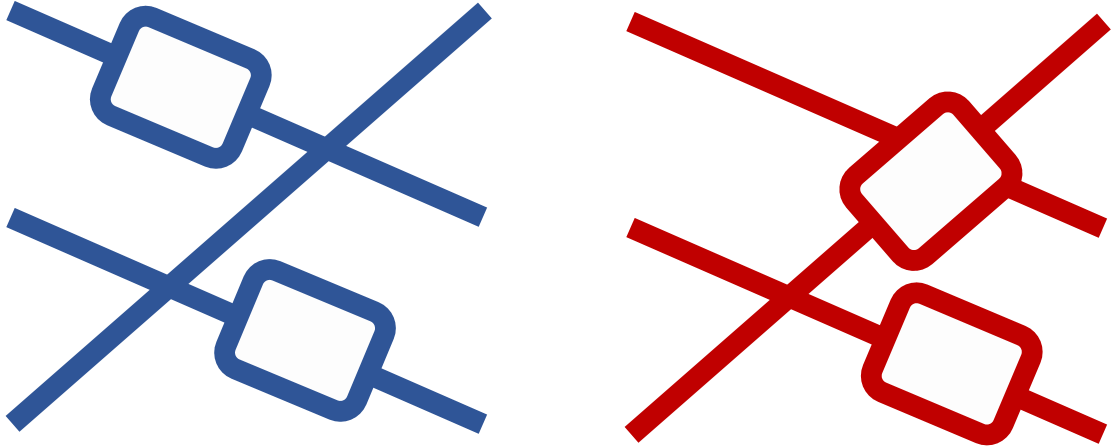} &\includegraphics[scale=0.1]{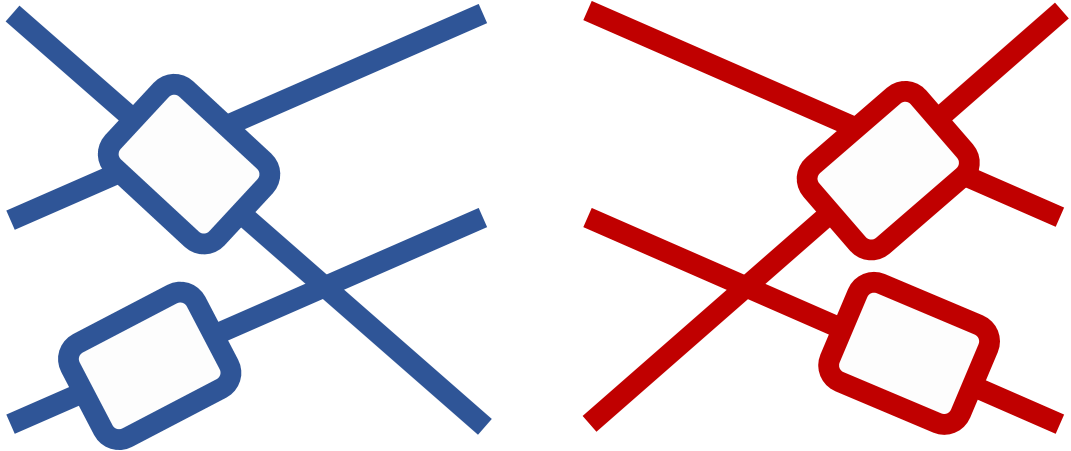} &\includegraphics[scale=0.1]{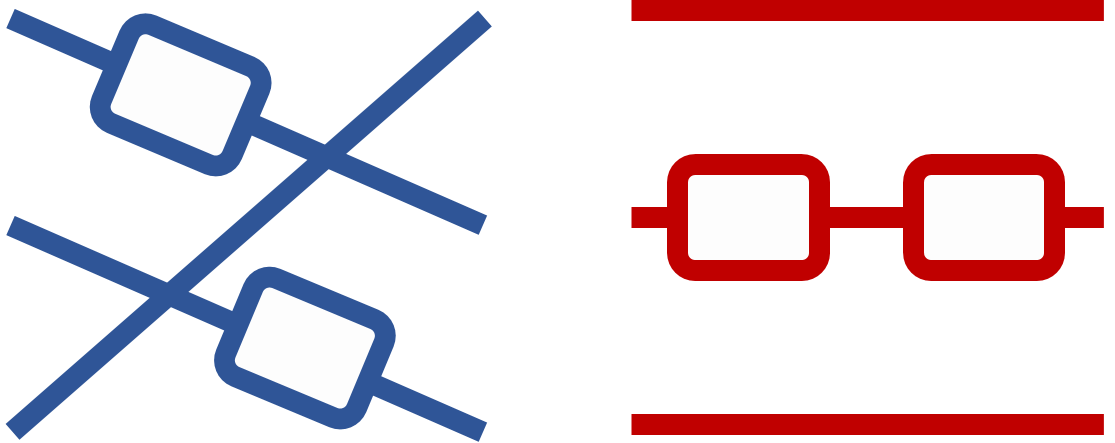} \\
         \hline
         \rule{0pt}{28pt}
         \includegraphics[scale=0.1]{O_3_3}&\includegraphics[scale=0.1]{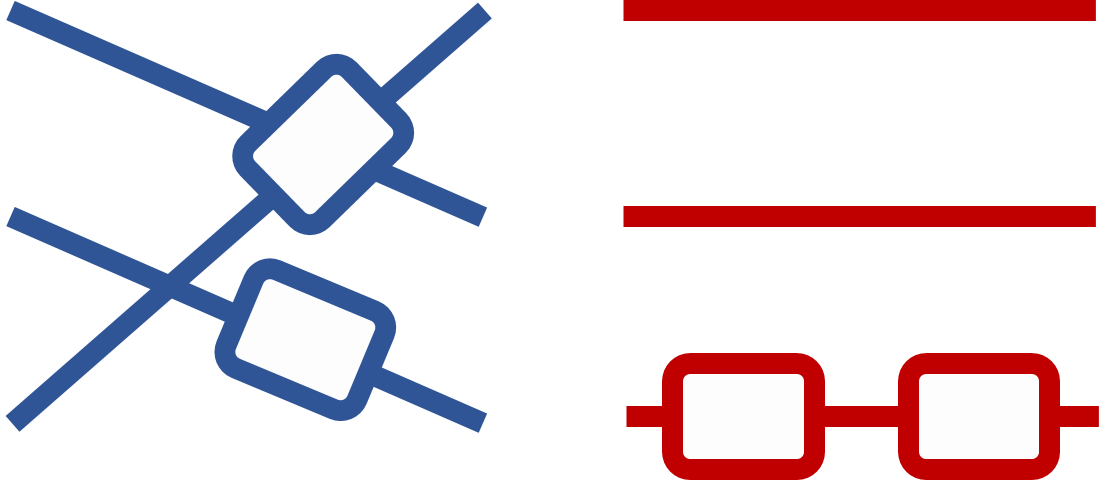} &\includegraphics[scale=0.1]{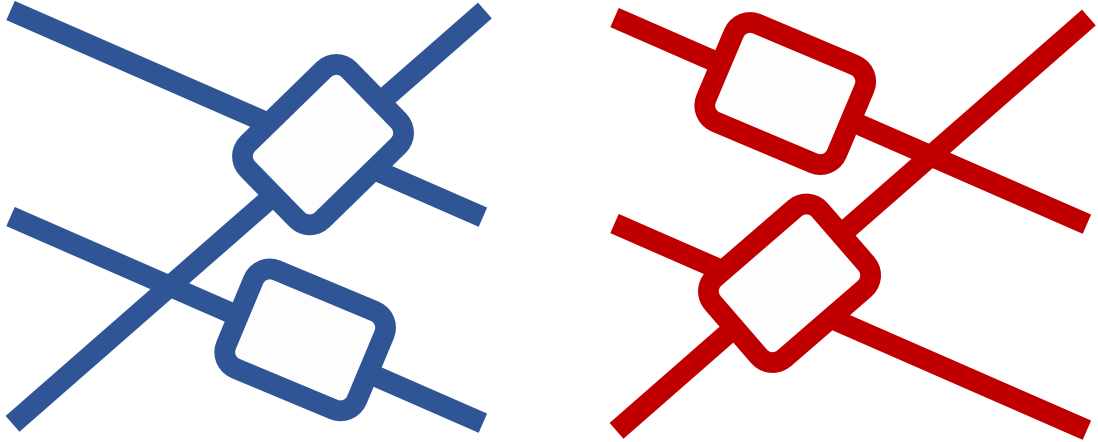} &\includegraphics[scale=0.1]{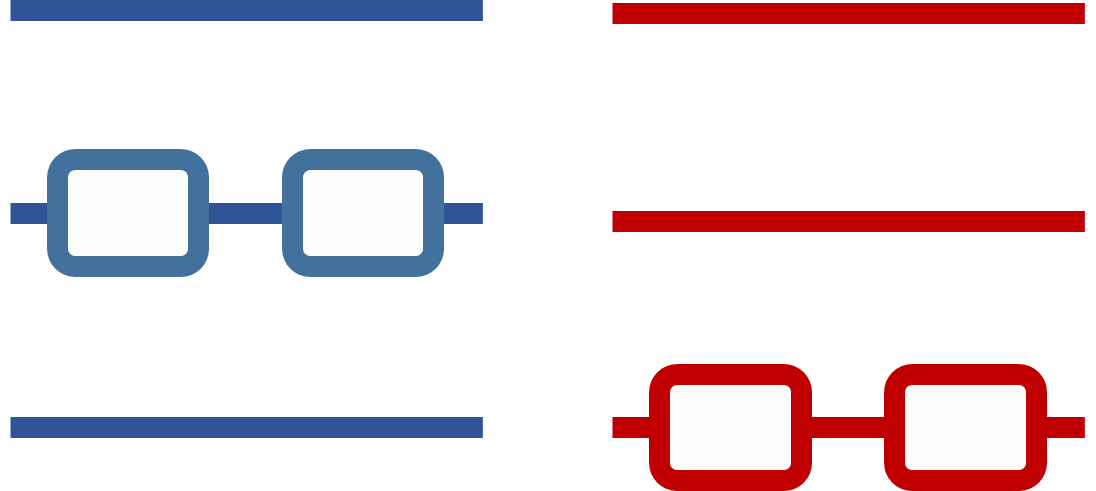} &\includegraphics[scale=0.1]{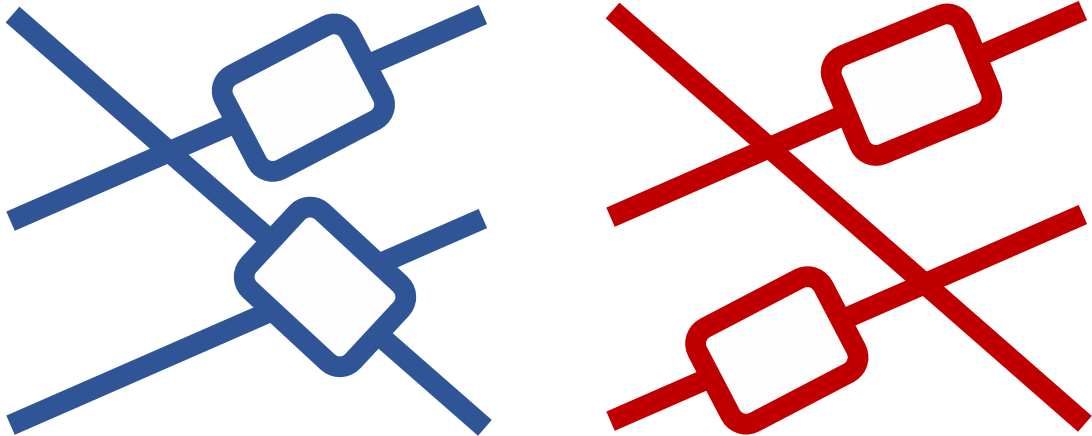} &\includegraphics[scale=0.1]{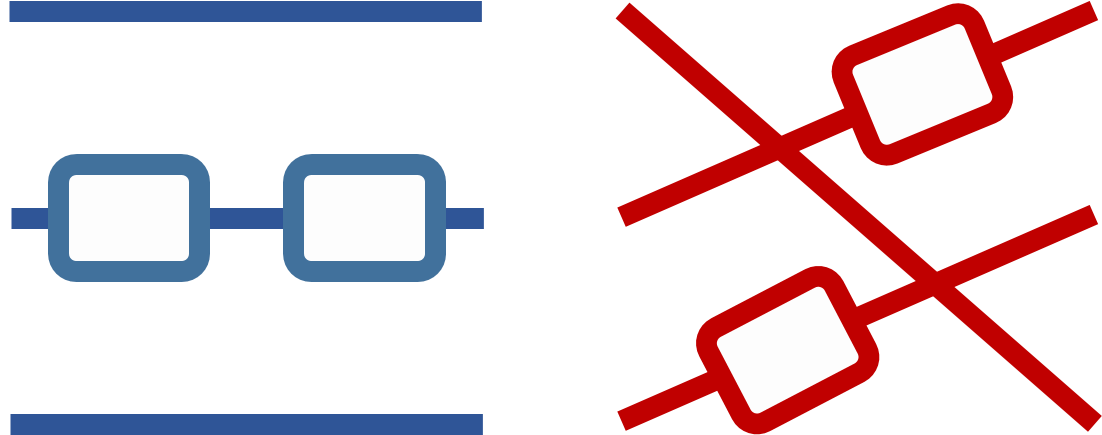} &\includegraphics[scale=0.1]{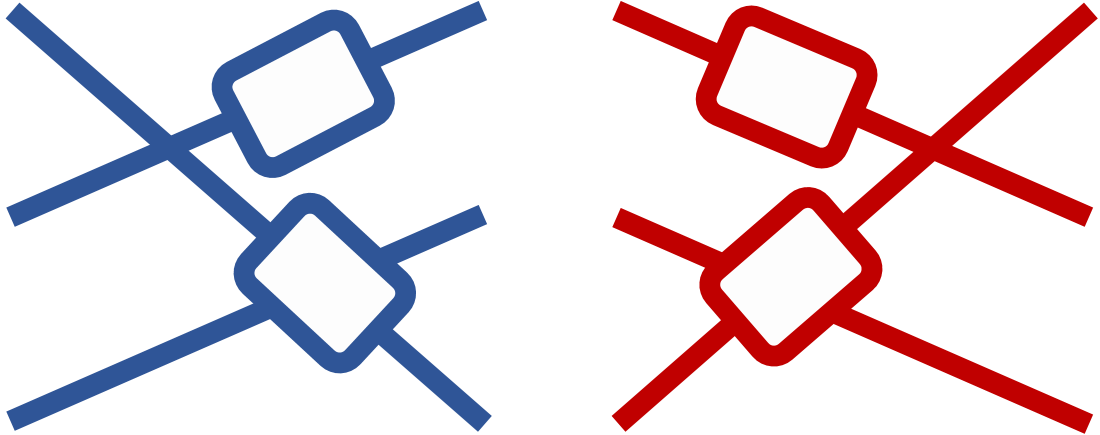} \\
         \hline
         \rule{0pt}{28pt}
         \includegraphics[scale=0.1]{O_3_4}&\includegraphics[scale=0.1]{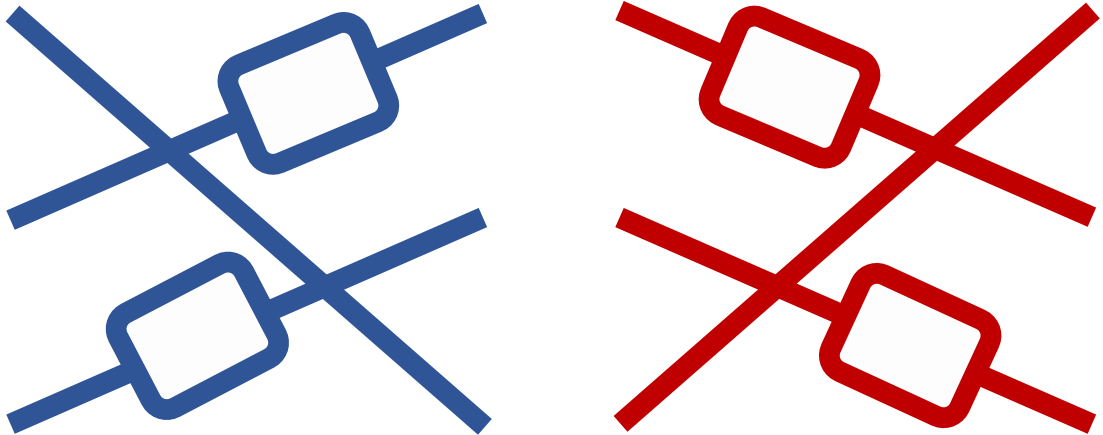} &\includegraphics[scale=0.1]{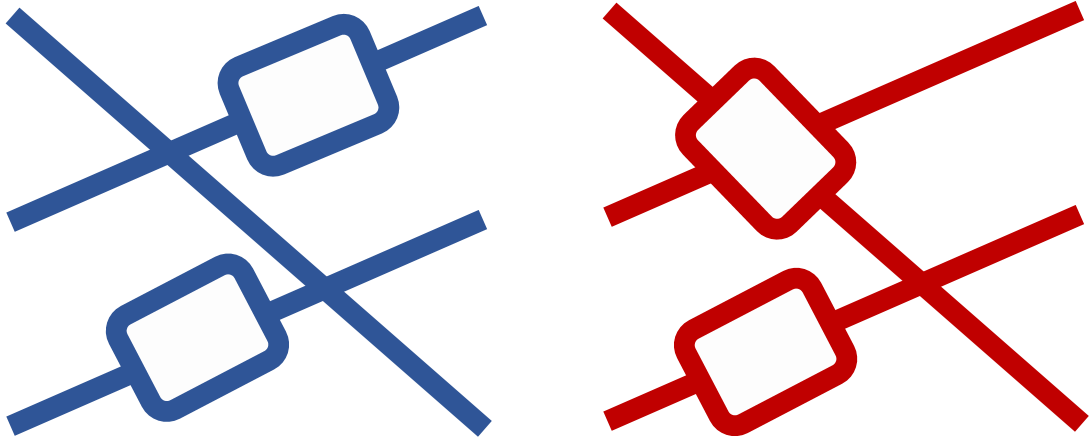} &\includegraphics[scale=0.1]{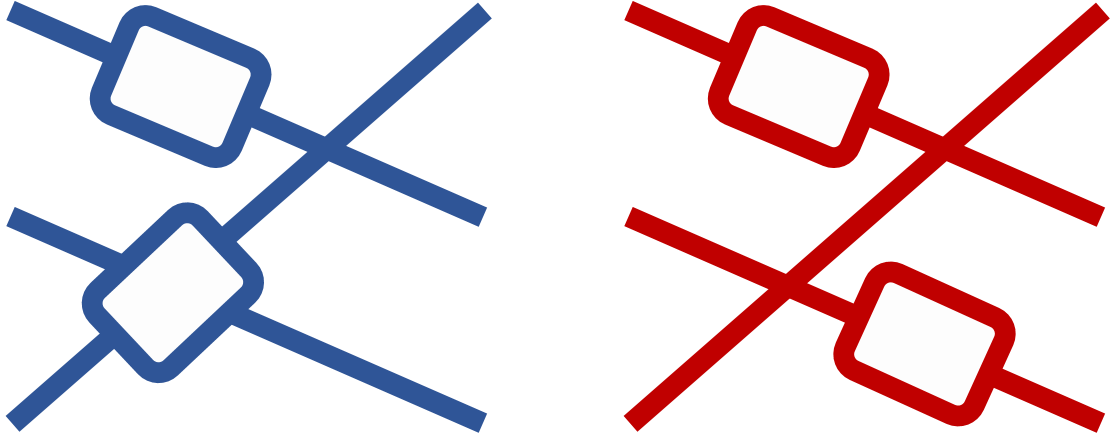} &\includegraphics[scale=0.1]{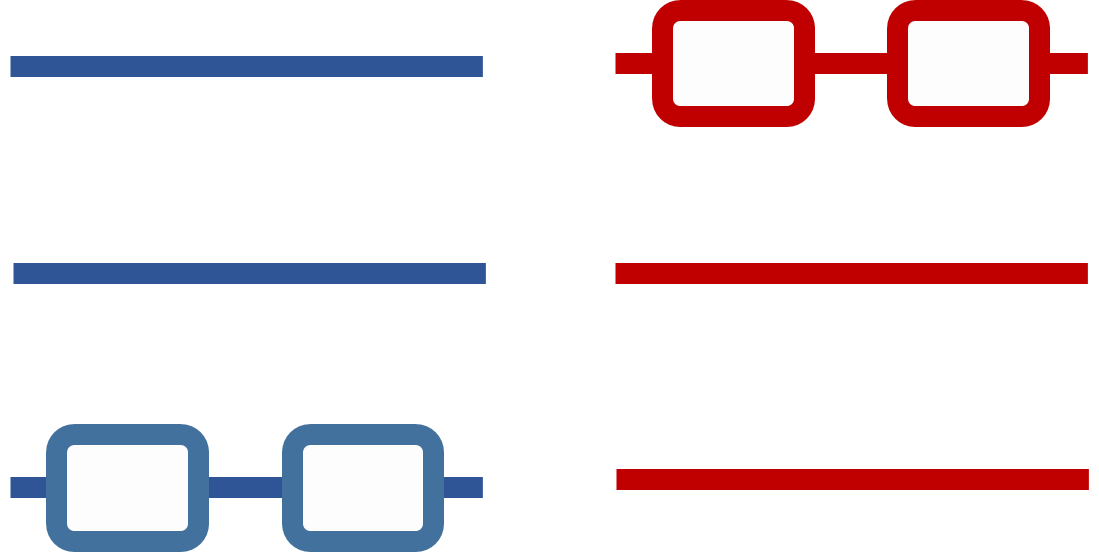} &\includegraphics[scale=0.1]{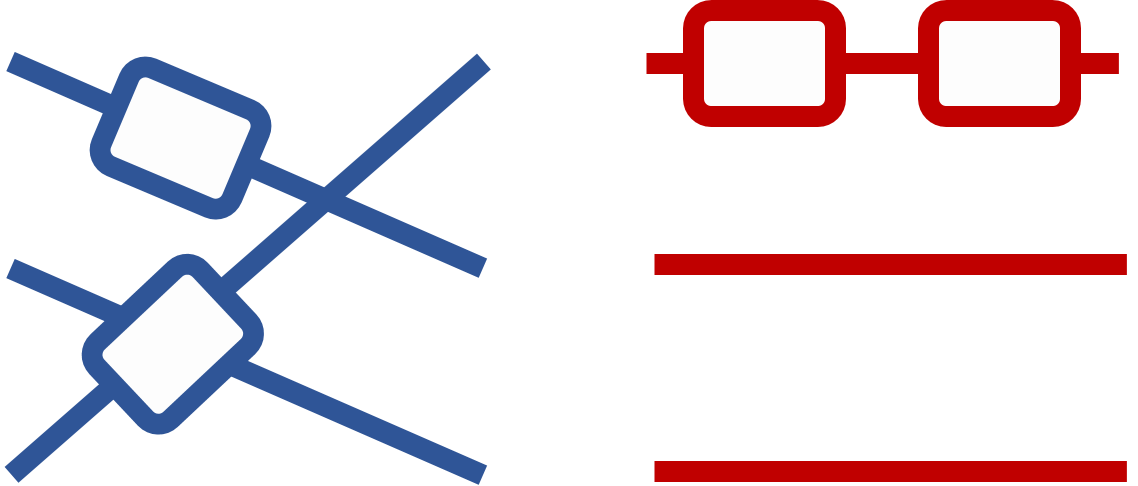} &\includegraphics[scale=0.1]{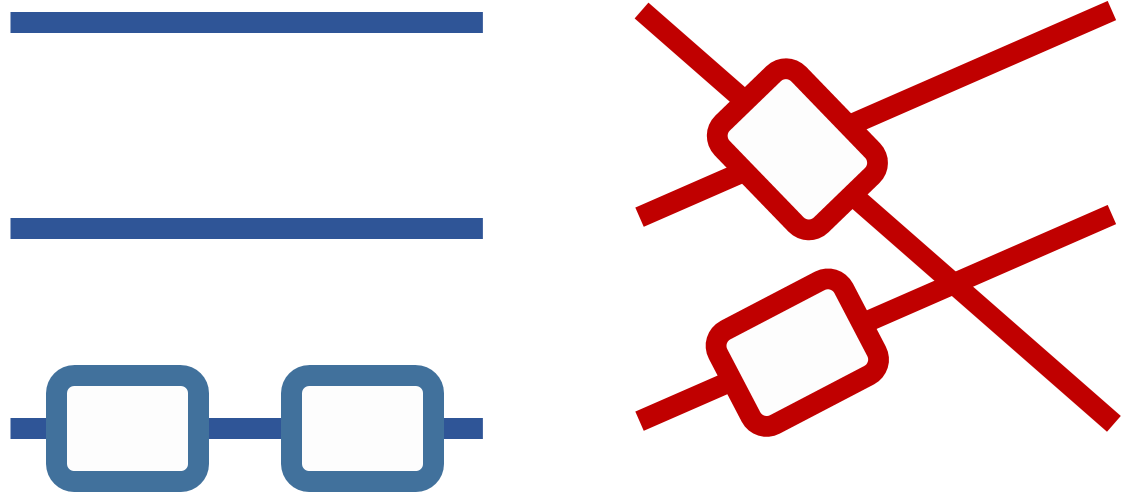} \\
         \hline
         \rule{0pt}{28pt}
         \includegraphics[scale=0.1]{O_3_5}&\includegraphics[scale=0.1]{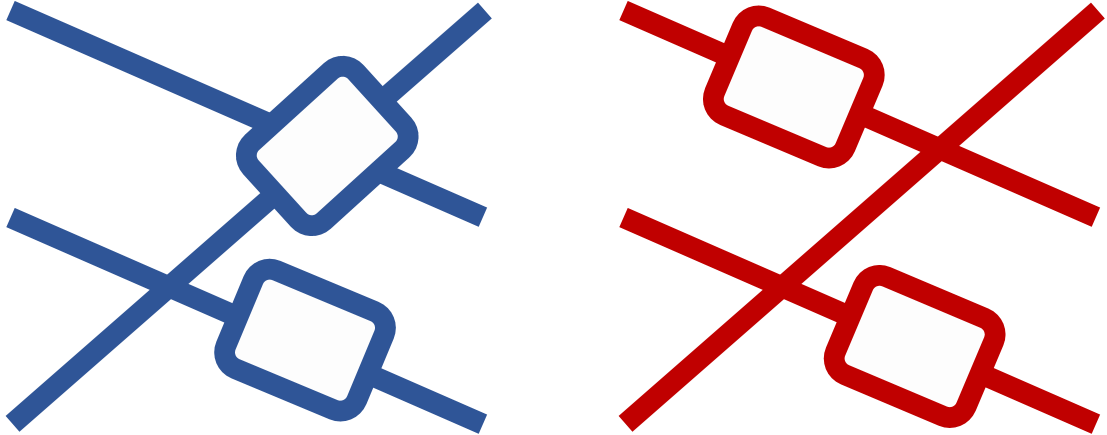} &\includegraphics[scale=0.1]{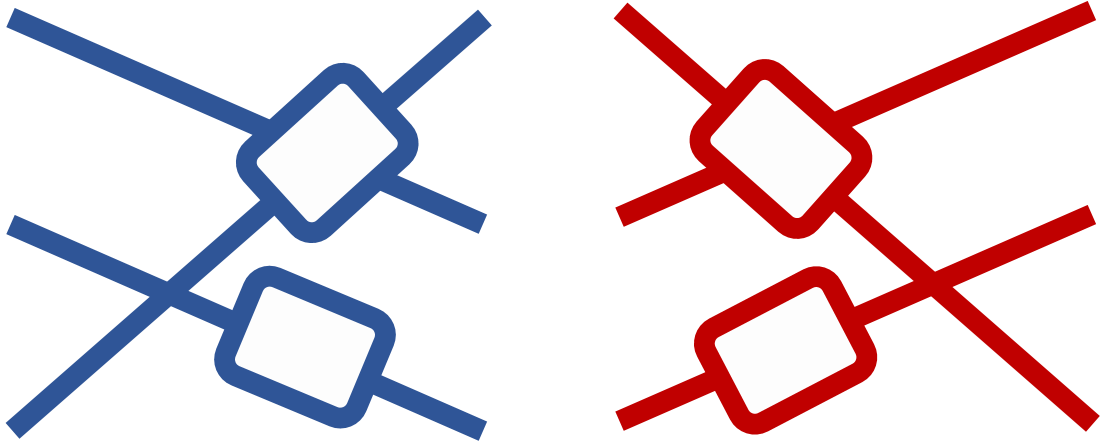} &\includegraphics[scale=0.1]{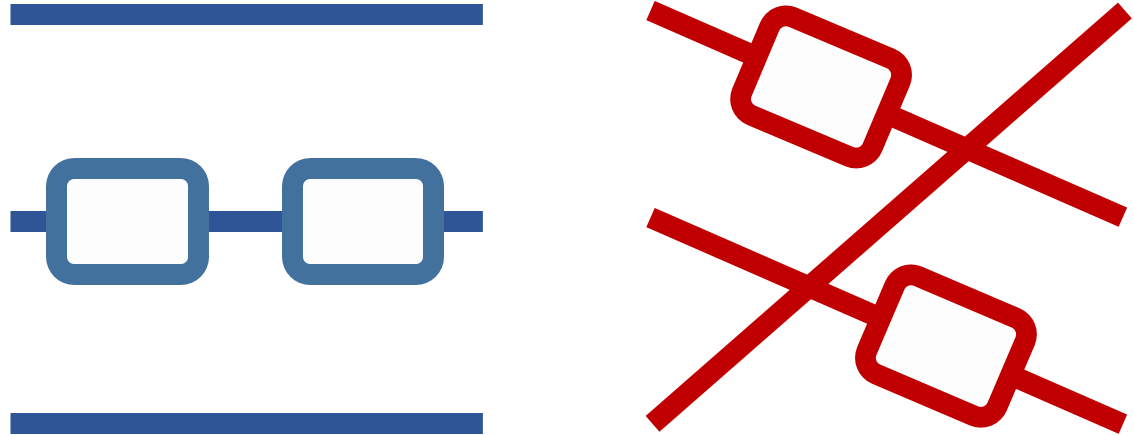} &\includegraphics[scale=0.1]{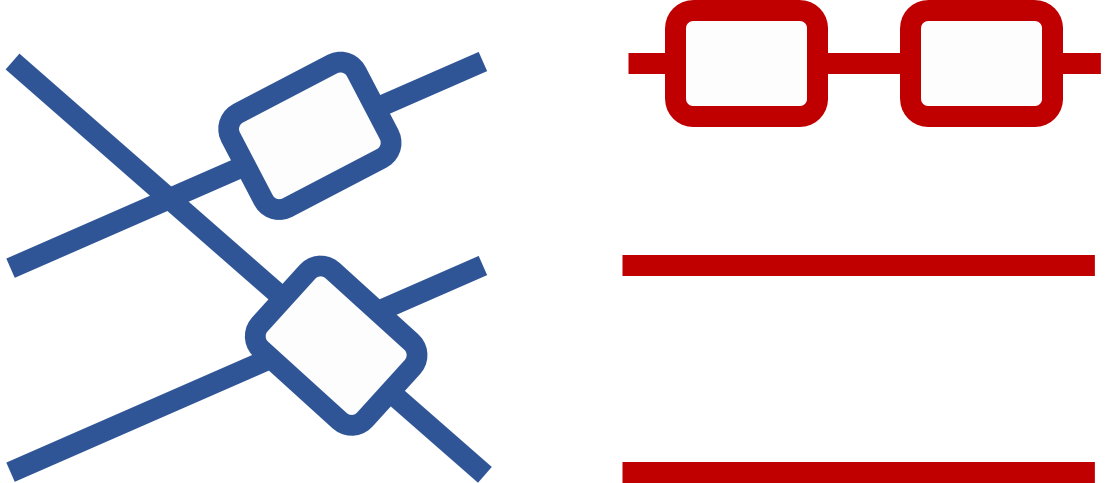} &\includegraphics[scale=0.1]{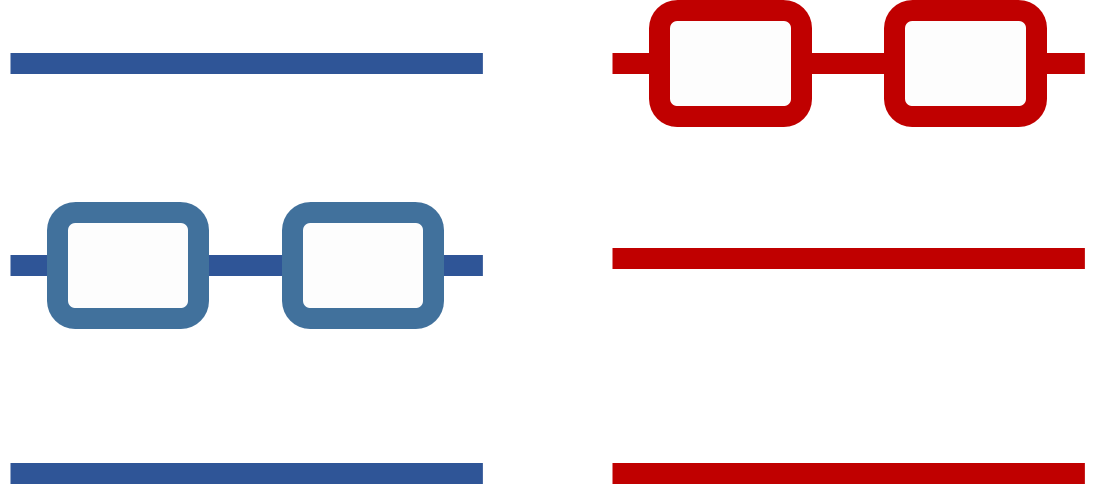} &\includegraphics[scale=0.1]{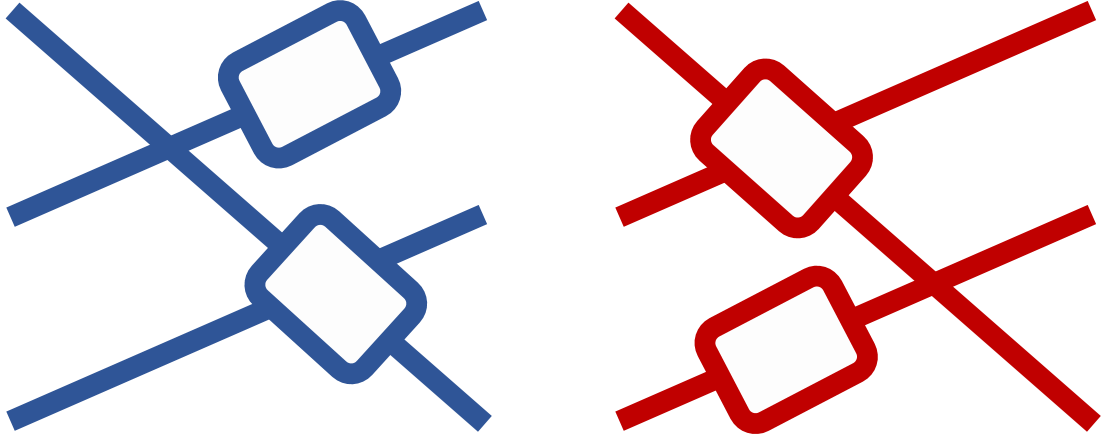} \\
         \hline
         \rule{0pt}{28pt}
         \includegraphics[scale=0.1]{O_3_6}&\includegraphics[scale=0.1]{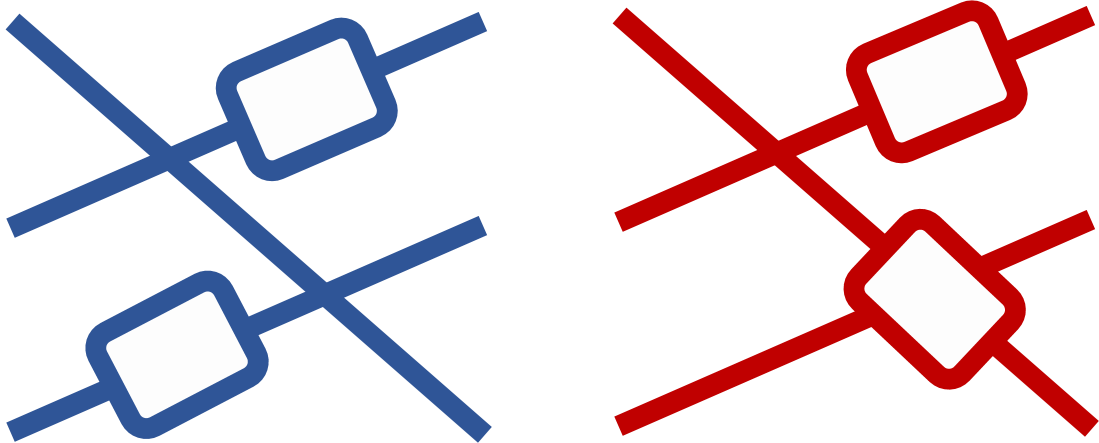} &\includegraphics[scale=0.1]{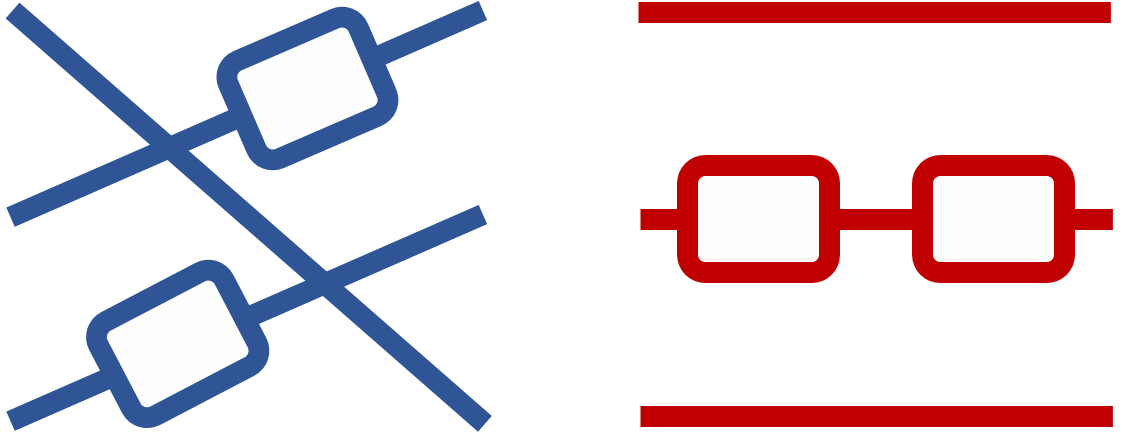} &\includegraphics[scale=0.1]{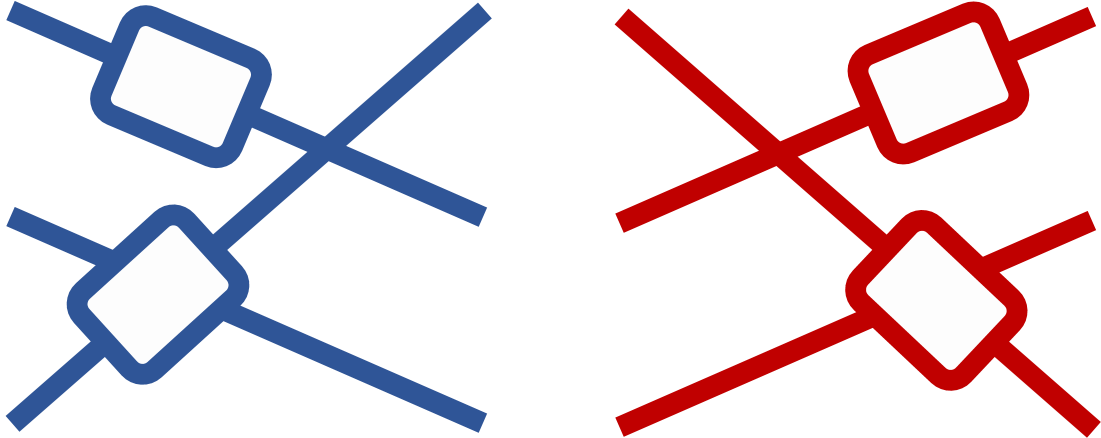} &\includegraphics[scale=0.1]{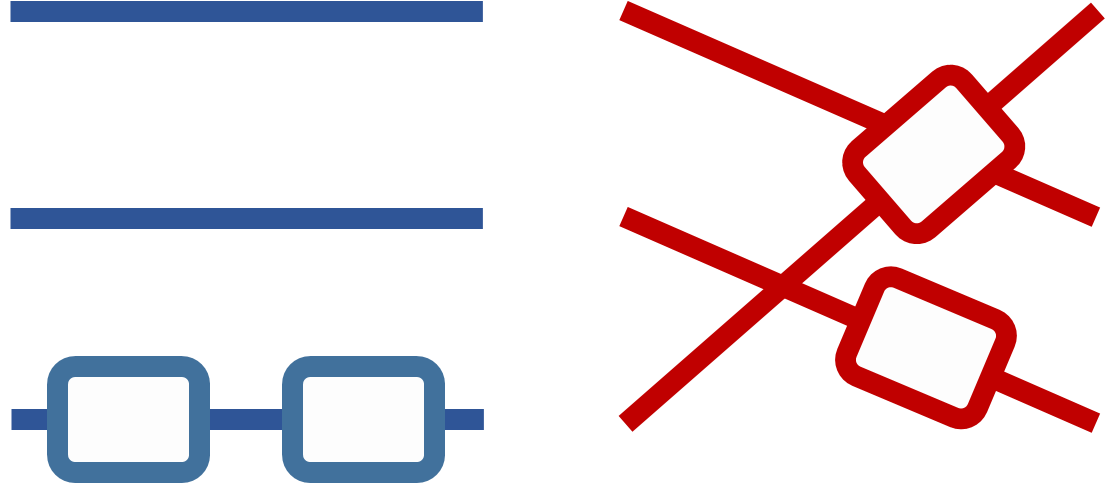} &\includegraphics[scale=0.1]{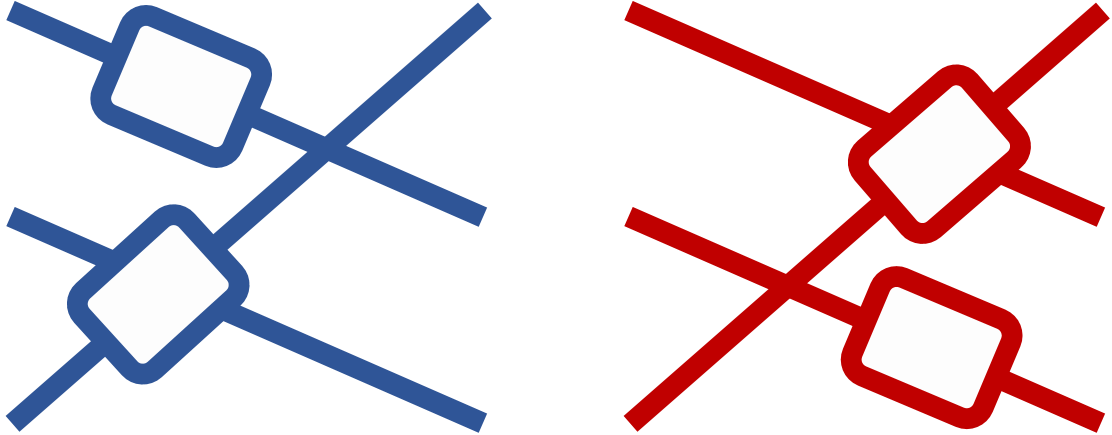} &\includegraphics[scale=0.1]{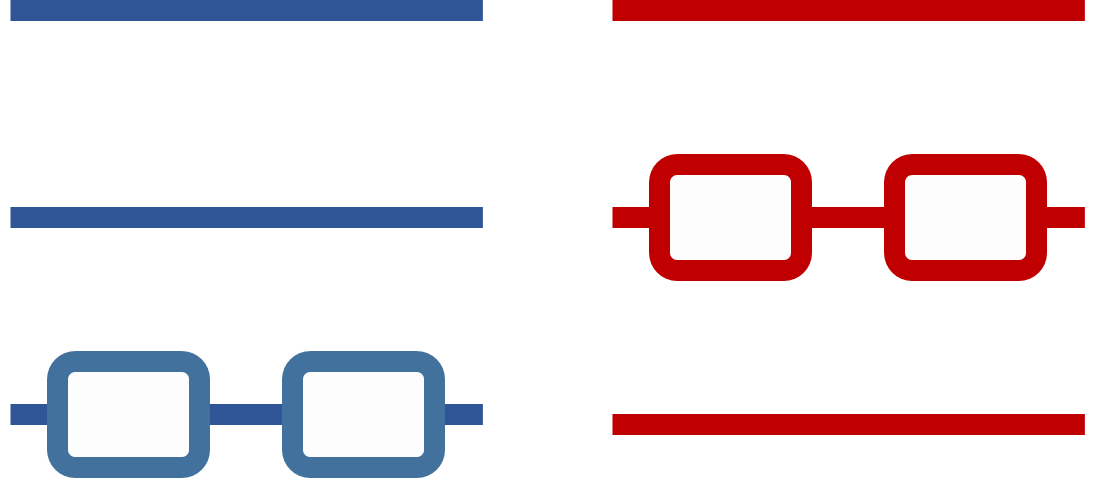} \\
         \hline\hline
    \end{tabular}}
    \caption{ Multiplication of the terms in $\hat{O}_3^2$. Here, colored boxes stand for density matrices corresponding to different subsystems. ``X'' form operators and the horizontal lines follow the same definition as Table.~\ref{table:O_4calculation}, others are tri-partite cyclic permutation operators. }
    \label{table:O_3calculation}
\end{table}

$\hat{O}_2$ can be similarly constructed from $\hat{O}_4$:
\begin{equation}
\begin{split}
\hat{O}_2=\frac{1}{2}\tr_{1,2}\left[\hat{O}_4(\rho\otimes\rho\otimes \mathbb{I}^{\otimes2})\right]=\frac{1}{2}\tr_1\left[\hat{O}_3(\rho\otimes\mathbb{I}^{\otimes2})\right],
\end{split}
\end{equation}
which can be graphically demonstrated as
\begin{eqnarray}
\hat{O}_2
&=
\frac{1}{2}\left(
\begin{tabular}{c}
     \includegraphics[scale=0.1]{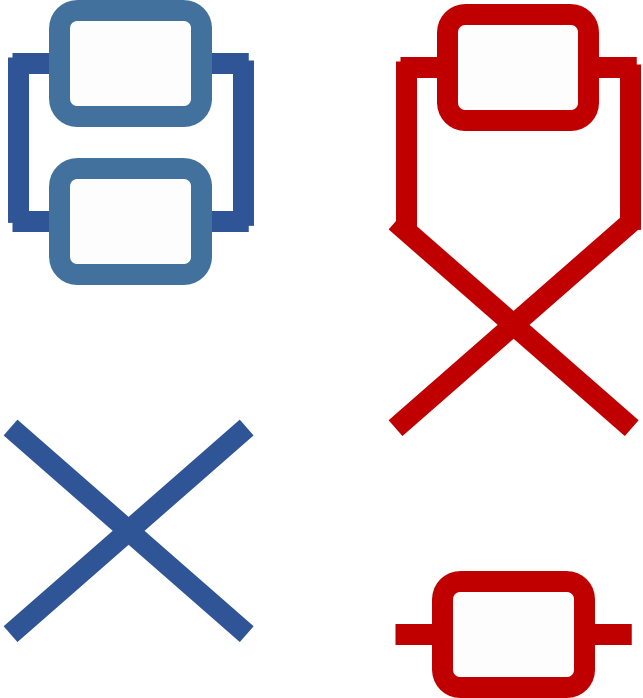}
\end{tabular}
+
\begin{tabular}{c}
     \includegraphics[scale=0.1]{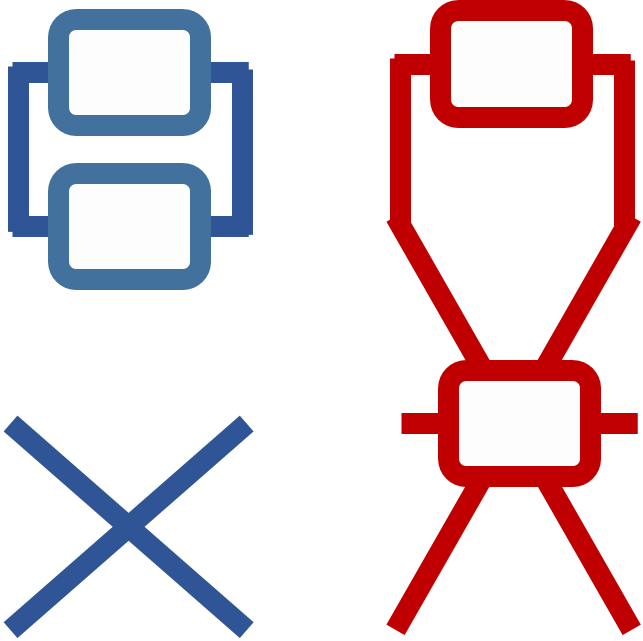}
\end{tabular}
+
\begin{tabular}{c}
     \includegraphics[scale=0.1]{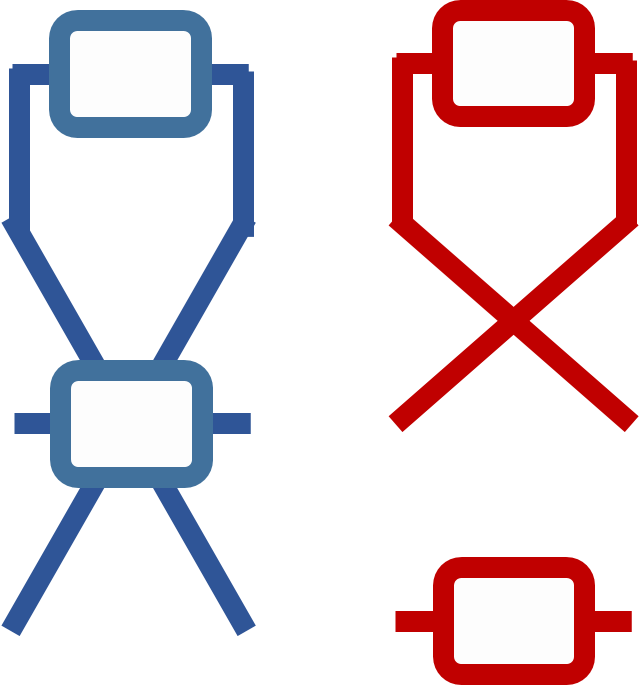}
\end{tabular}
+
\begin{tabular}{c}
     \includegraphics[scale=0.1]{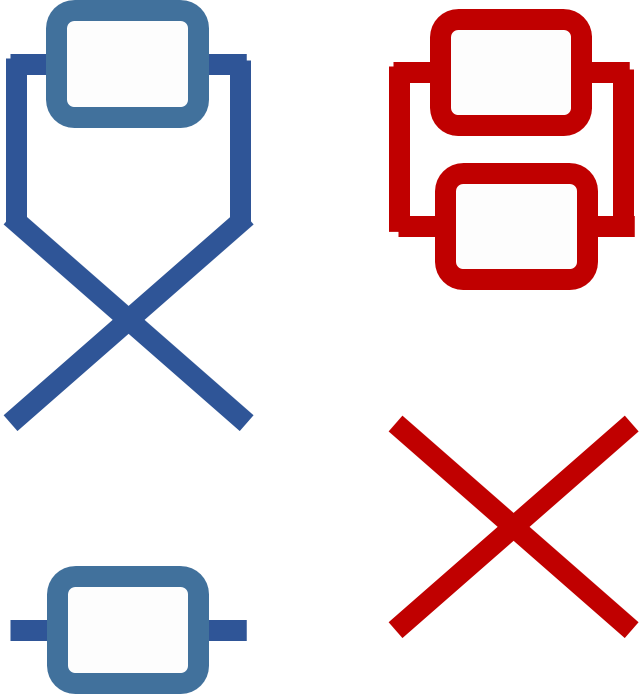}
\end{tabular}
+
\begin{tabular}{c}
     \includegraphics[scale=0.1]{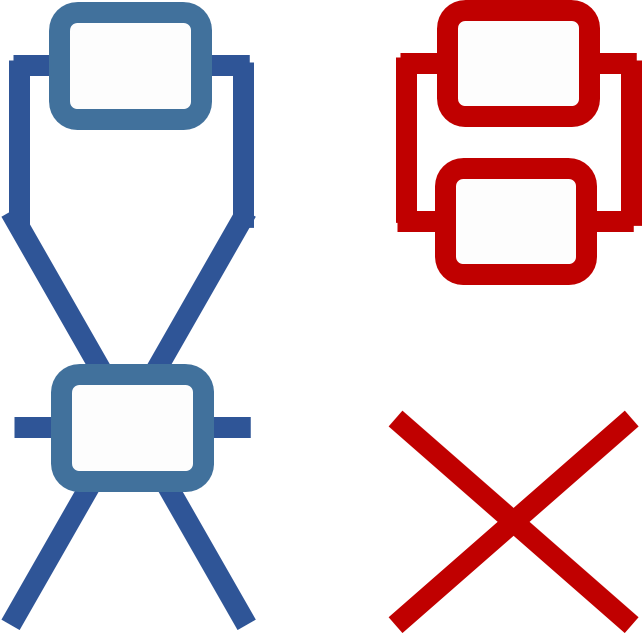}
\end{tabular}
+
\begin{tabular}{c}
     \includegraphics[scale=0.1]{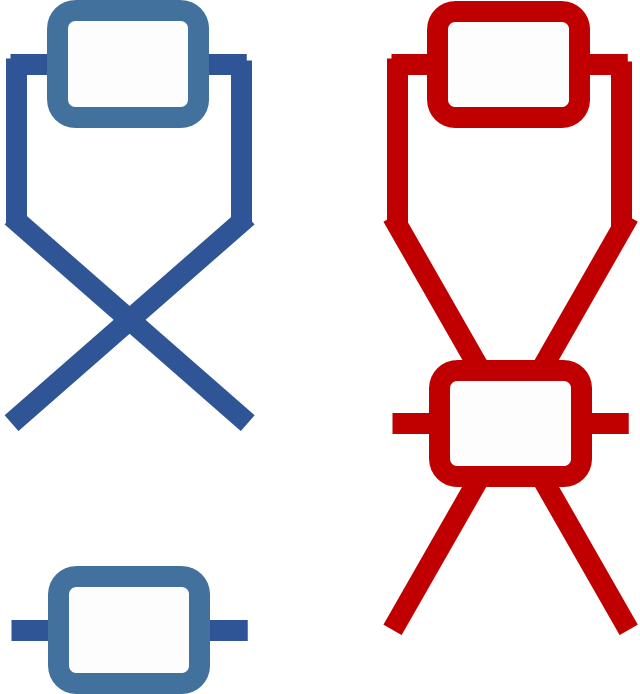}
\end{tabular}
\right)\\
&=
\begin{tabular}{c}
     \includegraphics[scale=0.1]{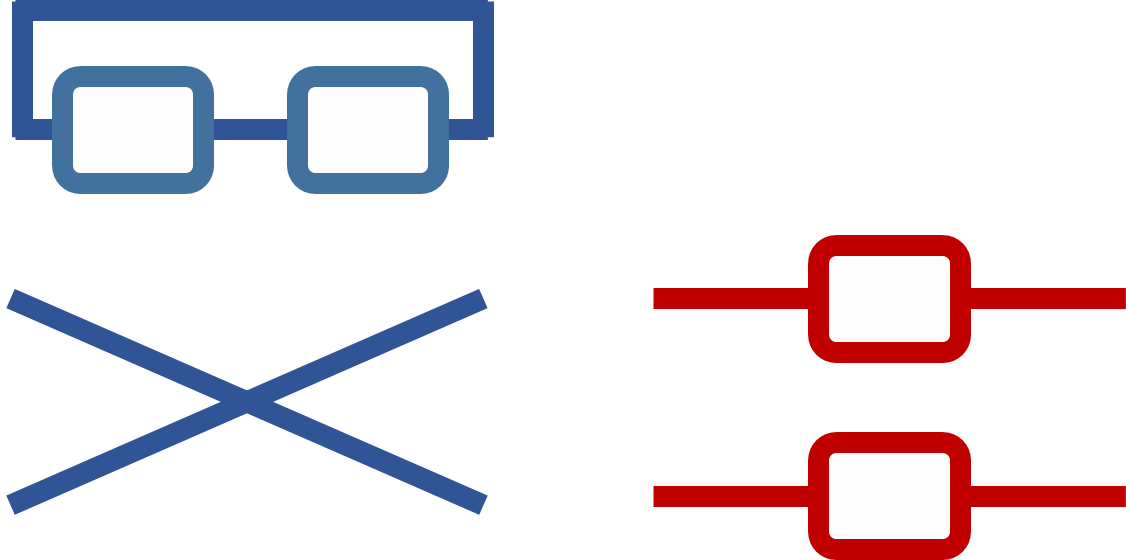}
\end{tabular}
+
\begin{tabular}{c}
     \includegraphics[scale=0.1]{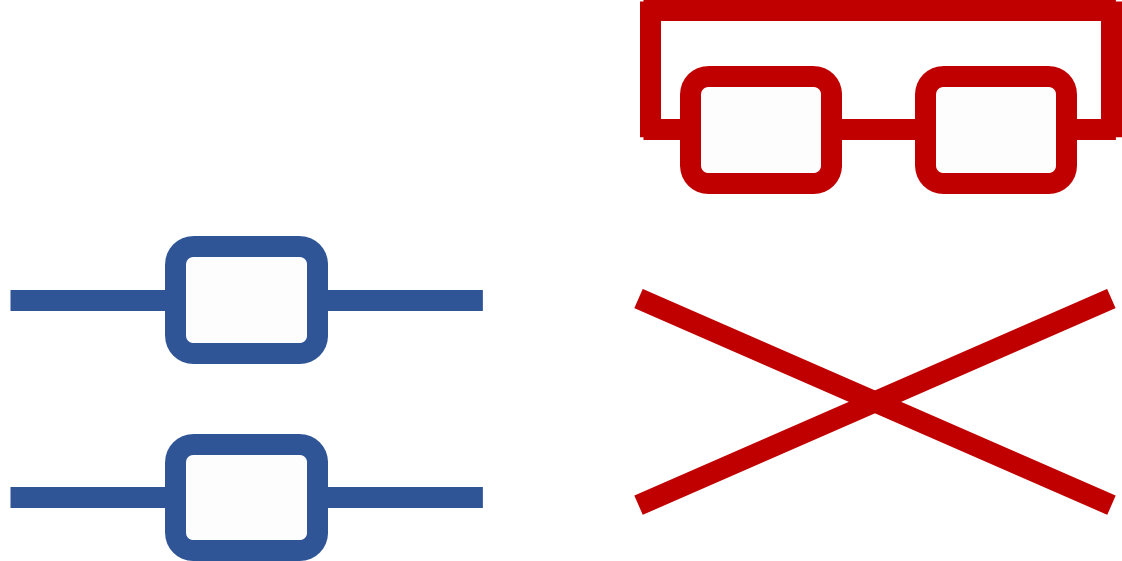}
\end{tabular}
+
\begin{tabular}{c}
     \includegraphics[scale=0.1]{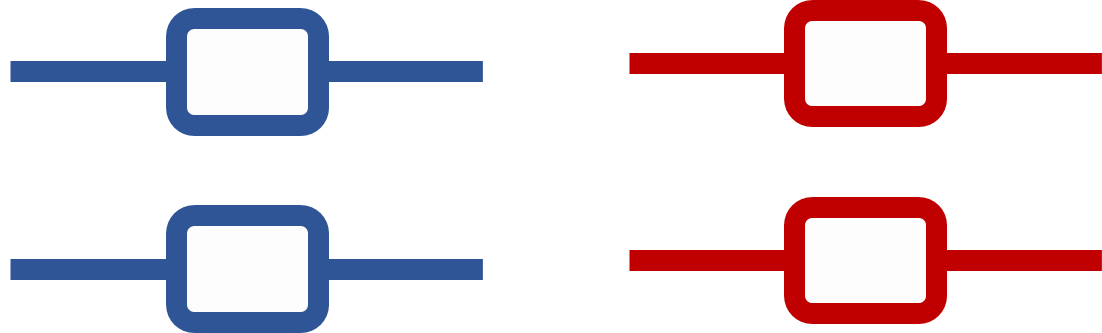}.
\end{tabular}
\end{eqnarray}
After the calculation of $\tr[\hat{O}_4^2]$ and $\tr[\hat{O}_3^2]$, one could realize that, the largest term contains most horizontal lines (identity operator $\mathbb{I}$) without box, which contributes $d$ when taking trace. As a result,
\begin{eqnarray}\label{eq:O_2square}
\tr[\hat{O}_2^2]\leq 9\max\left\{\tr\left[
\left(
\begin{tabular}{c}
     \includegraphics[scale=0.1]{O_2_1.png}
\end{tabular}
\right)^2
\right],
\tr\left[
\left(
\begin{tabular}{c}
     \includegraphics[scale=0.1]{O_2_2.png}
\end{tabular}
\right)^2
\right]
\right\}=\max\left\{d_A^2M_4^{(2,3)},d_B^2M_4^{(2,3)}\right\}.
\end{eqnarray}

Similarly, 
\begin{equation}
\begin{split}
\hat{O}_1=\frac{1}{6}\tr_{1,2,3}\left[\hat{O}_4(\rho^{\otimes 3}\otimes\mathbb{I})\right]=\frac{1}{6}\tr_{1,2}\left[\hat{O}_3(\rho^{\otimes 2}\otimes\mathbb{I})\right]=\frac{1}{3}\tr_1\left[\hat{O}_2(\rho\otimes\mathbb{I})\right],
\end{split}
\end{equation}
which can be graphically represented as
\begin{eqnarray}
\hat{O}_1 =
\frac{1}{3}\left(
\begin{tabular}{c}
    \includegraphics[scale=0.1]{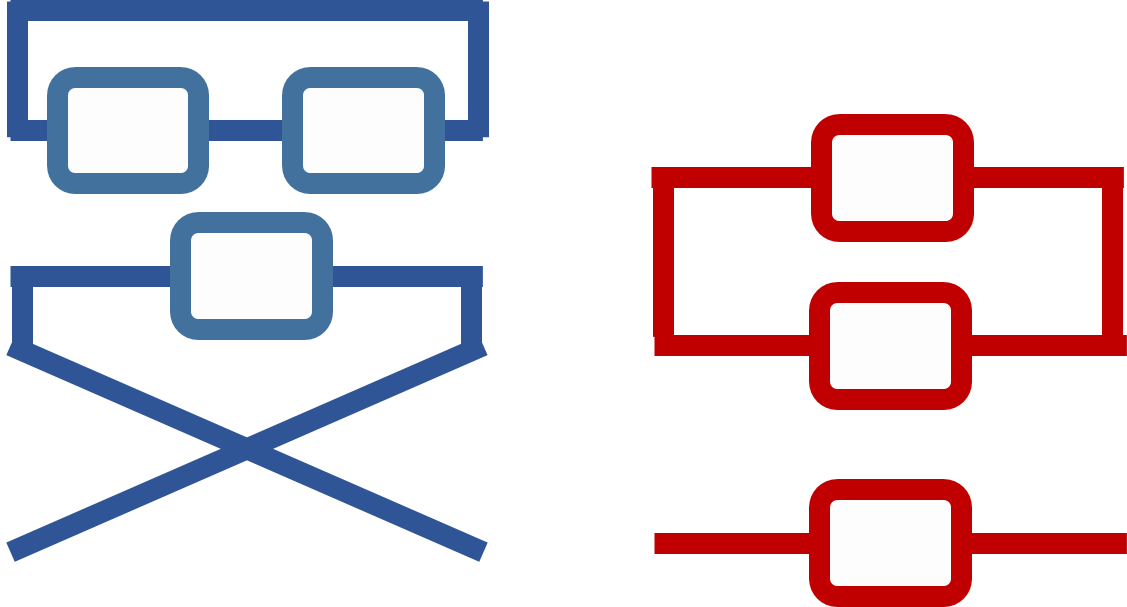}
\end{tabular}
+
\begin{tabular}{c}
    \includegraphics[scale=0.1]{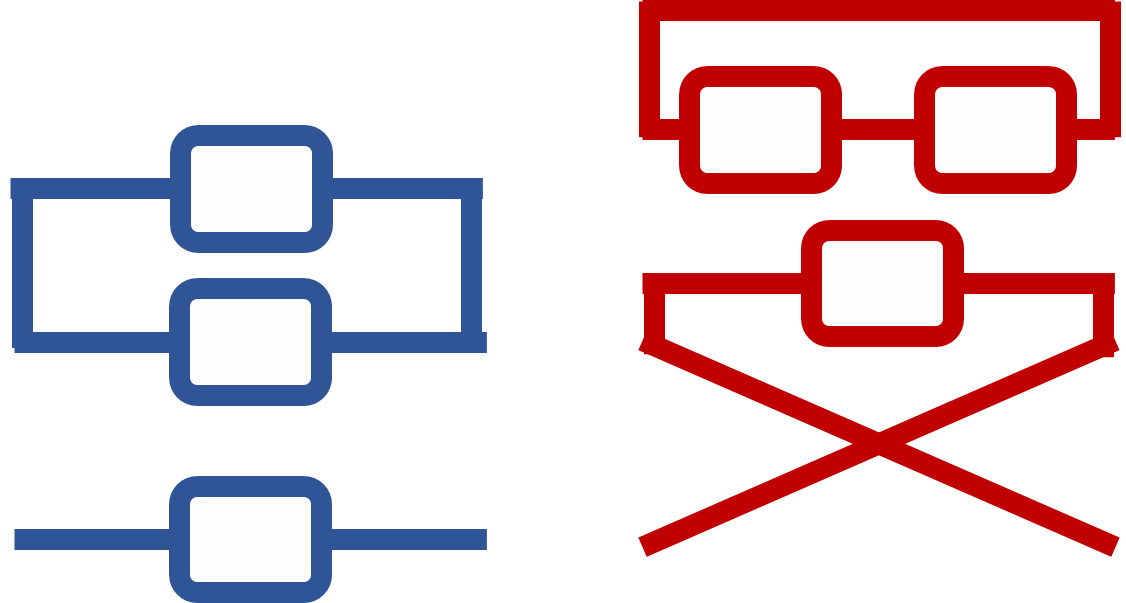}
\end{tabular}
+
\begin{tabular}{c}
    \includegraphics[scale=0.1]{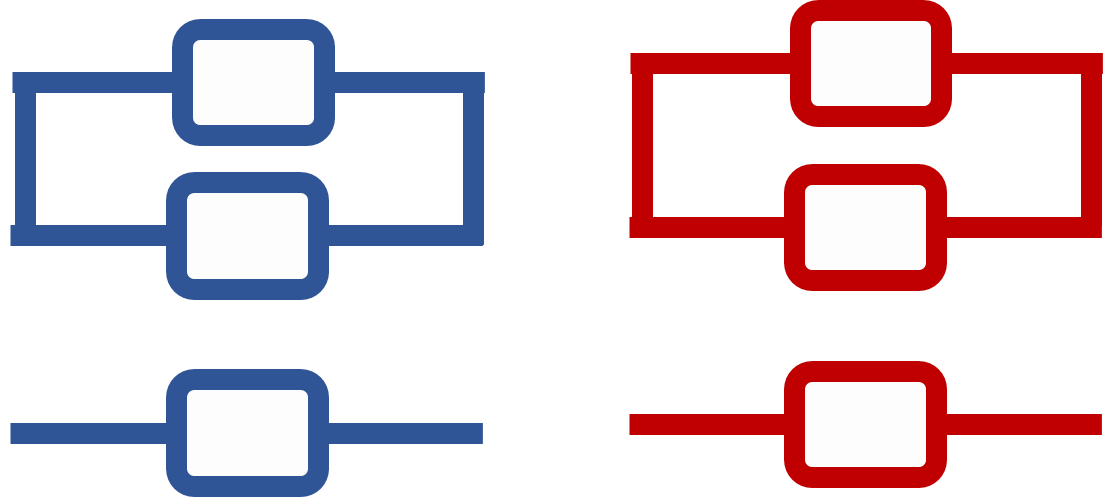}
\end{tabular}
\right)
=
\begin{tabular}{c}
    \includegraphics[scale=0.1]{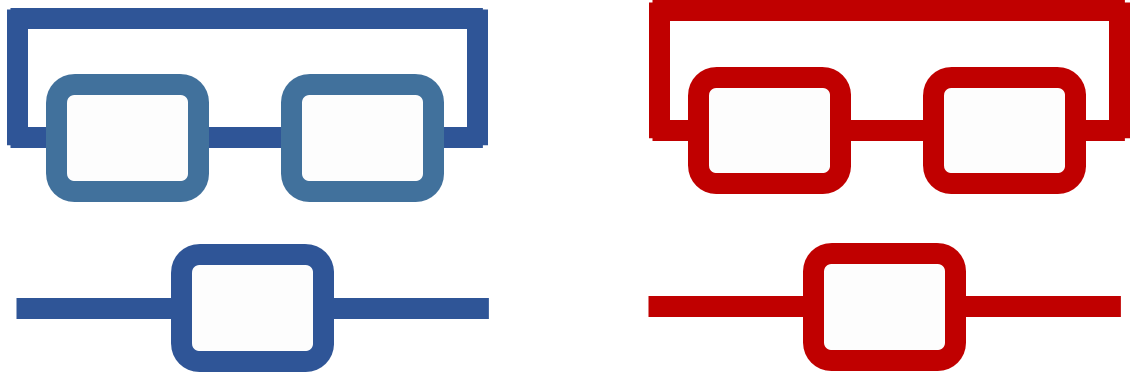}.
\end{tabular}
\end{eqnarray}
Notice that all lines in $\hat{O}_1$ are connected with boxes, so it is easily proved that
\begin{equation}\label{eq:O_1square}
\begin{split}
\tr[\hat{O}_1^2]\le 1.
\end{split}
\end{equation}

Now, substituting Eq.~\eqref{eq:O_4square}, Eq.~\eqref{eq:O_3square}, Eq.\eqref{eq:O_2square}, Eq.~\eqref{eq:O_1square}, and Lemma.~\ref{fact:shadow variance} into Eq.~\eqref{eq:M_4variance}, one can get the upper bound of the variance of the global shadow protocol
\begin{equation}\label{eq:GS variance}
\begin{split}
\mathrm{Var}\left(\hat{M}_4^{(2,3)}\right)\le C_1\frac{1}{M}+C_2\frac{\max\{d_A^2,d_B^2\}M_4^{(2,3)}}{M^2}+C_3\frac{D^2\tr(\rho^2)}{M^3}+C_4\frac{D^4}{M^4}.
\end{split}
\end{equation}
$C_1$, $C_2$, $C_3$ and $C_4$ in this equation are constants independent with number of snapshots $M$ and system dimension $D$.  They are actually not the same constants as those in Eq.~\eqref{eq:M_4variance}, but for simplicity, we use the same notation to represent them. To calculate the variance of local shadow protocol, we remark that
\begin{equation}
\begin{split}
\mathrm{locality}(\hat{O}_4)=4(N_A+N_B),\mathrm{locality}(\hat{O}_3)=3(N_A+N_B),\\
\mathrm{locality}(\hat{O}_2)=2(N_A+N_B),\mathrm{locality}(\hat{O}_1)=N_A+N_B.
\end{split}
\end{equation}
So we just need to make a little adjustment in the variance of global protocol to get the local one
\begin{equation}\label{eq:LS variance}
\begin{split}
\mathrm{Var}\left(\hat{M}_4^{(2,3)}\right)\le C_1\frac{D}{M}+C_2\frac{\max\{d_A^2,d_B^2\}M_4^{(2,3)}D^2}{M^2}+C_3\frac{D^5\tr(\rho^2)}{M^3}+C_4\frac{D^8}{M^4}.
\end{split}
\end{equation}
Therefore, in the large dimension scenario, the variance of global shadow protocol and the local one scale like $D^4/M^4$ and $D^8/M^4$ respectively. substituting this conclusion into the Chebyshev's equation, this proposition is proved. Similar analysis has also been conducted in the estimation of quantum negativity \cite{elben2020mixed} and quantum fisher information \cite{rath2021fisher}.

\end{proof}

\subsection{Randomized Measurements Protocol and Statistical Analysis}\label{subsec:randomprotocol}
The logic of randomized measurement protocol is different with shadow estimation. We first need to write $M_{2n}^{(2,3)}(\rho_{AB})$ in the observable form
\begin{equation}\label{eq:CCNR observable}
\begin{split}
M_{2n}^{(2,3)}(\rho_{AB})=\tr\left[(O_A\otimes O_B)\rho_{AB}^{\otimes 2n}\right],
\end{split}
\end{equation}
where $O_A=\mbb{S}_A^{(1,2)}\otimes\cdots\otimes \mbb{S}_A^{(2n-1,2n)}$ and $O_B=\mbb{S}_B^{(2n,1)}\otimes\cdots\otimes \mbb{S}_B^{(2n-2,2n-1)}$. Notice that the observable is just comprised by SWAP operators. SWAP operator, as the lowest order permutation operators, can be generated by averaging over 2-design global or qubit level local random unitary ensembles among virtual copies. Following the standard procedure of randomized measurements, we can design a protocol to measure $M_{2n}^{(2,3)}(\rho_{AB})$:

\begin{algorithm}[H]
\caption{Measurement protocol of $M_{2n}^{(2,3)}$}
\label{algo:CCNR measurement protocol}
\begin{algorithmic}[1]
\Require
Sequentially prepared $2n\times N_U\times N_M$ $\rho_{AB}$
\Ensure
Probability distribution of the measurement outcomes conditioned on the evolution unitary $\mathrm{Pr}(\vec{s}_A,\vec{s}_B|U_A,U_B)$.

\For{$i= 1~\text{\textbf{to}}~N_U$} 
 \State Construct $2n$ unitary matrix $\{U_1=U_A^1\otimes U_B^1,U_2=U_A^1\otimes U_B^2,\cdots,U_{2n-1}=U_A^n\otimes U_B^n,U_{2n}=U_A^n\otimes U_B^1\}$ with each $U_A^i$ and $U_B^i$ sampled uniformly and independently from $N_A$ and $N_B$ qubits Clifford group (or constructed by tensor product of local unitaries $U_A^i=\bigotimes_{j=1}^{N_A}u_{j}$,$U_B^i=\bigotimes_{j=1}^{N_B}u_j$ with each $u_j$ sampled uniformly from qubit clifford group)
 \State Operate these $2n$ unitaries on $\rho_{AB}$ to get the evolved states $\{U_1\rho_{AB}U_1^\dagger,\cdots,U_{2n}\rho_{AB}U_{2n}^\dagger\}$. 
 \For{$j= 1~\text{\textbf{to}}~N_M$} 
  \State  Measure the evolved states in the computational basis $\{\ket{\vec{s}_A,\vec{s}_B}\}$.
  \State Record the measurement results.
  \EndFor
 \State Estimate the probabilities $\mathrm{Pr}(\vec{s}_A,\vec{s}_B|U_i)$.
 \State Construct an unbiased estimator of $M_{2n}^{(2,3)}$ by these probabilities according to Proposition \ref{prop:CCNR measurement} 
  
\EndFor
\State Take an average of the estimators in each round to get the final estimator of $M_{2n}^{(2,3)}$.
\end{algorithmic}
\end{algorithm}

\begin{proposition}\label{prop:CCNR measurement}
The unbiased estimator of $2n$-th moment of realignment matrix can be constructed using these probabilities $\mathrm{Pr}(\vec{s}_A,\vec{s}_B|U_i)$ 
\begin{equation}\label{eq:CCNR M2n}
\begin{split}
\hat{M}_{2n}^{(2,3)}&=\sum_{\substack{\vec{s}_A^1,\cdots,\vec{s}_A^{2n}\\ \vec{s}_B^1,\cdots,\vec{s}_B^{2n}}}\Omega
\prod_{i=1}^n X_A(\vec{s}_A^{2i-1},\vec{s}_A^{2i})X_B(\vec{s}_B^{2i},\vec{s}_B^{2i+1}),
\end{split}
\end{equation}
where 
\begin{equation}\label{eq:CCNR omega}
\begin{split}
\Omega=\prod_{i=1}^n\mathrm{Pr}(\vec{s}_A^{2i-1},\vec{s}_B^{2i-1}|U_{2i-1})\mathrm{Pr}(\vec{s}_A^{2i},\vec{s}_B^{2i}|U_{2i})
\end{split}
\end{equation}
is the  product of probabilities, $U_i$ is chosen from the prepared unitary evolution group $\{U_1=U_A^1\otimes U_B^1,U_2=U_A^1\otimes U_B^2,\cdots,U_{2n-1}=U_A^n\otimes U_B^n,U_{2n}=U_A^n\otimes U_B^1\}$ as stated in Algo.~\ref{algo:CCNR measurement protocol}, $X(\vec{s},\vec{s}')$ is the weight function. For globally constructed $U_{Ai}$ and $U_{Bi}$
\begin{equation}
\begin{split}
X(\vec{s},\vec{s}')=X_g(\vec{s},\vec{s}')=-(-2^{|\vec{s}|})^{\delta_{\vec{s},\vec{s}'}},
\end{split}
\end{equation}
for locally constructed $U_A^i=\bigotimes_{j=1}^{N_A}u_{j}$,$U_B^i=\bigotimes_{j=1}^{N_B}u_j$
\begin{equation}
\begin{split}
X(\vec{s},\vec{s}')=X_l(\vec{s},\vec{s}')=\prod_{i=1}^{|\vec{s}|}X_g(s_i,s_i')=2^{|\vec{s}|}(-2)^{-D[\vec{s},\vec{s}']},
\end{split}
\end{equation}
where $|\vec{s}|$ is the length of $\vec{s}$, $D[\vec{s},\vec{s}']$ is the Hamming distance between $\vec{s}$ and $\vec{s}'$. In Eq.~\eqref{eq:CCNR M2n}, we set $\vec{s}_B^{2n+1}=\vec{s}_B^1$.

The estimator constructed above satisfies
\begin{equation}
\begin{split}
    \mathbb{E}_{U,\vec{s}}\left(\hat{M}_{2n}^{(2,3)}\right)&=M_{2n}^{(2,3)}
\end{split}
\end{equation}
where $\mbb{E}_{U,\vec{s}}$ denotes the expectation over all Clifford groups and measurement results.

\end{proposition}

\begin{proof}
To prove this proposition, we define the postprocessing operator $X_{g/l}=\sum_{\vec{s},\vec{s}'}X_{g/l}(\vec{s},\vec{s}')\ketbra{\vec{s},\vec{s}'}{\vec{s},\vec{s}'}$, which can be used to generate SWAP operator \cite{brydges2019probing,elben2019toolbox} through a twirling channel,
\begin{equation}\label{eq:SWAP construct global}
\begin{split}
\Phi_2(X_g)=\mathbb{E}_{U\in\mathcal{E}}U^{\otimes 2}X_g U^{\dagger\otimes 2}=\mathbb{S},
\end{split}
\end{equation}
where $\mathcal{E}$ is a unitary 2-design, which is fulfilled by Clifford group.  Substituting Born's rule, $\mathbb{E}_{\vec{s}}\mathrm{Pr}(\vec{s}_A,\vec{s}_B|U_A,U_B)=\tr\left[\ketbra{\vec{s}_A,\vec{s}_B}{\vec{s}_A,\vec{s}_B}(U_A\otimes U_B)\rho_{AB}(U_A^\dagger\otimes U_B^\dagger)\right]$, into the right hand side of Eq.~\eqref{eq:CCNR M2n}, and consider the global case,
\begin{equation}
\begin{split}
R.H.S&=\sum_{\substack{\vec{s}_A^1,\cdots,\vec{s}_A^{2n}\\ \vec{s}_B^1,\cdots,\vec{s}_B^{2n}}}\mbb{E}_{U}\prod_{i=1}^n X_A(\vec{s}_A^{2i-1},\vec{s}_A^{2i})X_B(\vec{s}_B^{2i},\vec{s}_B^{2i+1})\\
&\tr\left\{\left(\bigotimes_{i=1}^n\ketbra{\vec{s}_A^{2i-1}\vec{s}_A^{2i}}{\vec{s}_A^{2i-1}\vec{s}_A^{2i}}\otimes\ketbra{\vec{s}_B^{2i-1}\vec{s}_B^{2i}}{\vec{s}_B^{2i-1}\vec{s}_B^{2i}}\right)\left(\bigotimes_{i=1}^n U_i\right)\rho_{AB}^{\otimes 2n}\left(\bigotimes_{i=1}^n U_i\right)^\dagger\right\}\\
&=\mbb{E}_{U}\tr\left\{\left(\bigotimes_{i=1}^n X_A^{(2i-1,2i)}\otimes X_B^{(2i,2i+1)}\right)\left(\bigotimes_{i=1}^n U_i\right)\rho_{AB}^{\otimes 2n}\left(\bigotimes_{i=1}^n U_i\right)^\dagger\right\}\\
&=\mbb{E}_{U}\tr\left\{\left(\bigotimes_{i=1}^n U_i\right)^\dagger\left(\bigotimes_{i=1}^n X_A^{(2i-1,2i)}\otimes X_B^{(2i,2i+1)}\right)\left(\bigotimes_{i=1}^n U_i\right)\rho_{AB}^{\otimes 2n}\right\}\\
&=\tr\left\{\left(\bigotimes_{i=1}^n \Phi_2(X_A^{(2i-1,2i)})\otimes \Phi_2(X_B^{(2i,2i+1)})\right)\rho_{AB}^{\otimes 2n}\right\}\\
&=\tr\left\{\left(\bigotimes_{i=1}^n \mathbb{S}_A^{(2i-1,2i)}\otimes\mathbb{S}_B^{(2i,2i+1)} \right)\rho_{AB}^{\otimes 2n}\right\}\\
&=M_{2n}^{(2,3)}.
\end{split}
\end{equation}
Here we use subscript to denote the party and superscript to denote the copy number, like $X_A^{(2i-1,2i)}$ is the operator acting on the $A$ parties of the $(2i-1)$-th and $2i$-th copies of $\rho_{AB}$.

The proof of local case is quite similar to global one. Notice that $X_l=\bigotimes_{i=1}^N X_{gi}$, hence
\begin{equation}\label{eq:SWAP construct local}
\begin{split}
\mathbb{E}_U\left(\bigotimes_{i=1}^Nu_i\right)^{\otimes 2}X_l\left(\bigotimes_{i=1}^Nu_i\right)^{\dagger\otimes 2}=\mathbb{E}_U\left(\bigotimes_{i=1}^Nu_i^{\otimes 2}X_{gi}u_i^{\dagger \otimes 2}\right)=\bigotimes_{i=1}^N\Phi_2(X_{gi})=\bigotimes_{i=1}^N\mathbb{S}_i=\mathbb{S}.
\end{split}
\end{equation}
Then, Eq.~\eqref{eq:CCNR M2n} also holds for local unitary.

\end{proof}

\begin{figure}[htbp]
    \centering
    \includegraphics[scale=0.14]{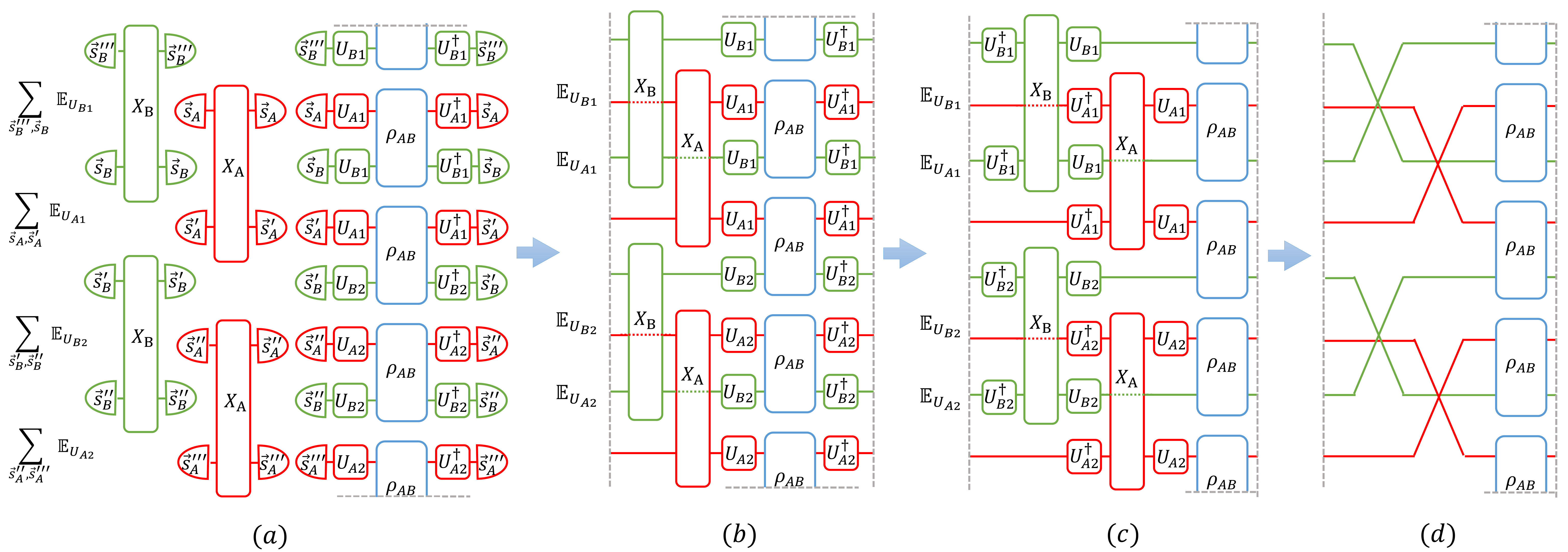}
    \caption{Graphical illustration of the proof of Proposition \ref{prop:CCNR measurement}. The grey dashed lines denote the periodic boundary condition and the colored dashed lines in (b) and (c) indicate that the corresponding indices represented by the lines do not contract with the tensors, they just pass through them. $X_A$ and $X_B$ are the data postprocessing operators, which are diagonal, resulting in the equivalence of (a) and (b).  In (d), we use the X-shaped signs to denote the SWAP operators.}
    \label{fig:proof of M4}
\end{figure}

Here we provide a graphical illustration to sketch the proof, Fig.~\ref{fig:proof of M4}, taking $M_4^{(2,3)}$ as an example. By multiplying the probabilities for four times, $\rho_{AB}^{\otimes 4}$ is introduced into the formula. The postprocessing function, together with measurement bases, constructs the diagonal postprocessing operator $X_A$ and $X_B$. Thus, $\hat{M}_4^{(2,3)}$ is equivalent to the measurement of diagonal observable $X$ on the four evolved states, see Fig.~\ref{fig:proof of M4}(b). Then, the unitary matrices can be moved to the both sides of $X$ 
based on the property of trace function, see Fig.~\ref{fig:proof of M4}(c). After calculating the expectation over unitary group, $X$ operators are turned into SWAP operators according to Eq.~\eqref{eq:SWAP construct global} and Eq.~\eqref{eq:SWAP construct local}, and $M_4^{(2,3)} = \tr[\left(\mathbb{S}_A^{(1,2)}\otimes\mathbb{S}_A^{(3,4)}\otimes\mathbb{S}_B^{(2,3)}\otimes\mathbb{S}_B^{(4,1)}\right)\rho_{AB}^{\otimes 4}]$ is successfully constructed, see Fig.~\ref{fig:proof of M4}(d).  

Now, we also take $M_4^{(2,3)}$ as an example to analyze the sample complexity of this randomized measurement protocol. The general case will be discussed in next section. To simplify our analysis, here we adopt another representation of $\hat{M}_4^{(2,3)}$. As stated in Algorithm \ref{algo:CCNR measurement protocol}, the whole process contains $N_U$ rounds of experiments, and in each round of experiment, we record the measurement results and label them as $\{r_{A1,B1}^1,\cdots,r_{A1,B1}^{N_M}\}$, ..., $\{r_{A2,B1}^1,\cdots,r_{A2,B1}^{N_M}\}$, where $r_{A_i,B_j}^k$ is the $k$-th measurement result measuring state operated by $U_{Ai}\otimes U_{Bj}$. Using these data, the unbiased estimator of $M_4^{(2,3)}$ can be alternatively represented as
\begin{equation}\label{eq:M4estimator2}
\begin{split}
\hat{M}_4^{(2,3)}=\frac{1}{N_M^4}\sum_{i,j,k,l=1}^{N_M}\tr\left[\hat{Q}\left(\hat{r}_{A1,B1}^i\otimes\hat{r}_{A1,B2}^j\otimes\hat{r}_{A2,B2}^k\otimes\hat{r}_{A2,B1}^l\right)\right].
\end{split}
\end{equation}
After $N_U$ independent rounds of experiments, we take an average of estimators in each round to get our final estimator. Here, we use rank-1 matrices $\hat{r}=\ketbra{r}{r}$ to label our measurement results. $\hat{r}_{A1,B1}^i$, $\hat{r}_{A1,B2}^j$, $\hat{r}_{A2,B2}^k$ and $\hat{r}_{A2,B1}^l$ together span a $4$-fold Hilbert space $\mathcal{H}_A^{\otimes 4}\otimes\mathcal{H}_B^{\otimes 4}$. $\hat{Q}=X_A^{(1,2)}\otimes X_A^{(3,4)}\otimes X_B^{(2,3)}\otimes X_B^{(4,1)}$ is the postprocessing operator. In fact, one can easily verify that the definition of $\hat{M}_4^{(2,3)}$ in Eq.\eqref{eq:CCNR M2n} is actually the same as the definition in Eq.\eqref{eq:M4estimator2}. After analytical analysis, we find

\begin{proposition}\label{prop:random variance}
To make sure the estimator of $M_4^{(2,3)}$ defined in Proposition.~\ref{prop:CCNR measurement} satisfies $|\hat{M}_4^{(2,3)}-M_4^{(2,3)}|\le\epsilon$ with probability at least $1-\delta$, the total times of experiments scale like
\begin{equation}
\begin{split}
4\times N_U\times N_M=O(\frac{D^{\frac{1}{2}}}{\epsilon^2\delta})
\end{split}
\end{equation}
for global random protocol, and 
\begin{equation}
\begin{split}
4\times N_U\times N_M=O(\frac{D^{1.187}}{\epsilon^2\delta})
\end{split}
\end{equation}
for local random protocol.
\end{proposition}
\begin{proof}

To verify this proposition, we also need to evaluate the variance of this estimator
\begin{equation}
\begin{split}
\mathrm{Var}(\hat{M}_4^{(2,3)})=\mathbb{E}\left[\left(\hat{M}_4^{(2,3)}\right)^2\right]-\left(M_4^{(2,3)}\right)^2.
\end{split}
\end{equation}
$\left(\hat{M}_4^{(2,3)}\right)^2$ can be decomposed as
\begin{equation}\label{eq:randomM4}
\begin{split}
\left(\hat{M}_4^{(2,3)}\right)^2=\frac{1}{N_M^8}\sum_{i,j,k,l=1}^{N_M}\sum_{i',j',k',l'=1}^{N_M}\tr\left[\hat{Q}\left(\hat{r}_{A1,B1}^i\otimes\hat{r}_{A1,B2}^j\otimes\hat{r}_{A2,B2}^k\otimes\hat{r}_{A2,B1}^l\right)\right]\tr\left[\hat{Q}\left(\hat{r}_{A1,B1}^{i'}\otimes\hat{r}_{A1,B2}^{j'}\otimes\hat{r}_{A2,B2}^{k'}\otimes\hat{r}_{A2,B1}^{l'}\right)\right].
\end{split}
\end{equation}
Following the similar idea of the analysis of the variance of shadow protocol, these $N_M^8$ terms can be divided into several groups by the relation between the indices $i$ and $i'$; $j$ and $j'$; $k$ and $k'$; and $l$ and $l'$. First, because $\hat{Q}$ is a diagonal matrix in the computational basis, we have
\begin{equation}
\begin{split}
\tr\left[\hat{Q}\left(\hat{r}_{A1,B1}^i\otimes\hat{r}_{A1,B2}^j\otimes\hat{r}_{A2,B2}^k\otimes\hat{r}_{A2,B1}^l\right)\right]^2=\tr\left[\hat{Q}^2\left(\hat{r}_{A1,B1}^i\otimes\hat{r}_{A1,B2}^j\otimes\hat{r}_{A2,B2}^k\otimes\hat{r}_{A2,B1}^l\right)\right].
\end{split}
\end{equation}
Then, it is helpful to introduce a new processing operator
\begin{equation}
\begin{split}
X_{g/l}^{(1,2,3)}=\left(X_{g/l}^{(1,2)}\otimes\mathbb{I}^{3}\right)\left(\mathbb{I}^{1}\otimes X_{g/l}^{(2,3)}\right)=\sum_{\vec{s}^1,\vec{s}^3}\sum_{\vec{s}_2}X_{g/l}(\vec{s}^1,\vec{s}^2)X_{g/l}(\vec{s}^2,\vec{s}^3)\ketbra{\vec{s}^1,\vec{s}^2,\vec{s}^3}{\vec{s}^1,\vec{s}^2,\vec{s}^3},
\end{split}
\end{equation}
which satisfies
\begin{equation}
\begin{split}
\tr\left[X^{(1,2)}\left(\hat{r}^i\otimes\hat{r}^j\right)\right]\tr\left[X^{(1,2)}\left(\hat{r}^j\otimes\hat{r}^k\right)\right]=\tr\left[X^{(1,2,3)}\left(\hat{r}^i\otimes\hat{r}^j\otimes\hat{r}^k\right)\right]
\end{split}
\end{equation}
when $i$ is not equal to $k$. Lastly, one can easily prove that
\begin{equation}
\begin{split}
\tr\left[X^{(1,2)}\left(\hat{r}^i\otimes\hat{r}^j\right)\right]\tr\left[X^{(1,2)}\left(\hat{r}^{i'}\otimes\hat{r}^{j'}\right)\right]=\tr\left[\left(X^{(1,2)}\otimes X^{(3,4)}\right)\left(\hat{r}^i\otimes\hat{r}^j\otimes\hat{r}^{i'}\otimes\hat{r}^{j'}\right)\right]
\end{split}
\end{equation}
for $i\neq i'$ and $j\neq j'$. Thus, Eq.~\eqref{eq:randomM4} can be rewritten as
\begin{equation}\label{eq:randomM_4square}
\begin{split}
\left(\hat{M}_4^{(2,3)}\right)^2=&\frac{1}{N_M^8}\sum_{i,j,k,l=1}^{N_M}\tr\left[\hat{Q}_0\left(\hat{r}_{A1,B1}^i\otimes\hat{r}_{A1,B2}^j\otimes\hat{r}_{A2,B2}^k\otimes\hat{r}_{A2,B1}^l\right)\right]\\
+&\frac{4}{N_M^8}\sum_{i,j,k,l=1}^{N_M}\sum_{i'\neq i}\tr\left[\hat{Q}_1\left(\hat{r}_{A1,B1}^i\otimes\hat{r}_{A1,B1}^{i'}\otimes\hat{r}_{A1,B2}^j\otimes\hat{r}_{A2,B2}^k\otimes\hat{r}_{A2,B1}^l\right)\right]\\
+&\frac{2}{N_M^8}\sum_{i,j,k,l=1}^{N_M}\sum_{i'\neq i,j'\neq j}\tr\left[\hat{Q}_2\left(\hat{r}_{A1,B1}^i\otimes\hat{r}_{A1,B1}^{i'}\otimes\hat{r}_{A1,B2}^j\otimes\hat{r}_{A1,B2}^{j'}\otimes\hat{r}_{A2,B2}^k\otimes\hat{r}_{A2,B1}^l\right)\right]\\
+&\frac{2}{N_M^8}\sum_{i,j,k,l=1}^{N_M}\sum_{i'\neq i,k'\neq k}\tr\left[\hat{Q}_2'\left(\hat{r}_{A1,B1}^i\otimes\hat{r}_{A1,B1}^{i'}\otimes\hat{r}_{A1,B2}^j\otimes\hat{r}_{A2,B2}^k\otimes\hat{r}_{A2,B2}^{k'}\otimes\hat{r}_{A2,B1}^l\right)\right]\\
+&\frac{2}{N_M^8}\sum_{i,j,k,l=1}^{N_M}\sum_{i'\neq i,l'\neq l}\tr\left[\hat{Q}_2''\left(\hat{r}_{A1,B1}^i\otimes\hat{r}_{A1,B1}^{i'}\otimes\hat{r}_{A1,B2}^j\otimes\hat{r}_{A2,B2}^k\otimes\hat{r}_{A2,B1}^l\otimes\hat{r}_{A2,B1}^{l'}\right)\right]\\
+&\frac{4}{N_M^8}\sum_{i,j,k,l=1}^{N_M}\sum_{\substack{i'\neq i,j'\neq j,\\k'\neq k}}\tr\left[\hat{Q}_3\left(\hat{r}_{A1,B1}^i\otimes\hat{r}_{A1,B1}^{i'}\otimes\hat{r}_{A1,B2}^j\otimes\hat{r}_{A1,B2}^{j'}\otimes\hat{r}_{A2,B2}^k\otimes\hat{r}_{A2,B2}^{k'}\otimes\hat{r}_{A2,B1}^l\right)\right]\\
+&\frac{1}{N_M^4}\sum_{i,j,k,l=1}^{N_M}\sum_{\substack{i'\neq i,j'\neq j,\\ k'\neq k,l'\neq l}}\tr\left[\hat{Q}_4\left(\hat{r}_{A1,B1}^i\otimes\hat{r}_{A1,B1}^{i'}\otimes\hat{r}_{A1,B2}^j\otimes\hat{r}_{A1,B2}^{j'}\otimes\hat{r}_{A2,B2}^k\otimes\hat{r}_{A2,B2}^{k'}\otimes\hat{r}_{A2,B1}^l\otimes\hat{r}_{A2,B1}^{l'}\right)\right].
\end{split}
\end{equation}
In this equation, different rows correspond to different number of equal indices, like the first row stands for $i=i'$, $j=j'$, $k=k'$ and $l=l'$. And the operators appear in Eq.~\eqref{eq:randomM_4square} are 
\begin{equation}
\begin{split}
\hat{Q}_0&=X_A^{(1,2)2}\otimes X_A^{(3,4)2}\otimes X_B^{(2,3)2}\otimes X_B^{(4,1)2}\\
\hat{Q}_1&=X_A^{(1,3,2)}\otimes X_A^{(4,5)2}\otimes X_B^{(3,4)2}\otimes X_B^{(1,5,2)}\\
\hat{Q}_2&=X_A^{(1,3)}\otimes X_A^{(2,4)}\otimes X_A^{(5,6)2}\otimes X_B^{(1,6,2)}\otimes X_B^{(3,5,4)}\\
\hat{Q}_2'&=X_A^{(1,3,2)}\otimes X_A^{(4,6,5)}\otimes X_B^{(4,3,5)}\otimes X_B^{(1,6,2)}\\
\hat{Q}_2''&=X_A^{(1,3,2)}\otimes X_A^{(5,4,6)}\otimes X_B^{(3,4)2}\otimes X_B^{(5,1)}\otimes X_B^{(6,2)}\\
\hat{Q}_3&=X_A^{(1,3)}\otimes X_A^{(2,4)}\otimes X_A^{(5,7,6)}\otimes X_B^{(3,5)}\otimes X_B^{(4,6)}\otimes X_B^{(1,7,2)}\\
\hat{Q}_4&=X_A^{(1,3)}\otimes X_A^{(2,4)}\otimes X_A^{(5,7)}\otimes X_A^{(6,8)}\otimes X_B^{(3,5)}\otimes X_B^{(4,6)}\otimes X_B^{(7,1)}\otimes X_B^{(8,2)}
\end{split}
\end{equation}
respectively. Taking an average of it, we get
\begin{equation}
\begin{split}
\mathbb{E}_{U,r}\left[\left(\hat{M}_4^{(2,3)}\right)^2\right]&=\frac{1}{N_M^4}\mathbb{E}_{U,r}\left\{\tr\left[\hat{Q}_0\left(\hat{r}_{A1,B1}\otimes\hat{r}_{A1,B2}\otimes\hat{r}_{A2,B2}\otimes\hat{r}_{A2,B1}\right)\right]\right\}\\
&+\frac{4(N_M-1)}{N_M^4}\mathbb{E}_{U,r}\left\{\tr\left[\hat{Q}_1\left(\hat{r}_{A1,B1}\otimes\hat{r}_{A1,B1}'\otimes\hat{r}_{A1,B2}\otimes\hat{r}_{A2,B2}\otimes\hat{r}_{A2,B1}\right)\right]\right\}\\
&+\frac{2(N_M-1)^2}{N_M^4}\mathbb{E}_{U,r}\left\{\tr\left[\hat{Q}_2\left(\hat{r}_{A1,B1}\otimes\hat{r}_{A1,B1}'\otimes\hat{r}_{A1,B2}\otimes\hat{r}_{A1,B2}'\otimes\hat{r}_{A2,B2}\otimes\hat{r}_{A2,B1}\right)\right]\right\}\\
&+\frac{2(N_M-1)^2}{N_M^4}\mathbb{E}_{U,r}\left\{\tr\left[\hat{Q}_2'\left(\hat{r}_{A1,B1}\otimes\hat{r}_{A1,B1}'\otimes\hat{r}_{A1,B2}\otimes\hat{r}_{A2,B2}\otimes\hat{r}_{A2,B2}'\otimes\hat{r}_{A2,B1}\right)\right]\right\}\\
&+\frac{2(N_M-1)^2}{N_M^4}\mathbb{E}_{U,r}\left\{\tr\left[\hat{Q}_2''\left(\hat{r}_{A1,B1}\otimes\hat{r}_{A1,B1}'\otimes\hat{r}_{A1,B2}\otimes\hat{r}_{A2,B2}\otimes\hat{r}_{A2,B1}\otimes\hat{r}_{A2,B1}'\right)\right]\right\}\\
&+\frac{4(N_M-1)^3}{N_M^4}\mathbb{E}_{U,r}\left\{\tr\left[\hat{Q}_3\left(\hat{r}_{A1,B1}\otimes\hat{r}_{A1,B1}'\otimes\hat{r}_{A1,B2}\otimes\hat{r}_{A1,B2}'\otimes\hat{r}_{A2,B2}\otimes\hat{r}_{A2,B2}'\otimes\hat{r}_{A2,B1}\right)\right]\right\}\\
&+\frac{(N_M-1)^4}{N_M^4}\mathbb{E}_{U,r}\left\{\tr\left[\hat{Q}_4\left(\hat{r}_{A1,B1}\otimes\hat{r}_{A1,B1}'\otimes\hat{r}_{A1,B2}\otimes\hat{r}_{A1,B2}'\otimes\hat{r}_{A2,B2}\otimes\hat{r}_{A2,B2}'\otimes\hat{r}_{A2,B1}\otimes\hat{r}_{A2,B1}'\right)\right]\right\},
\end{split}
\end{equation}
where $\mathbb{E}_{U,r}$ denotes taking average over both random unitary and measurement results. Taking $\hat{Q}_0$ term as an example, substituting Born's rule,
\begin{equation}
\begin{split}
&\mathbb{E}_{U,r}\left\{\tr\left[\hat{Q}_0\left(\hat{r}_{A1,B1}\otimes\hat{r}_{A1,B2}\otimes\hat{r}_{A2,B2}\otimes\hat{r}_{A2,B1}\right)\right]\right\}\\
=&\mathbb{E}_U\sum_{\substack{r_{A1,B1},r_{A1,B2},\\r_{A2,B2},r_{A2,B1}}}\tr\left[\hat{r}_{A1,B1}(U_{A1}\otimes U_{B1})\rho(U_{A1}\otimes U_{B1})^{\dagger}\right]\tr\left[\hat{r}_{A1,B2}(U_{A1}\otimes U_{B2})\rho(U_{A1}\otimes U_{B2})^{\dagger}\right]\\
&\tr\left[\hat{r}_{A2,B2}(U_{A2}\otimes U_{B2})\rho(U_{A2}\otimes U_{B2})^{\dagger}\right]\tr\left[\hat{r}_{A2,B1}(U_{A2}\otimes U_{B1})\rho(U_{A2}\otimes U_{B1})^{\dagger}\right]\\
&\tr\left[\hat{Q}_0\left(\hat{r}_{A1,B1}\otimes\hat{r}_{A1,B2}\otimes\hat{r}_{A2,B2}\otimes\hat{r}_{A2,B1}\right)\right]\\
=&\mathbb{E}_U\tr\left[\hat{Q}_0\left(U_{A1}^{\otimes 2}\otimes U_{A2}^{\otimes 2}\otimes U_{B1}^{\otimes2}\otimes U_{B2}^{\otimes 2}\right)\rho^{\otimes4}\left(U_{A1}^{\otimes 2}\otimes U_{A2}^{\otimes 2}\otimes U_{B1}^{\otimes2}\otimes U_{B2}^{\otimes 2}\right)^\dagger\right]\\
=&\mathbb{E}_U\tr\left[\left(U_{A1}^{\otimes 2\dagger} X_A^{(1,2)2}U_{A1}^{\otimes 2}\right)\otimes\left(U_{A2}^{\otimes 2\dagger} X_A^{(3,4)2}U_{A2}^{\otimes 2}\right)\otimes\left(U_{B1}^{\otimes2\dagger} X_B^{(2,3)2}U_{B1}^{\otimes2}\right)\otimes\left(U_{B2}^{\otimes2\dagger} X_B^{(4,1)2}U_{B2}^{\otimes2}\right)\rho^{\otimes 4}\right]\\
=&\tr\left[\left(\Phi_2(X_A^{(1,2)2})^{\otimes 2}\otimes\Phi_2(X_B^{(1,2)2})^{\otimes 2}\right)\rho^{\otimes 4}\right].
\end{split}
\end{equation}
Following the same idea, $\mathbb{E}_{U,r}\left[\left(\hat{M}_4^{(2,3)}\right)^2\right]$ can be rewritten as
\begin{equation}\label{eq:M_4squareexpectation}
\begin{split}
\mathbb{E}_{U,r}\left[\left(\hat{M}_4^{(2,3)}\right)^2\right]&=\frac{1}{N_M^4}\tr\left[\left(\Phi_2(X_A^{(1,2)2})^{\otimes 2}\otimes\Phi_2(X_B^{(1,2)2})^{\otimes 2}\right)\rho^{\otimes 4}\right]\\
&+\frac{4(N_M-1)}{N_M^4}\tr\left[\left(\Phi_3(X_A^{(1,2,3)})\otimes\Phi_2(X_A^{(1,2)2})\otimes\Phi_2(X_B^{(1,2)2})\otimes\Phi_3(X_B^{(1,2,3)})\right)\rho^{\otimes 5}\right]\\
&+\frac{2(N_M-1)^2}{N_M^4}\tr\left[\Phi_4(X_A^{(1,2)\otimes 2})\otimes\Phi_2(X_A^{(1,2)2})\otimes\Phi_3(X_B^{(1,2,3)})^{\otimes 2}\rho^{\otimes 6}\right]\\
&+\frac{2(N_M-1)^2}{N_M^4}\tr\left[\Phi_3(X_A^{(1,2,3)})^{\otimes2} \otimes\Phi_3(X_B^{(1,2,3)})^{\otimes 2}\rho^{\otimes 6}\right]\\
&+\frac{2(N_M-1)^2}{N_M^4}\tr\left[\Phi_3(X_A^{(1,2,3)})^{\otimes2}\otimes\Phi_2(X_B^{(1,2)2})\otimes\Phi_4(X_B^{(1,2)\otimes 2})\rho^{\otimes 6}\right]\\
&+\frac{4(N_M-1)^3}{N_M^4}\tr\left[\left(\Phi_4(X_A^{(1,2)\otimes 2})\otimes\Phi_3(X_A^{(1,2,3)})\otimes\Phi_4(X_B^{(1,2)\otimes 2})\otimes\Phi_3(X_B^{(1,2,3)})\right)\rho^{\otimes 7}\right]\\
&+\frac{(N_M-1)^4}{N_M^4}\tr\left[\left(\Phi_4(X_A^{(1,2)\otimes 2})^{\otimes 2}\otimes \Phi_4(X_B^{(1,2)\otimes 2})^{\otimes 2}\right)\rho^{\otimes 8}\right].
\end{split}
\end{equation}
In this equation, we make the assumption that the unitary ensemble we use is at least a unitary 4-design, which is in fact not fulfilled by Clifford group \cite{zhu2017multiqubit}. However, this assumption does not affect the leading term (we will see later it is indeed the leading term) in the above equation, $\tr\left[\left(\Phi_2(X_A^{(1,2)2})^{\otimes 2}\otimes\Phi_2(X_B^{(1,2)2})^{\otimes 2}\right)\rho^{\otimes 4}\right]/N_M^4$, and our result meets well with the numerical results. 

According to the definition of twirling channel, introduced in Sec.~\ref{subsec:randomprotocol}, $\Phi_t(\cdot)$ terms in Eq.~\eqref{eq:M_4squareexpectation} can be derived,
\begin{equation}\label{eq:twirling conclusion}
\begin{split}
\Phi_2(X^{(1,2)2})&=D\hat{W}_{(1,2)}+(D-1)\hat{W}_{(2,1)}=D\mathbb{I}+(D-1)\mathbb{S}\\
\Phi_3(X^{(1,2,3)})&=-\frac{1}{D+2}\left(\hat{W}_{(1,2,3)}+\hat{W}_{(1,3,2)}+\hat{W}_{(2,1,3)}\right)+\frac{D+1}{D+2}\left(\hat{W}_{(3,2,1)}+\hat{W}_{(3,1,2)}+\hat{W}_{(2,3,1)}\right)\\
\Phi_4(X^{(1,2)\otimes2})&=\frac{2}{D(D+2)(D+3)}\left(\hat{W}_{(1,2,3,4)}+\hat{W}_{(1,2,4,3)}+\hat{W}_{(2,1,3,4)}\right)\\
&-\frac{D+1}{D(D+2)(D+3)}\left(\hat{W}_{(1,3,2,4)}+\hat{W}_{((1,3,4,2))}+\hat{W}_{(1,4,2,3)}+\hat{W}_{(1,4,3,2)}+\hat{W}_{(2,3,1,4)}+\hat{W}_{(2,3,4,1)}+\cdots\right)\\
&+\frac{D+1}{D(D+3)}\left(\hat{W}_{(3,4,1,2)}+\hat{W}_{(3,4,2,1)}+\hat{W}_{(4,3,1,2)}+\hat{W}_{(4,3,2,1)}\right)\\
&+\frac{D(D+2)(D+3)+2}{D(D+2)(D+3)}\hat{W}_{(2,1,4,3)}.
\end{split}
\end{equation}
Here, the subscript $(i,j,k,l)$ represents the permutation operator that can convert $(1,2,3,4)$ to it. In the large dimension regime, $D\gg 1$, we only keep the leading terms,
\begin{equation}
\begin{split}
\Phi_2(X^{(1,2)2})&=D\hat{W}_{(1,2)}+(D-1)\hat{W}_{(2,1)}\approx D\mathbb{I}+D\mathbb{S}\\
\Phi_3(X^{(1,2,3)})&\approx\hat{W}_{(3,2,1)}+\hat{W}_{(3,1,2)}+\hat{W}_{(2,3,1)}\\
\Phi_4(X^{(1,2)\otimes2})&\approx\hat{W}_{(2,1,4,3)}.
\end{split}
\end{equation}
Substituting it into Eq.~\eqref{eq:M_4squareexpectation}, we get
\begin{equation}
\begin{split}
\mathrm{Var}\left(\hat{M}_4\right)\leq\frac{D^2f_0(\rho)}{N_M^4}+\frac{D f_1(\rho)}{N_M^3}+\frac{d_Af_2(\rho)}{N_M^2}+\frac{f_2'(\rho)}{N_M^2}+\frac{d_Bf_2''(\rho)}{N_M^2}+\frac{f_3(\rho)}{N_M}+f_4(\rho),
\end{split}
\end{equation}
where $f_0(\rho)$, ..., $f_4(\rho)$ are functions of $\rho$, bounded by some constants independent of dimension $D$ (This is because $|\tr(\hat{W}\rho^{\otimes k})|\leq 1$ for any permutation operator $\hat{W}$). Considering the average over different random unitary, the total variance is
\begin{equation}
\begin{split}
\frac{1}{N_U}\left\{\frac{D^2f_0(\rho)}{N_M^4}+\frac{D f_1(\rho)}{N_M^3}+\frac{d_Af_2(\rho)}{N_M^2}+\frac{f_2'(\rho)}{N_M^2}+\frac{d_Bf_2''(\rho)}{N_M^2}+\frac{f_3(\rho)}{N_M}+f_4(\rho)\right\}.
\end{split}
\end{equation}
Thus, to make sure the estimator constructed by global random protocol satisfies $|\hat{M}_4^{(2,3)}-M_4^{(2,3)}|\le \epsilon$ with probability at least $1-\delta$, $N_U$ and $N_M$ scales as
\begin{equation}
\begin{split}
N_U=O(\frac{1}{\epsilon^2\delta}),N_M=O(D^{1/2})
\end{split}
\end{equation}

The variance of local protocol is not easy to compute for general case, so we recall that in global protocol, the leading term (which determines the scaling of $N_M$) is the first term $\tr[\Phi_2(X_A^{(1,2)2})^{\otimes 2}\Phi_2(X_B^{(1,2)2})^{\otimes 2}\rho^{\otimes 4}]/N_M^4$. We assume this result also holds in local protocol, so we need to choose a state to maximize the leading term. From Eq.~\eqref{eq:twirling conclusion}, we know that $\Phi_2(X_A^{(1,2)2})^{\otimes 2}\Phi_2(X_B^{(1,2)2})^{\otimes 2}$ is composed of many permutation operators with positive coefficients, so that the maximum of the leading term is reached when $\rho=\ketbra{\psi}{\psi}_A\otimes\ketbra{\varphi}{\varphi}_B$ (according to the mathematical property of permutation operator $|\tr(\hat{W}_\sigma\rho)|\le 1$, and pure product state can reach this maximum). Here, we make a reasonable assumption that, even in local case, the leading term is also the first term. Hence, if we successfully find a state that maximize the leading term, the variance computed using this state will offer an upper bound of the variance in general case.

In local protocol, $X^{(1,2)}$, $X^{(1,2,3)}$ are turned into two qubit and three qubit operators respectively. While for the leading term $\tr[\Phi_2(X_A^{(1,2)2})^{\otimes 2N}\Phi_2(X_B^{(1,2)2})^{\otimes 2N}\rho^{\otimes 4}]$ in global random protocol analyzed before, the state $\rho$ that maximizes it is also the product pure state, in the qubit level, $\rho=\bigotimes_{i=1}^N\ketbra{\psi_i}{\psi_i}$. When $\rho$ is set to be a pure product state, the variance is easy to compute, we first do some necessary calculation:
\begin{equation}
\begin{split}
\tr\left[\Phi_2(X^{(1,2)})\rho^{\otimes 2}\right]=&2\tr[\mathbb{I}\rho^{\otimes 2}]+(2-1)\tr[\mathbb{S}\rho^{\otimes 2}]=2\times 2-1=3\\
\tr\left[\Phi_3(X^{(1,2,3)}\rho^{\otimes 3})\right]=&-\frac{3}{2+1}+\frac{3(2+1)}{2+2}=\frac{3}{2}\\
\tr\left[\Phi_4(X^{(1,2)\otimes 2})\rho^{\otimes 4}\right]=&\frac{6}{2(2+2)(2+3)}-\frac{16(2+1)}{2(2+2)(2+3)}+\frac{4(2+1)}{2(2+3)}+1+\frac{2}{2(2+2)(2+3)}=\frac{6}{5}.
\end{split}
\end{equation}
substituting these results into Eq.~\eqref{eq:M_4squareexpectation}, and change this equation into local version, we get
\begin{equation}
\begin{split}
\mathrm{Var}(\hat{M}_4^{(2,3)})\le\frac{3^{2N}}{N_M^4}+\frac{C_1 (\frac{9}{2})^N}{N_M^3}+\frac{C_2(\frac{18}{5})^{N_A}(\frac{3}{2})^{2N_B}}{N_M^2}+\frac{C_2'(\frac{3}{2})^{2N}}{N_M^2}+\frac{C_2''(\frac{3}{2})^{2N_A}(\frac{18}{5})^{N_B}}{N_M^2}+\frac{C_3(\frac{9}{5})^{N}}{N_M}+C_4(\frac{6}{5})^{2N}.
\end{split}
\end{equation}
From this equation, one can see the apparent difference between the variance of global protocol and local protocol: there is an exponential increasing term in the variance of local protocol that cannot be compressed by $N_M$. In order to suppress the variance, $N_U$ cannot be constant any more. Similarly, the total variance is 
\begin{equation}\label{eq:local total variance}
\begin{split}
\frac{1}{N_U}\left\{\frac{3^{2N}}{N_M^4}+\frac{C_1 (\frac{9}{2})^N}{N_M^3}+\frac{C_2(\frac{18}{5})^{N_A}(\frac{3}{2})^{2N_B}}{N_M^2}+\frac{C_2'(\frac{3}{2})^{2N}}{N_M^2}+\frac{C_2''(\frac{3}{2})^{2N_A}(\frac{18}{5})^{N_B}}{N_M^2}+\frac{C_3(\frac{9}{5})^{N}}{N_M}+C_4(\frac{6}{5})^{2N}\right\}.
\end{split}
\end{equation}
Hence, to make the locally constructed estimator satisfies $|\hat{M}_4^{(2,3)}-M_4^{(2,3)}|\le\epsilon$ with probability at least $1-\delta$, $N_U$ needs to scale like
\begin{equation}
\begin{split}
N_U = O(\frac{(\frac{6}{5})^{2N}}{\epsilon^2\delta}).
\end{split}
\end{equation}
Substituting it into the total variance Eq.~\eqref{eq:local total variance}, it scales like
\begin{equation}
\begin{split}
\epsilon^2\delta\left\{\frac{(\frac{5}{2})^{2N}}{N_M^4}+\frac{C_1 (\frac{25}{8})^N}{N_M^3}+\frac{C_2(\frac{5}{2})^{N_A}(\frac{5}{4})^{2N_B}}{N_M^2}+\frac{C_2'(\frac{5}{4})^{2N}}{N_M^2}+\frac{C_2''(\frac{5}{4})^{2N_A}(\frac{5}{2})^{N_B}}{N_M^2}+\frac{C_3(\frac{5}{4})^{N}}{N_M}+C_4\right\}.
\end{split}
\end{equation}
After calculation, one can find that the first term $(\frac{5}{2})^{2N}/N_M^4$ is indeed the leading term, 
\begin{equation}
\begin{split}
N_M=O((\frac{5}{2})^{N/2})
\end{split}
\end{equation}
is enough to suppress the error. Therefore, the total number of experiments scales
\begin{equation}
\begin{split}
4\times N_U\times N_M = O(\frac{(\frac{18\sqrt{2}}{5\sqrt{5}})^N}{\epsilon^2\delta})\approx O(\frac{D^{1.187}}{\epsilon^2\delta}).
\end{split}
\end{equation}

\end{proof}

\subsection{Hybrid Protocols to Measure General Permutation Moments}
As mentioned earlier, shadow protocols can measure all kinds of permutation moments, while the sample complexities are exponentially higher than the randomized-measurement protocols for the SWAP operators. For a general kind of index permutation, the observable of measuring the moments may be the tensor product of SWAP operators and other permutation operators that are not suitable for randomized measurements protocols. Take a tripartite state $\rho_{ABC}$ as an example, suppose the permutation operation is $\pi=\tbinom{1,2,3,4,5,6}{1,3,2,4,5,6}$. The fourth order of the permutation moment is
\begin{equation}
\begin{aligned}
M_4^{\pi} = \tr\left[\left(\mathbb{S}_A^{(1,2)}\otimes\mathbb{S}_A^{(3,4)}\otimes\mathbb{S}_B^{(2,3)}\otimes\mathbb{S}_B^{(4,1)}\otimes\overrightarrow{P}_C\right)\rho_{ABC}^{\otimes 4}\right],
\end{aligned}
\end{equation}
where $\overrightarrow{P}_C$ is the fourth order permutation operator on party $C$. This moment is important for the multipartite entanglement detection that we will discuss later. We can develop a shadow and randomized measurements hybrid protocol to measure this moment.

First, one needs to prepare $4\times N_U\times N_M$ copies of $\rho_{ABC}$ and divide them into $N_U$ sets, each of which has $4\times N_M$ copies. Then, one operates the unitaries $\{U_A^1\otimes U_B^1,U_A^1\otimes U_B^2,U_A^2\otimes U_B^2,U_A^2\otimes U_B^1\}\otimes\{U_C^i\}_{i=1}^{N_M}$ on the first set and measure them in the computational basis to acquire the results $\{r_{A1,B1}^i,r_C^i\}_{i=1}^{N_M}$, $\{r_{A1,B2}^i,r_C^i\}_{i=1}^{N_M}$, $\{r_{A2,B2}^i,r_C^i\}_{i=1}^{N_M}$, $\{r_{A2,B1}^i,r_C^i\}_{i=1}^{N_M}$. Using these unitaries and measurement results, we can construct the unbiased estimator of $M_4^\pi$ as
\begin{equation}
    \hat{M}_4^\pi=\frac{1}{N_M^4}\sum_{i,j,k,l=1}^{N_M}\tr\left[\hat{Q}_{AB}\left(\hat{r}_{A1,B1}^i\otimes\hat{r}_{A1,B2}^j\otimes\hat{r}_{A2,B2}^k\otimes\hat{r}_{A2,B1}^l\right)\right]\tr\left[\overrightarrow{P}_C\left(\hat{\rho}_C^i\otimes\hat{\rho}_C^j\otimes\hat{\rho}_C^k\otimes\hat{\rho}_C^l\right)\right],
\end{equation}
where $\hat{Q}$ follows the same definition in Eq.~\eqref{eq:M4estimator2} and $\hat{\rho}_C^i$ is the shadow snapshot constructed using $U_C^i$ and $r_C^i$. Then, we repeat this procedure for the remaining $(N_U-1)$ sets with i.i.d.~unitary sets and average the estimators to get the final estimator.

The unbiaseness of this estimator can be proved following the similar idea of the last two subsections. Using this estimator, we can get a reasonable sample complexity for measuring $M_4^\pi$.

\subsection{Numerical Results and Further Discussion}\label{subsec:numerical}
From the statistical analysis, one could find that the random protocol is more suitable than shadow protocol to estimate $\hat{M}_4^{(2,3)}$, especially for large-scale quantum systems. The experiment protocols of former entanglement detection works based on PPT criterion \cite{elben2020mixed,zhou2020Single} are either local shadow protocol, or global random protocol. Therefore, the method developed in this work, based on CCNR criterion that can be conducted by local random protocol, is more practical than the former works in terms of error scaling and feasibility.

So far, we have derived the upper bound of variance for shadow protocols and random protocols, while it still cannot convince us that the random protocols are definitely better than shadow protocols, for these bounds may not be tight, especially for shadow protocol. To make our error analysis more convincing, we carry out some numerical experiments. And the results are listed in Fig.~\ref{fig:numerical study}. In Fig.~\ref{fig:numerical study}(a), we investigate how the variance of predicting $M_4^{(2,3)}$ of a multiqubit W state scales with the qubit number. Comparing these four protocols, we could find that generally, random protocols are better than shadow protocols, and global protocols are better than local protocols. To further investigate the variance scaling in large dimensional scenario, we carry out a linear regression of the last five dots in each line, $\mathrm{Log}(\mathrm{Var}(\hat{M}_4))=\alpha N+\beta$, and find that the slopes are $\alpha=1.8521$ for GR protocol, $\alpha=3.9528$ for GS protocol, $\alpha=3.0457$ for LR protocol, and $\alpha=5.5509$ for LS protocol respectively which all satisfy our variance analysis. Surprisingly, one could see that the LR protocol is even better than GS protocol, this fact reveals the great advantage of random protocol over shadow protocol, which will be investigated in our further works. Fig.~\ref{fig:numerical study}(b) shows how the variance of $\hat{M}_4^{(2,3)}$ scales with number of snapshots $M$ in GS and LS protocols. These two lines show a similar trend, the slopes are large in small $M$ area and getting smaller with the increment of $M$. Such observation actually meets the variance computed in Eq.~\eqref{eq:GS variance} and Eq.~\eqref{eq:LS variance}, which tells us that the dominant term will gradually shift from $\frac{1}{M^4}$ term to $\frac{1}{M}$ term. Using Fig.~\ref{fig:numerical study}(c), we confirm the fact that the variance of $\hat{M}_4^{(2,3)}$ constructed by random protocols is inversely proportional with $N_U$.

In the statistical analysis, we find that the variance of $\hat{M}_4^{(2,3)}$ in random protocols contains some terms independent of $N_M$, which means that even $N_M$ is set to infinity, the variance will not reduce to zero, it just approaches a constant value. To reduce the variance, we need a large enough $N_U$. This result is numerically proved by Fig.~\ref{fig:numerical study}(d), in which the dashed lines correspond to the variance when $N_M$ is set to be infinity. Another important result in the statistical analysis is that when $N_M=\infty$, the variance of GR protocol is bounded by a constant independent of system dimension, while the variance of LR protocol increases with the system size. Such result is demonstrated in Fig.~\ref{fig:numerical study}(e) by the two lines corresponding to GR and LR respectively. One can see that even if $N_M$ is set to be infinity, the variance of locally constructed $\hat{M}_4^{(2,3)}$ also increases with system size, while the other does not. Fig.~\ref{fig:numerical study}(f) shows a generic numerical experiment. The measured state is a 4-qubit noisy W state, $\rho=p\ketbra{W}{W}+\frac{1-p}{2^4}\mathbb{I}$, with a changeable $p$. In this experiment, we use LR protocol, the most practical one, and set $N_U=100$ to observe the performances with different values of $N_M$.  

\begin{figure}
    \centering
    \subfigure{\includegraphics[width=5.3cm]{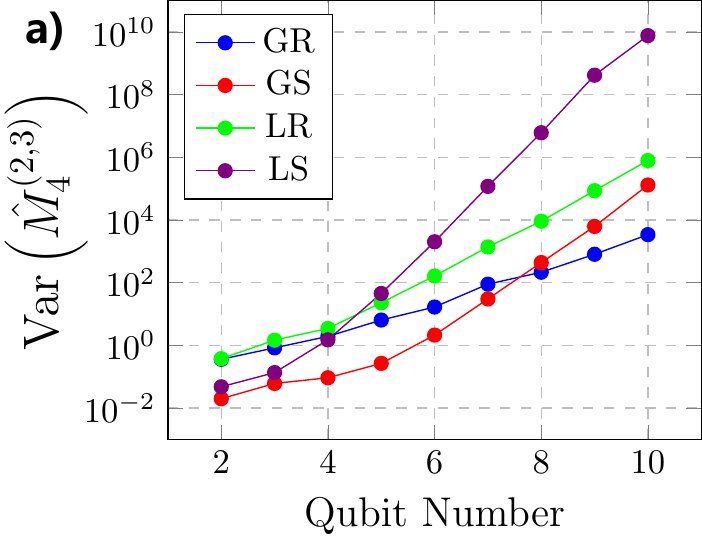}}
    \subfigure{\includegraphics[width=5.3cm]{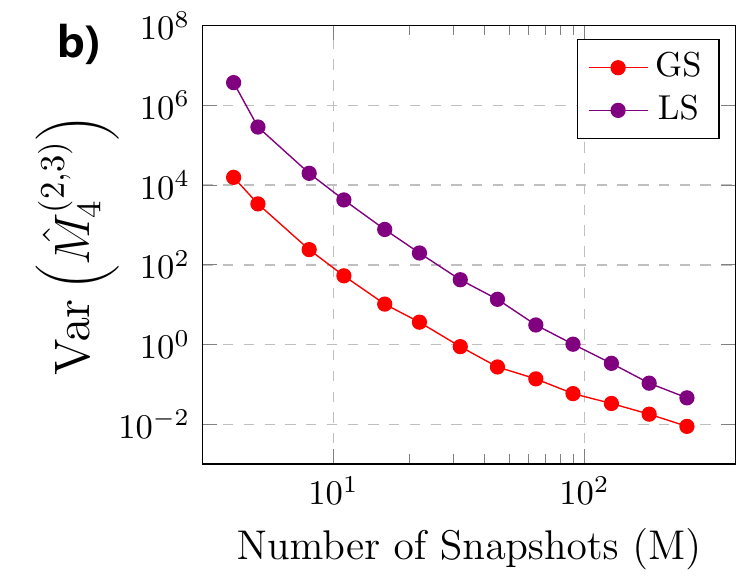}}
    \subfigure{\includegraphics[width=5.3cm]{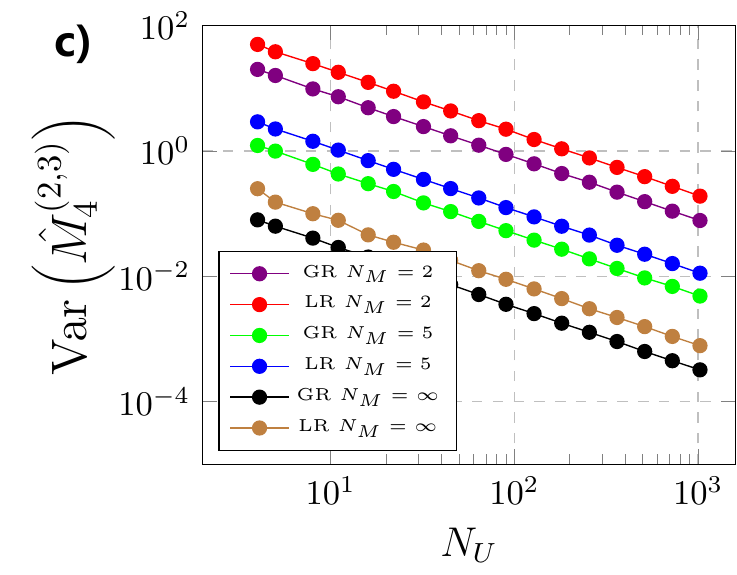}}
    
    \subfigure{\includegraphics[width=5.3cm]{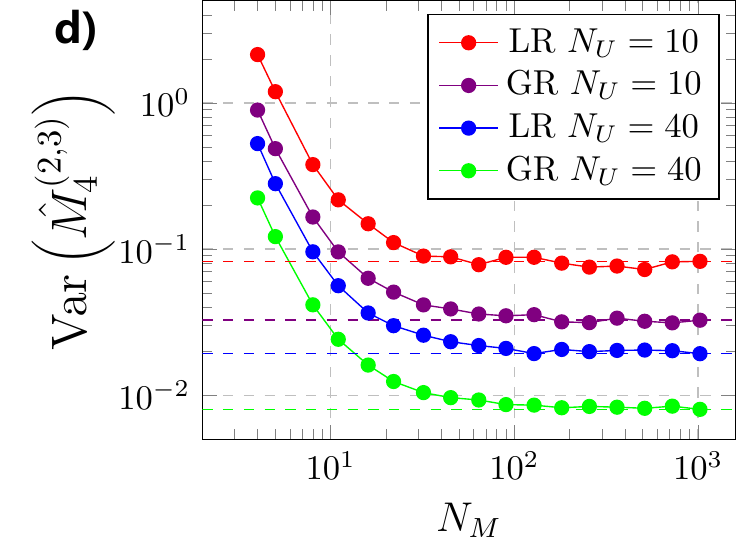}}
    \subfigure{\includegraphics[width=5.3cm]{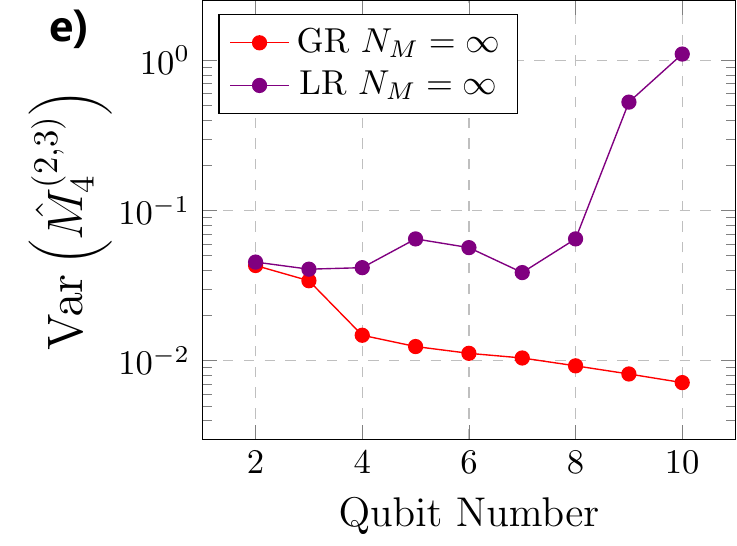}}
    \subfigure{\includegraphics[width=5.3cm]{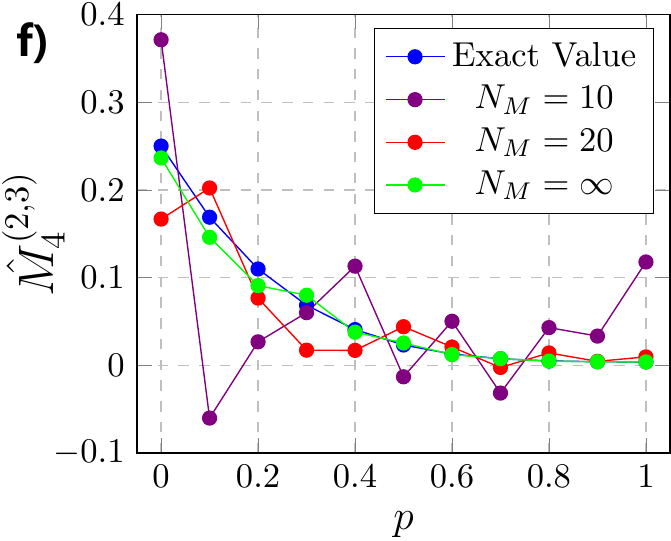}}
    \caption{We conduct some numerical experiments to test our statistical results, using all of these four protocols, global random (GR), local random (LR), global shadow (GS), local shadow (LS). In (a), we use a $N$-qubit W state to examine how the variance scales with qubit number. By definition, the value of $M_4^{(2,3)}$ depends on how we divide it into two parts. When W state contains even number qubits, it is divided into two equal number parts, and in the other case, it is divided into two parts that differ by one qubit. For GS and LS protocols, the number of snapshots, which is actually the number of experiments, is set to be $M=80$. And for GR and LR protocols, we set $N_U=4$ and $N_M=5$ to make the total times of experiment $4\times N_U\times N_M=80$ equal to that of shadow protocol. The variance data points are computed using 200 times of independent numerical experiments. (b), (c) and (d) use the same state, a four qubit W state which is equally divided into two parts, as the test state. In (b), we demonstrate how the variance scales with the number of snapshots $M$ when predicting $M_4$ using shadow protocol. The variance is calculated using 2000 times of independent experiments. (c) shows how the variance of $\hat{M}_4^{(2,3)}$ changes with $N_U$ for different $N_M$ in global random and local random protocols. Here the numbers of independent experiments to compute the variance is 20000. (d) shows how the variance changes with $N_M$ in different $N_U$ for global and local random protocols. The dashed lines denote the values of the variance of $\hat{M}_4^{(2,3)}$ when $N_M=\infty$. The number of independent experiments is also 20000. In (e), we use a $N$-qubit W state to investigate whether the variance of $\hat{M}_4^{(2,3)}$ will increase with system size, using random protocols and setting $N_M$ to be infinity. The number of independent experiments is 20000. (f) is a demonstration of a specific experiment with a 4-qubit noisy W state, $\rho = p\ketbra{W}{W}+\frac{1-p}{2^4}\mathbb{I}$. We use local random protocol to measure $M_4^{(2,3)}$ with different values of $p$. We set $N_U=100$, and set $N_M$ to be $10$, $20$ and infinity to show how the performance varies with different values of $N_M$.  }
    \label{fig:numerical study}
\end{figure}

Now, we will briefly analyze the error scaling of $\hat{M}_{2n}^{(2,3)}$ for a generic value of $n\ll d$. We will directly adopt the conclusion in analysing $\hat{M}_4^{(2,3)}$. For shadow protocol, following the analysis of $\mathrm{Var}\left(\hat{M}_4^{(2,3)}\right)$, one can prove that the leading term in $\mathrm{Var}\left(\hat{M}_{2n}^{(2,3)}\right)$ is also
\begin{equation}
\begin{split}
\frac{1}{M^{2n}}\mathrm{Var}\left\{\tr\left[\hat{O}_{2n}\left(\hat{\rho}_1\otimes\hat{\rho}_2\otimes\cdots\otimes\hat{\rho}_{2n}\right)\right]\right\},
\end{split}
\end{equation}
where 
\begin{equation}
\begin{split}
\hat{O}_{2n}=\sum_{\pi\in\mathcal{S}_{2n-1}}\left(\mathbb{I}\otimes\hat{W}_\pi\right)\left(\mathbb{S}_A^{(1,2)}\otimes\cdots\otimes\mathbb{S}_A^{(2n-1,2n)}\otimes\mathbb{S}_B^{(2,3)}\otimes\cdots\otimes\mathbb{S}_B^{(2n,1)}\right)\left(\mathbb{I}\otimes\hat{W}_\pi^\dagger\right).
\end{split}
\end{equation}
As shown in Fact \ref{fact:shadow variance}, the key value in estimating the variance of $\tr[O\hat{\rho}]$ is $\tr(O^2)$, no matter in global shadow protocol or in local shadow protocol. And following the schematic calculation in $\mathrm{Var}\left(\hat{M}_4^{(2,3)}\right)$, one can verify that $\tr(\hat{O}_{2n}^2)$ is bounded by $Const\times D^{2n}$. Therefore, the leading term of $\mathrm{Var}(\hat{M}_{2n})$ is bounded by $Const \times \frac{D^{2n}}{M^{2n}}$ for global shadow protocol, and $Const \times \frac{D^{4n}}{M^{2n}}$ for local shadow protocol. As a result, the number of $M$ will not increase with $n$, $M=O(D)$ and $M=O(D^2)$ are enough to suppress the error for global shadow and local shadow protocol respectively for arbitrary value of $n$.

However, this simple conclusion cannot be directly adopted into random protocols. First, consider global random protocol. In Eq.~\eqref{eq:M_4squareexpectation}, the one determining the scaling of $N_U$ is the last one, $\tr[(\Phi_4(X_A^{(1,2)\otimes 2})^{\otimes 2n}\otimes\Phi_4(X_B^{(1,2)\otimes 2})^{\otimes 2n})\rho^{\otimes 4n}]=O(1)$. And the one determining the scaling of $N_M$ is the first one, $\tr[(\Phi_2(X_A^{(1,2)2})^{\otimes n}\otimes\Phi_2(X_B^{(1,2)2})^{\otimes n})\rho^{2n}]^{1/2n}=O(D^{\frac{1}{2}})$. Hence, the upper bound of sample complexity has nothing to do with $n$. Consider local random protocol, and also set $\rho$ to be the pure product state, as we do before. One will find that the last term scales like $\tr[(\Phi_4(X_{AL}^{(1,2)\otimes 2})^{\otimes 2nN}\otimes\Phi_4(X_{BL}^{(1,2)\otimes 2})^{\otimes 2nN})\rho^{\otimes 4n}]=O((\frac{6}{5})^{2nN})$. Therefore, $N_U$ will exponentially increase with $n$, which makes local random protocol less feasible in large value of $n$ compared with local shadow protocol.

In conclusion, although random protocols have significant advantages in estimating $\hat{M}_4^{(2,3)}$ than shadow protocols, no matter global one or local one. When experimental conditions are limited, only local operation and measurement are implementable, which is indeed the case for state-of-the-art quantum devices, local shadow protocol will be more practical than local random protocol for estimating $\hat{M}_{2n}^{(2,3)}$ with large $n$. Besides, for a general index permutation $\pi$, where not all the parties are R-type, it is hard to design a local random protocol to measure them.  This is why we spend much effort to analyze both kinds of protocols, shadow and random.

\subsection{Efficient Estimation of Enhanced CCNR Criterion}\label{subsec:enhenced}
In \cite{zhang2008entanglement}, Zhang et,al proposed a bi-partite entanglement criterion which is strictly more enhanced than CCNR criterion and can be generalized to multipartite scenario. It has been proved that, for separable bi-partite quantum state $\rho_{AB}$, 
\begin{equation}\label{eq:enhanced CCNR}
\begin{split}
\norm{\mathcal{R}_{(2,3)}(\rho_{AB}-\rho_A\otimes \rho_B)}\le\sqrt{(1-\tr\rho_A^2)(1-\tr\rho_B^2)}.
\end{split}
\end{equation}
 The R.H.S of Eq.~\eqref{eq:enhanced CCNR} is the function of the purities of $\rho_A$ and $\rho_B$ and there already exists many methods to estimate them, like multi-copy observable or randomized measurements. Following the same idea of main text, the L.H.S of Eq.~\eqref{eq:enhanced CCNR} can also be bounded by the these higher order terms
\begin{equation}
\begin{split}
M_{2n}^{(2,3)'}=\tr[(O_A\otimes O_B)(\rho_{AB}-\rho_A\otimes\rho_B)^{\otimes 2n}],
\end{split}
\end{equation}
where $O_A$ and $O_B$ follow the same definition in Eq.~\eqref{eq:CCNR observable}. Although $(\rho_{AB}-\rho_A\otimes\rho_B)$ is not a quantum state, $(\rho_{AB}-\rho_A\otimes\rho_B)^{\otimes 2n}$ can be expended as linear combination of tensor products of $\rho_{AB}$ and $\rho_A\otimes\rho_B$, so that $M_{2n}^{(2,3)'}$ can also be estimated by randomized measurements protocol without bias. Here we aim to prove that, the second and fourth moments can also be efficiently estimated using the data acquired in Algorithm \ref{algo:CCNR measurement protocol}.

The second order moment is
\begin{equation}
\begin{split}
M_2^{(2,3)'}=\tr[(\rho_{AB}-\rho_A\otimes\rho_A)^2]=\tr(\rho_{AB}^2)+\tr(\rho_A^2)\tr(\rho_B^2)-2\tr[\rho_{AB}(\rho_A\otimes\rho_B)].
\end{split}
\end{equation}
The measurement of $M_2^{(2,3)'}$ has been discussed in \cite{liu2021characterizing}. The fourth order moment is
\begin{equation}
\begin{split}
M_4^{(2,3)'}=\tr[(O_A\otimes O_B)(\rho_{AB}-\rho_A\otimes\rho_B)^{\otimes 4}].
\end{split}
\end{equation}
After expanding $(\rho_{AB}-\rho_A\otimes\rho_B)^{\otimes 4}$, the only non-trivial terms are
\begin{equation}
\begin{split}
M_{4,1}^{(2,3)'}=\tr\left[O_A\otimes O_B\left(\rho_{AB}^{\otimes 3}\otimes\rho_A\otimes\rho_B\right)\right]
\end{split}
\end{equation}
and 
\begin{equation}
\begin{split}
M_{4,2}^{(2,3)'}=\tr\left[O_A\otimes O_B\left(\rho_{AB}^{\otimes 2}\otimes\rho_A^{\otimes 2}\otimes\rho_B^{\otimes 2}\right)\right]/\tr(\rho_A^2)=\tr\left[\left(\mathbb{S}_A^{(1,2)}\otimes\mathbb{S}_B^{(2,3)}\otimes\mathbb{S}_B^{(4,1)}\right)\rho_{AB}^{\otimes 2}\otimes \rho_B^{\otimes 2}\right].
\end{split}
\end{equation}
Other trivial terms contain the multiplication of purity terms $\tr(\rho_A^2)$ and $\tr(\rho_B^2)$, and correlation term $\tr[\rho_{AB}(\rho_A\otimes\rho_B)]$. The measurement protocols measuring them have been discussed in other works. We can prove that:
\begin{proposition}
Using the data acquired in Algorithm \ref{algo:CCNR measurement protocol}, one can construct the unbiased estimator of $M_{4,1}^{(2,3)}$ and $M_{4,2}^{(2,3)}$,
\begin{equation}
\begin{split}
\hat{M}_{4,1}^{(2,3)}=\frac{1}{N_M^4(N_M-1)}\sum_{i,j,k,l=1}^{N_M}\sum_{l'\neq l}\tr\left[\hat{Q}_{4,1}\left(\hat{r}_{A1,B1}^i\otimes\hat{r}_{A1,B2}^j\otimes\hat{r}_{A2,B2}^k\otimes\hat{r}_{A2,B1}^l\otimes\hat{r}_{A2,B1}^{l'}\right)\right],
\end{split}
\end{equation}
where $\hat{Q}_{4,1}=X_A^{(1,2)}\otimes X_A^{(3,4)}\otimes X_B^{(2,3)} \otimes X_B^{(5,1)}$, and
\begin{equation}
\begin{split}
\hat{M}_{4,2}^{(2,3)}=\frac{1}{N_M^4}\sum_{i,j,k,l=1}^{N_M}\tr\left[\hat{Q}_{4,2}\left(\hat{r}_{A1,B1}^i\otimes\hat{r}_{A1,B2}^j\otimes\hat{r}_{A2,B2}^k\otimes\hat{r}_{A2,B1}^l\right)\right]
\end{split}
\end{equation}
where $\hat{Q}_{4,2}=X_A^{(1,2)}\otimes X_B^{(2,3)}\otimes X_B^{(4,1)}$. Besides, the number of experiments proposed in Prop.~\ref{prop:random variance} is enough to estimate them accurately.
\end{proposition}

\begin{proof}

Following the proof of Proposition \ref{prop:CCNR measurement}, one can easily prove that these two estimators are indeed unbiased. Here we only prove the case of global protocol, the proof for local one is quite similar.
\begin{equation}
\begin{split}
\mathbb{E}_{U,r}\left(\hat{M}_{4,1}^{(2,3)'}\right)=&\frac{1}{N_M^4(N_M-1)}\sum_{i,j,k,l=1}^{N_M}\sum_{l'\neq l}\mathbb{E}_{U,r}\tr\left[\hat{Q}_{4,1}\left(\hat{r}_{A1,B1}^i\otimes\hat{r}_{A1,B2}^j\otimes\hat{r}_{A2,B2}^k\otimes\hat{r}_{A2,B1}^l\otimes\hat{r}_{A2,B1}^{l'}\right)\right]\\
=&\mathbb{E}_{U,r}\tr\left[\hat{Q}_{4,1}\left(\hat{r}_{A1,B1}^i\otimes\hat{r}_{A1,B2}^j\otimes\hat{r}_{A2,B2}^k\otimes\hat{r}_{A2,B1}^l\otimes\hat{r}_{A2,B1}^{l'}\right)\right]\\
=&\tr\left[\Phi_2(X_A^{(1,2)})\otimes\Phi_2(X_A^{(3,4)})\otimes\Phi_2(X_B^{(2,3)})\otimes\Phi_2(X_B^{(5,1)})\rho^{\otimes 5}\right]\\
=&\tr\left[\mathbb{S}_A^{(1,2)}\otimes\mathbb{S}_A^{(3,4)}\otimes\mathbb{S}_B^{(2,3)}\otimes\mathbb{S}_B^{(5,1)}\rho^{\otimes 5}\right]\\
=&\tr\left[\mathbb{S}_A^{(1,2)}\otimes\mathbb{S}_A^{(3,4)}\otimes\mathbb{S}_B^{(2,3)}\otimes\mathbb{S}_B^{(4,1)}\left(\rho^{\otimes 3}\otimes\rho_A\otimes\rho_B\right)\right]=M_{4,1}^{(2,3)'}.
\end{split}
\end{equation}
And similarly,
\begin{equation}
\begin{split}
\mathbb{E}_{U,r}\left(\hat{M}_{4,2}^{(2,3)'}\right)=&\tr\left[\Phi_2(X_A^{(1,2)})\otimes\Phi_2(X_B^{(2,3)})\otimes\Phi_2(X_B^{(4,1)})\rho^{\otimes 4}\right]\\
=&\tr\left[\mathbb{S}_A^{(1,2)}\otimes\mathbb{S}_B^{(2,3)}\otimes\mathbb{S}_B^{(4,1)}\rho^{\otimes 4}\right]=M_{4,2}^{(2,3)'}.
\end{split}
\end{equation}

Besides the similarity in estimation, the sample complexity is similar to $\hat{M}_4^{(2,3)}$, too. Following the spirit of statistical analysis of $\hat{M}_4^{(2,3)}$, we give a simple proof for the global random protocol here. Recalling that the leading term in variance is the one that all the indices are coincident,
\begin{equation}
\begin{split}
\mathbb{E}_{U,r}\left[\left(\hat{M}_{4,1}^{(2,3)'}\right)^2\right]\approx&\frac{1}{N_M^4(N_M-1)}\mathbb{E}_{U,r}\left\{\tr\left[\hat{Q}_{4,1}^2\left(\hat{r}_{A1,B1}\otimes\hat{r}_{A1,B2}\otimes\hat{r}_{A2,B2}\otimes\hat{r}_{A2,B1}\otimes\hat{r}_{A2,B1}'\right)\right]\right\}\\
=&\frac{1}{N_M^4(N_M-1)}\tr\left[\Phi_2(X_A^{(1,2)2})\otimes\Phi_2(X_A^{(3,4)2})\otimes\Phi_2(X_B^{(2,3)2})\otimes\Phi_2(X_B^{(5,1)2})\rho^{\otimes 5}\right]\\
\sim&\frac{D^2}{N_M^5},
\end{split}
\end{equation}
and
\begin{equation}
\begin{split}
\mathbb{E}_{U,r}\left[\left(\hat{M}_{4,2}^{(2,3)'}\right)^2\right]\approx&\frac{1}{N_M^4}\mathbb{E}_{U,r}\left\{\tr\left[\hat{Q}_{4,2}^2\left(\hat{r}_{A1,B1}\otimes\hat{r}_{A1,B2}\otimes\hat{r}_{A2,B2}\otimes\hat{r}_{A2,B1}\right)\right]\right\}\\
=&\frac{1}{N_M^4}\tr\left[\Phi_2(X_A^{(1,2)2})\otimes\Phi_2(X_B^{(2,3)2})\otimes\Phi_2(X_B^{(4,1)2})\rho^{\otimes 4}\right]\\
\sim&\frac{Dd_B}{N_M^4}.
\end{split}
\end{equation}
Therefore, $N_U$ and $N_M$ proposed in Proposition \ref{prop:random variance} are enough to suppress the error in estimating $\tr[O_A\otimes O_B(\rho_{AB}-\rho_A\otimes\rho_B)^{\otimes 4}]$.

\end{proof}

\section{Detection Capability and Physical Simulation}\label{sec:physics}
\subsection{Detection of Bound Entanglement}
There are also PPT states that are entangled. By definition, they cannot be detected by the PPT criterion. This is called bound entanglement. 
In Ref.~\cite{bennet1999unextendible}, the authors propose a systematic way to construct bound entanglement using unextendible product bases. A typical bounded state in $3\times 3$ quantum system is defined as follows:
\begin{equation}
\begin{split}
    \rho_{AB}=\frac{1}{4}\left(\mathbb{I}_{AB}-\sum_{i=1}^5\ketbra{u_i}{u_i}\right)
\end{split}
\end{equation}
where 
\begin{equation}
\begin{split}
&\ket{u_1}=\ket{0}\otimes \frac{\ket{0}-\ket{1}}{\sqrt{2}}\\
&\ket{u_2}=\ket{2}\otimes\frac{\ket{1}-\ket{2}}{\sqrt{2}}\\
&\ket{u_3}=\frac{\ket{0}-\ket{1}}{\sqrt{2}}\otimes\ket{2}\\
&\ket{u_4}=\frac{\ket{1}-\ket{2}}{\sqrt{2}}\otimes\ket{0}\\
&\ket{u_5}=\frac{\ket{0}+\ket{1}+\ket{2}}{\sqrt{3}}\otimes\frac{\ket{0}+\ket{1}+\ket{2}}{\sqrt{3}}.
\end{split}
\end{equation}
Numerical results shows that the entanglement of this state can not only be detected by the original CCNR and enhanced CCNR criterion \cite{zhang2008entanglement}, it can also be detected by the criterion proposed in this work using only second and fourth moments. Specifically speaking, we derive that
\begin{equation}
\begin{split}
    E_4^{(2,3)}(\rho_{AB}-\rho_A\otimes\rho_B)=0.693>\sqrt{(1-\tr(\rho_A^2))(1-\tr(\rho_B^2))}=0.635,
\end{split}
\end{equation}
which shows the effectiveness of the proposed entanglement criteria.

Note that, with a different approach, another criterion also shows that the detection of bound entanglement is possible in a randomized measurement scheme \cite{imai2021bound}. The technique is also used to certify certain forms of multipartite entanglement \cite{ketter2022stat}.

\subsection{Local Entanglement Decay in Thermal System}\label{subsec:localentdecay}

Quantum thermalization is a dominant dynamical phase in interacting quantum many body systems. Due to quantum thermalization, the initial localized information, like polarization, correlation function, and local entanglement, will be scrambled during the evolution and cannot be recovered by local measurements. To study this phenomenon and investigate the detection capability of the moment-based permutation criteria, we choose a 10-qubit Ising model evolved under a long range $XY$ Hamiltonian with open boundary condition, 
\begin{equation}\label{eq:XY Hamiltonian}
\begin{split}
H_{XY}=\sum_{i<j}J_{ij}(\hat{\sigma}_i^+\hat{\sigma}_j^-+\hat{\sigma}_i^-\hat{\sigma}_j^+)+ B_z\sum_i\hat{\sigma}_i^z,
\end{split}
\end{equation}
where $\hat{\sigma}_i^z$, $\hat{\sigma}_i^+$, and $\hat{\sigma}_i^-$ are the spin-$\frac{1}{2}$ Pauli-$Z$, raising, and lowering operator acting on the $i$-th qubit; $J_{ij}=\frac{J_0}{|i-j|^\alpha}$ is the interaction strength following the power-law decay with $J_0$ and $\alpha$ set to be $420s^{-1}$ and $1.24$, respectively \cite{brydges2019probing}; $B_z$ stands for transverse field and is set to be $400s^{-1}$.

We use four criteria to detect the local entanglement in this system, and summarize these four criteria as four quantities. When they are larger than 0, the entanglement is successfully detected. They have been normalized by some constants to make sure the initial values of them in the numerical simulation are same. 
\begin{enumerate}
    \item $E(\rho_{AB})=1-\frac{1}{E_4^{(2,3)}(\rho_{AB})}$. 
    
    \item $E(\rho_{AB})=1-\frac{\sqrt{(1-\tr\rho_A^2)(1-\tr\rho_B^2)}}{E_4^*(\rho_{AB})}$, where $E_4^*(\rho_{AB})=E_4^{(2,3)}(\rho_{AB}-\rho_A\otimes\rho_B)$. 
    
    \item $E(\rho_{AB})=1-\frac{\mathrm{max}\{\tr\rho_A^2,\tr\rho_B^2\}}{\tr\rho_{AB}^2}$.
    
    \item $E(\rho_{AB})=1-\frac{\tr[(\rho_{AB}^{T_A})^3]}{\mathrm{P_3}(\rho_{AB})}$ where $\mathrm{P_3}(\rho_{AB})=\beta x^3+(1-\beta x)^3$, $\beta=\lfloor\frac{1}{\tr[(\rho_{AB}^{T_A})^2]}\rfloor$, $x=\frac{\beta+\sqrt{\beta\left\{\tr[(\rho_{AB}^{T_A})^2](\beta+1)-1\right\}}}{\beta(\beta+1)}$.
\end{enumerate}

We also list more results of our numerical simulations, varying the choices of $A$ and $B$, shown in Fig.~\ref{fig:Quench Dynamics}. The initial states of $A$ and $B$ is the GHZ-type state, $\frac{1}{\sqrt{2}}(\ket{0}^{\otimes N_{AB}}+\ket{1}^{\otimes N_{AB}})$. $C$ acts as the bath, which is initialized to be the tensor product of $\ket{0}$. One could find that, no matter for what choices of $A$ and $B$, the $E_4^*$ criterion has obvious advantage over the others. 

\begin{figure}
    \centering
    \includegraphics[width=14cm]{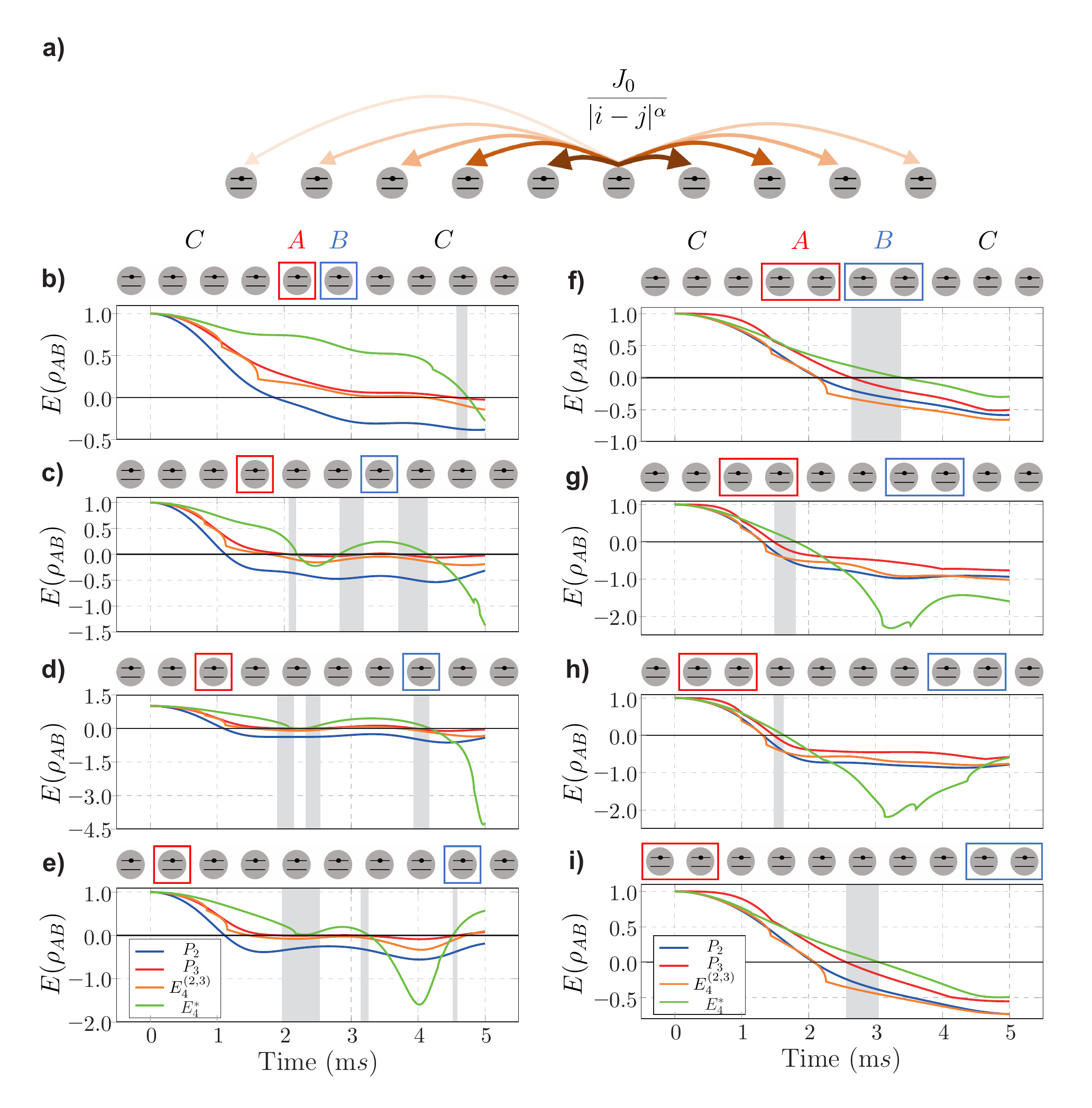}
    \caption{ \textbf{Local entanglement decaying in thermal system} a) The illustration of the 10-qubit Ising Hamiltonian,  where the grey balls represent qubits, and lines with different widths represent the power-law decay interaction. b) - i) Entanglement dynamics of local system $AB$ indicated by four criteria. The grey areas represent the time periods in which the entanglement can only be detected by the $E_4^*$ criterion.}
    \label{fig:Quench Dynamics}
\end{figure}

\subsection{Entanglement of MBL Eigenstates}
In addition to thermalization, many-body localization (MBL) \cite{abanin2019manybody} is another dynamical phase. One of the most commonly-used characteristics to distinguish between the thermal and MBL Hamiltonian is the entanglement scaling behaviour of their eigenstates. The thermal eigenstates, except for the ground state, usually have volume-law entanglement. While MBL eigenstates have area-law entanglement \cite{serbyn2013local}.  We investigate a 10-qubit disordered Ising chain with long-range interaction and an open boundary condition
\begin{equation}\label{eq:MBLhamiltonian}
\begin{split}
H_{\mathrm{Ising}}=\sum_{i<j}J_{i,j}\hat{\sigma}_i^x\hat{\sigma}_j^x+\frac{B_z}{2}\sum_i\hat{\sigma}_i^z+\sum_iD_i\hat{\sigma}_i^z,
\end{split}
\end{equation}
where $J_{i,j}=J_0/|i-j|^\alpha$ is the coupling strength following power-law decay with $\alpha=1.13$ and $B_z=4J_0$ is a uniform transverse field. The last term of this Hamitonian represents a local disordered potential, where $D_i\in[-W,W]$ is sampled from a uniform distribution where $W$ stands for the disorder strength. This Hamiltonian has been proved numerically and experimentally to have an MBL phase for large enough $W$ \cite{smith2016many,wu2016understanding}.

We apply $E_4^{(2,3)}(\cdot)$ to the eigenstate $\psi_{AB}$ of $H_{\mathrm{Ising}}$, where $A$ and $B$ are disjoint parties of these 10 qubits, see Fig.~\ref{fig:MBLeigen}. $\psi_{AB}$ is chosen to be the one with the middle eigenenergy of the spectrum.  For each value of $W$, we perform 50 numerical experiments with independent samples $\{D_i\}$ and take an average to calculate $E_4^{(2,3)}$. The results are shown in Fig.~\ref{fig:MBLeigen}(b), from which one can see the scaling transition from volume law to area law with the increment of disorder strength. For comparison, we also compute the entanglement scaling for a 10-qubit random pure state, which shows the most rapid increment of entanglement with increasing system size. These results show numerical evidences that our quantities have the potential to quantify the entanglement, not only to witness it.

\begin{figure}
    \centering
    \includegraphics[width=8cm]{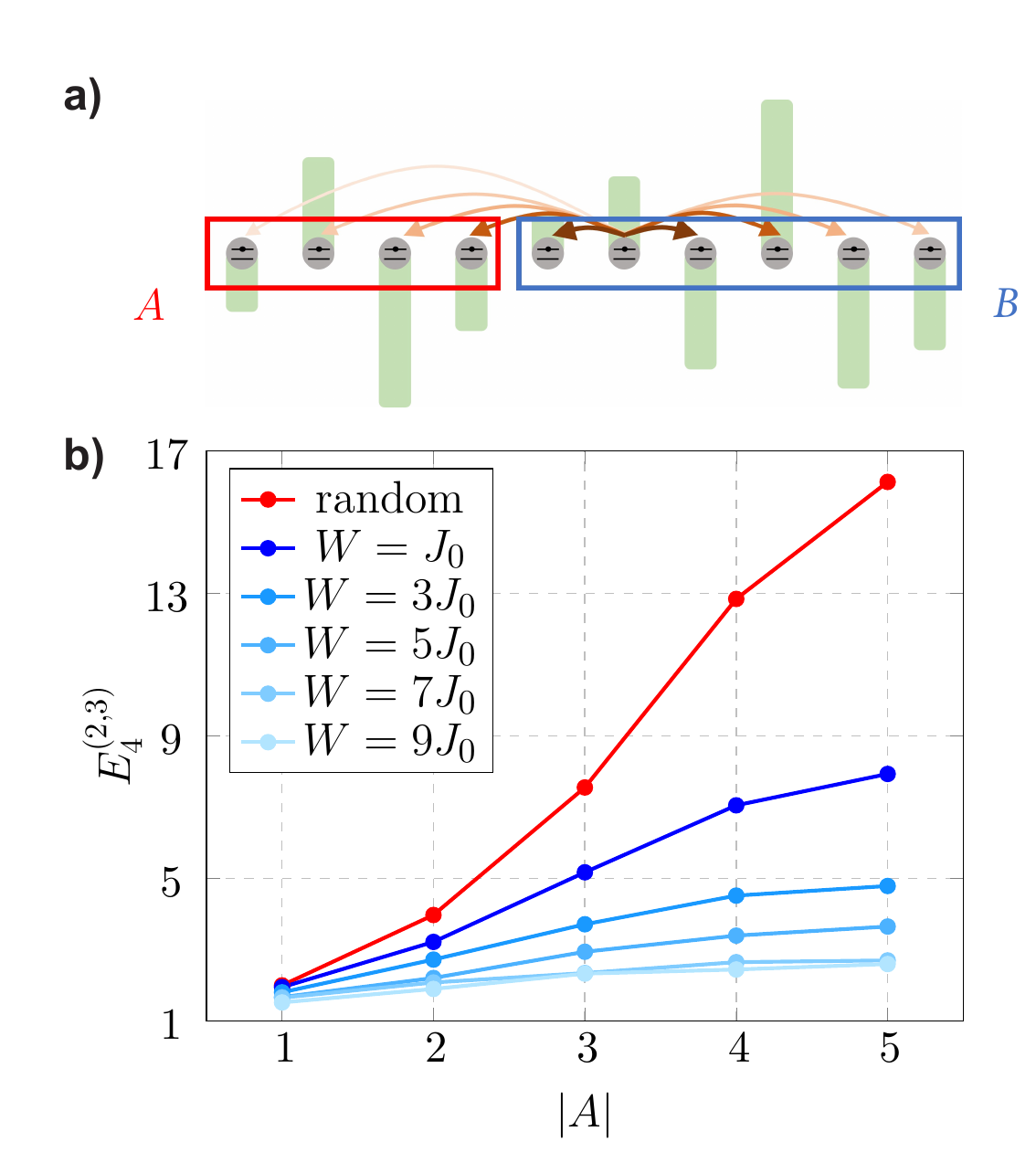}
    \caption{Entanglement of MBL eigenstates. (a) Illustration of our Hamiltonian, with arrows standing for power-law decay coupling interaction $J_{i,j}$ and solid bars standing for local disorder $\{D_i\}$. The spin chain is divided into two disjoint parties, $A$ and $B$, labeled by squares. (b) The scaling of the entanglement indicator $E_4^{(2,3)}$ of the middle eigenstate of $H_{\mathrm{Ising}}$ for different disorder strength $W$, where $|A|$ labels the qubit number of party $A$.}
    \label{fig:MBLeigen}
\end{figure}

\subsection{Entanglement of ETH Eigenstates}\label{sec:ETH}
Eigenstate Thermalization Hypothesis (ETH) is the most important hypothesis in quantum thermodynamics \cite{ueda2020quantum}. In ETH, the reduced density matrix of an eigenstate of a generic Hamiltonian, $\rho_A=\tr_{\bar{A}}(\ketbra{i}{i})$ with $H\ket{i}=E_i\ket{i}$, has the canonical form $\rho_A\propto e^{-H_A/T_i}$, where $T_i$ is the temperature given by the energy condition $\langle H\rangle_{T_i}=E_i$, and $H_A$ is the Hamiltonian that constructed by projecting $H$ into subsystem $A$. Normally speaking, the absolute values of temperature of eigenstate with the eigenenergy near the middle of the spectrum are higher than those of the edge eigenstates. Hence, the local reduced 
states of the middle eigenstates are close to the maximally mixed states and have less entanglement. When we increase or decrease the energy of the chosen eigenstate, the absolute value of temperature will fall and the local density matrices will deviate from the maximum mixed states, and the local entanglement will increase. Therefore, if we draw a diagram to investigate the relationship between the local entanglement of eigenstates and the eigenenergy, we expect to see a rainbow structure in this diagram.

Here we choose a 10-qubit 1D quantum Ising spin system with mixed fields (QIMF) to investigate this phenomenon. The Hamiltonian of this system is 
\begin{equation}
\begin{split}
H_{\mathrm{QIMF}}=\sum_{i=1}^N\left(g\hat{\sigma}_i^y+h\hat{\sigma}_i^x+J\hat{\sigma}_i^x\hat{\sigma}_{i+1}^x\right),
\end{split}
\end{equation}
where $N=10$, $\hat{\sigma}_i^x$ and $\hat{\sigma}_i^y$ are the spin-$\frac{1}{2}$ Pauli-$X$ and Pauli-$Y$ operators acting on the $i$-th qubit, respectively. Here we adopt the periodic boundary condition, $N+1=1$. This Hamiltonian has been proved to satisfy the ETH when the parameters are set to be $(g,h,J)=(0.9045,0.8090,1)$ \cite{kim2014eth,cotler2021emergent}, which are also the values we use in the numerical experiments. To observe the rainbow structure, we do the spectral decomposition of this Hamiltonian to get the eigenstates $\ket{E}$ of different eigenvalues $E$. Then, we compute $E_4^{(2,3)}(\rho_{AB})$ with $\rho_{AB}=\tr_{\overline{AB}}(\ketbra{E})$ for local parties $A$ and $B$ and different values of $E$.

The results of the numerical simulation are shown in Fig.~\ref{fig:Entanglement Eigenstates}, from which one can see the clear rainbow structure in the four diagrams. These results meet the prediction of ETH, and show another numerical evidence that the quantities we develop in this work, $E_{2n}^\pi$, have the potential to be the entanglement quantifiers.

\begin{figure}
    \centering
    \includegraphics[width=14cm]{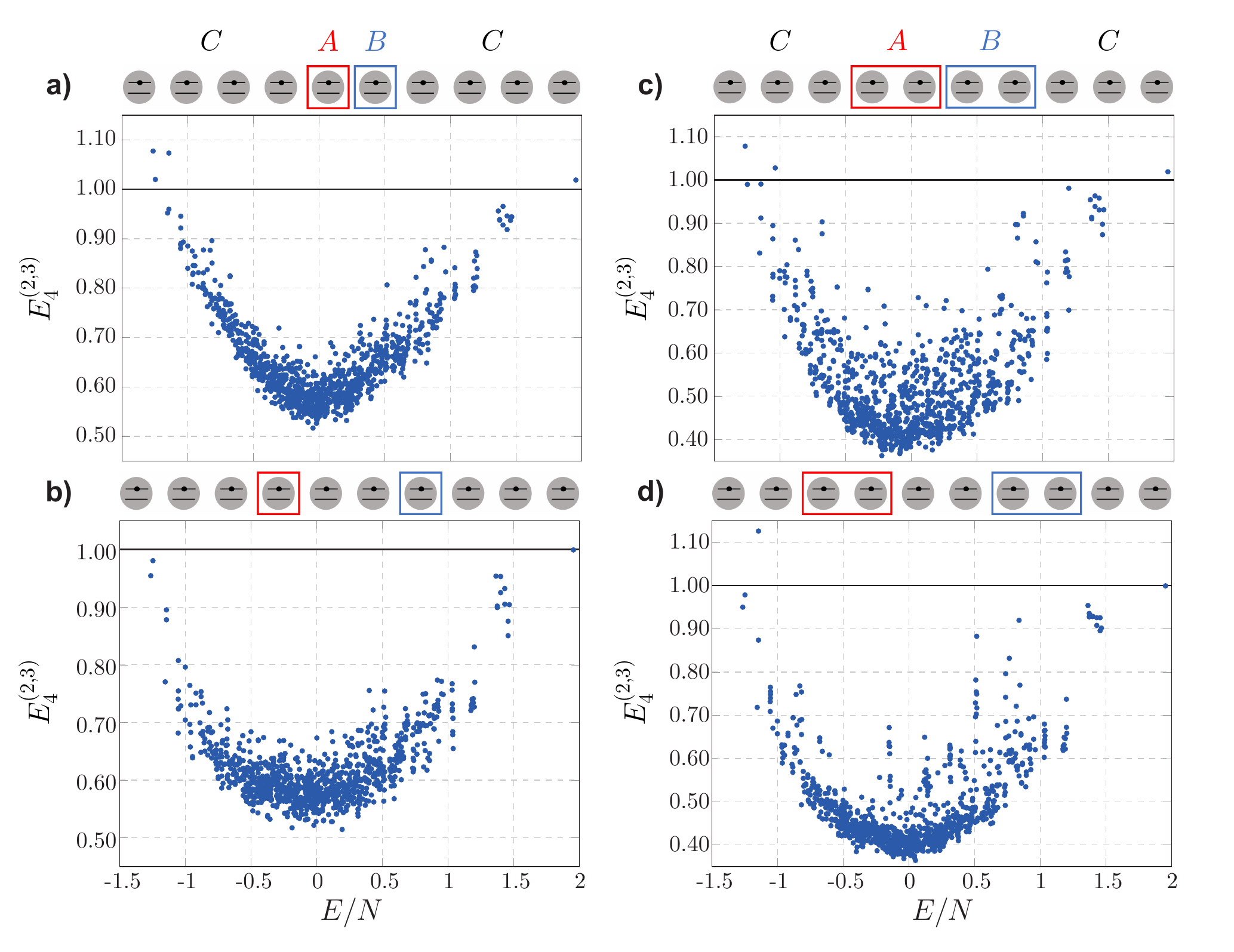}
    \caption{\textbf{Entanglement of ETH eigenstates} We investigate the relationship between the local entanglement of eigenstate of $H_{\mathrm{QIMF}}$ and the eigenenergy of eigenstates. In a) - d), the local parties are chosen to be $A$ and $B$ that marked by squares. The entanglement between $A$ and $B$ is quantified with $E_4^{(2,3)}=E_4^{(2,3)}(\rho_{AB})$ where $\rho_{AB}$ is the reduced density matrix for the eigenstate of $H_{\mathrm{QIMF}}$, $\ket{E}$. Here $E$ denotes the eigenvalue of eigenstate and $N$ denotes the qubit number of the Ising system. }
    \label{fig:Entanglement Eigenstates}
\end{figure}

\subsection{Multipartite Entanglement Detection}\label{subsec:multientdetection}
We provide details about detecting multipartite entanglement. We take the case in which moments $M_2^\pi$, $M_4^\pi$, $M_6^\pi$, and $M_8^\pi$ are given as an example. According to the Theorem~2 in the main text, if only four orders of moments are given, the minimal value of the permutation norm is reached when four different singular values exist. Denote the singular values to be $\lambda_1$, $\lambda_2$, $\lambda_3$, and $\lambda_4$ with the degeneracy of $q_1$, $q_2$, $q_3$ and $q_4$, respectively. Here, we have $q_1+q_2+q_3+q_4\le L$ where $L$ is the number of singular values. $E_8^\pi(\rho)$ can be calculated by the following optimization problem,
\begin{equation}
\begin{aligned}
\min_{q_1,q_2,q_3,q_4\in\mathbb{N}}& E_{8}^\pi(\rho)= q_1\lambda_1+q_2\lambda_2+q_3\lambda_3+q_4\lambda_4,\\
\mathrm{s.t.}& \ q_1\lambda_1^2+q_2\lambda_2^2+q_3\lambda_3^2+q_4\lambda_4^2=M_2^\pi,\\
& \ q_1\lambda_1^4+q_2\lambda_2^4+q_3\lambda_3^4+q_4\lambda_4^4=M_4^\pi,\\
& \ q_1\lambda_1^6+q_2\lambda_2^6+q_3\lambda_3^6+q_4\lambda_4^6=M_6^\pi,\\
& \ q_1\lambda_1^8+q_2\lambda_2^8+q_3\lambda_3^8+q_4\lambda_4^8=M_8^\pi,\\
& \ q_1+q_2+q_3+q_4\le L.
\end{aligned}
\end{equation}
First, we solve the equation set of the constraints for each possible value of $\{q_1,q_2,q_3,q_4\}$ to get the corresponding $\{\lambda_1,\lambda_2,\lambda_3,\lambda_4\}$. Then, we find the minimal value of $q_1\lambda_1+q_2\lambda_2+q_3\lambda_3+q_4\lambda_4$ for all the choices of $\{q_1,q_2,q_3,q_4\}$. This minimal value is thus $E_{8}^\pi(\rho)$.

The state we use in the main text is 
\begin{equation}
    \rho = \frac{1}{4}\left(\mathbb{I}_8-\sum_{i=1}^4\ketbra{\psi_i}{\psi_i}\right),
\end{equation}
where $\{\ket{\psi}_i\}_i=\{\ket{0,1,+},\ket{1,+,0},\ket{+,0,1},\ket{-,-,-}\}$. Note that this state is separable in any bipartition and hence its entanglement cannot be detected by any criterion extended from the bipartite case, such as PPT and CCNR. The permutation operation we use is $\pi=\tbinom{1,2,3,4,5,6}{1,3,2,4,5,6}$. The moments $M_{2n}^\pi(\rho)$, the optimized permutation norm $E_{2n}^\pi(\rho)$, and the singular values distribution for the minimal permutation norm are listed in Table \ref{tab:moments}. With the moments up to the eighth order, the optimized permutation norm is larger than 1, and hence the entanglement can be successfully detected.
\begin{table}[htbp]
    \centering
    \begin{tabular}{c|c|c|c}
        \hline
        $n$ & $M_{2n}^\pi$ & $E_{2n}^\pi$ & Singular Value Distribution\\
        \hline
        1 & 0.25 & 0.5 & $1\times 0.5$\\
        2 & 0.01928711 & 0.94882494 & $3\times0.282788+1\times 0.10045973$\\
        3 & 0.00175476 & 0.97296386 & $1\times 0.3353396 +2\times 0.2309162 +1\times0.1757919$\\
        4 & 0.00016548 & 1.08649082 & $2\times 0.30870218+2\times0.15042053+1\times0.08412272+1\times0.08412269$\\
        \hline
    \end{tabular}
    \caption{Moments and the optimized permutation norm.}
    \label{tab:moments}
\end{table}

\end{document}